\documentclass[11pt]{article}
\usepackage[margin=1in]{geometry}
\usepackage{amsmath,amsfonts,amssymb,amsthm}

\usepackage{graphicx}
\usepackage{tikz}
\usetikzlibrary{math}

\usepackage{enumerate}
\usepackage{bbm}

\usepackage{thmtools}
\usepackage{thm-restate}

\usepackage{hyperref,color}
\hypersetup{colorlinks=true,citecolor=blue, linkcolor=blue, urlcolor=blue}
\usepackage[linesnumbered,ruled,vlined]{algorithm2e}

\usepackage{cite}

\usepackage[T1]{fontenc}
\usepackage{libertine} 

\newtheorem{thm}{Theorem}
\numberwithin{thm}{section}

\newtheorem{defn}[thm]{Definition}
\newtheorem{lemma}[thm]{Lemma}

\newtheorem{rmk}[thm]{Remark}

\newtheorem{fact}[thm]{Fact}
\newtheorem{claim}[thm]{Claim}

\newcommand{\Exp}{\mathbb{E}}

\newcommand{\Ind}{\mathbbm{1}}

\newcommand{\norm}[1]{\left\lVert #1 \right\rVert}
\newcommand{\TV}[2]{\left\| #1 - #2 \right\|_{\normalfont\textsc{tv}}}
\newcommand{\ceil}[1]{\left\lceil #1 \right\rceil}
\newcommand{\floor}[1]{\left\lfloor #1 \right\rfloor}

\newcommand{\grad}{\nabla}
\newcommand{\hessian}{\nabla^2}

\newcommand{\diag}{\mathrm{diag}}
\newcommand{\poly}{\mathrm{poly}}
\newcommand{\BIS}{\textsc{\#BIS}}

\newcommand{\Gs}{G}

\newcommand{\NP}{\textsc{NP}}
\newcommand{\RP}{\textsc{RP}}

\newcommand{\eps}{\varepsilon}

\newcommand{\N}{\mathbb{N}}
\newcommand{\Z}{\mathbb{Z}}

\newcommand{\R}{\mathbb{R}}

\newcommand{\D}{\mathcal{D}}
\newcommand{\fpras}{\mathsf{FPRAS}}

\newcommand{\MM}{\mathrm{M}}
\newcommand{\DD}{\mathrm{D}}
\renewcommand{\SS}{\mathrm{S}}

\newcommand{\Bo}{\mathfrak{B}_o}

\newcommand{\MIsing}{\mathcal{M}_{\normalfont\textsc{Ising}}}
\newcommand{\MfIsing}{\mathcal{M}^+_{\normalfont\textsc{Ising}}}
\newcommand{\hatMaIsing}{\mathcal{\hat M}^-_{\normalfont\textsc{Ising}}}
\newcommand{\hatMfIsing}{\mathcal{\hat M}^+_{\normalfont\textsc{Ising}}}
\newcommand{\MaIsing}{\mathcal{M}^-_{\normalfont\textsc{Ising}}}
\newcommand{\MRBM}{\mathcal{M}_{\normalfont\textsc{RBM}}}
\newcommand{\MfRBM}{\mathcal{M}^+_{\normalfont\textsc{RBM}}}

\newcommand{\MPotts}{\mathcal{M}_{\normalfont\textsc{Potts}}}
\newcommand{\MfPotts}{\mathcal{M}^+_{\normalfont\textsc{Potts}}}

\newcommand{\hatMfPotts}{\mathcal{\hat M}^+_{\normalfont\textsc{Potts}}}

\newcommand{\MBiPotts}{\mathcal{M}_{\normalfont\textsc{Potts-Bip}}}

\newcommand{\MBiPottsMono}{\mathcal{M}_{\normalfont\textsc{Potts-Bip-Mono}}}
\newcommand{\MPottsMono}{\mathcal{M}_{\normalfont\textsc{Potts-Mono}}}
\newcommand{\MfPottsMono}{\mathcal{M}^+_{\normalfont\textsc{Potts-Mono}}}
\newcommand{\MfBiPottsMono}{\mathcal{M}^+_{\normalfont\textsc{Potts-Bip-Mono}}}
\newcommand{\MBiIsing}{\mathcal{M}_{\normalfont\textsc{Ising-Bip}}}

\newcommand{\din}{d_{\normalfont\textsc{in}}}
\newcommand{\dout}{d_{\normalfont\textsc{out}}}
\newcommand{\e}{e}
\newcommand{\Gr}{\mathcal{G}}

\newcommand{\Good}{\Omega_{\mathrm{good}}}
\newcommand{\BB}{\mathbb{B}}
\newcommand{\1}{\mathbbm{1}}

\title{Hardness of Identity Testing for Restricted Boltzmann Machines and Potts models}

\author{
	Antonio Blanca\thanks{Pennsylvania State University.  Email: ablanca@cse.psu.edu. Research supported in part by NSF grant CCF-1850443.}
	\and
	Zongchen Chen\thanks{Georgia Institute of Technology.  Email: \{chenzongchen,vigoda\}@gatech.edu.
		Research supported in part by NSF grants CCF-1617306 and CCF-1563838.}
	\and
	Daniel \v{S}tefankovi\v{c}\thanks{University of Rochester.
		Email: stefanko@cs.rochester.edu. Research
		supported in part by NSF grant CCF-1563757.}
	\and
	Eric Vigoda$^\dag$
}

\begin{document}

\maketitle

\begin{abstract}%

We study identity testing for restricted Boltzmann machines (RBMs), and
more generally for undirected graphical models.
Given sample access to the Gibbs distribution corresponding to an unknown or hidden
model $M^*$  and given an explicit model $M$,
can we distinguish if either $M = M^*$ or
if they are (statistically) far apart?
Daskalakis et al.~(2018) presented a polynomial-time algorithm for identity testing for
the ferromagnetic (attractive) Ising model.
In contrast,
for the antiferromagnetic (repulsive) Ising model, Bez\'akov\'a et al.~(2019) proved that unless $\RP=\NP$
there is no identity testing algorithm when $\beta d=\omega(\log{n})$,
where $d$ is the maximum degree of the visible graph and $\beta$ is the largest
edge weight (in absolute value).

We prove analogous hardness results for RBMs (i.e., mixed Ising models on bipartite graphs),
even when there are no latent variables or an external field. Specifically, we show that if $\RP\neq \NP$, then when $\beta d=\omega(\log{n})$
there is no polynomial-time algorithm for identity testing for RBMs; when $\beta d =O(\log{n})$ there is an efficient
identity testing algorithm that utilizes the structure learning algorithm of Klivans and
Meka (2017).
In addition, we prove similar lower bounds for purely ferromagnetic RBMs with inconsistent
external fields, and for the ferromagnetic Potts model.
Previous hardness results for identity testing of Bez\'akov\'a et al. (2019) utilized the hardness
of finding the maximum cuts, which corresponds to the ground states of the antiferromagnetic Ising model.
Since RBMs are on bipartite graphs such an approach is not feasible.
We instead introduce a novel methodology to reduce from the corresponding approximate counting problem and utilize the
phase transition that is exhibited by RBMs and the mean-field Potts model.
We believe that our method is general, and that it can be used
to establish the hardness of identity testing for other spin systems.
\end{abstract}

%\smallskip
%\begin{keywords}%
%  Identity Testing, Phase Transitions, Partition Functions, Gibbs Distribution
%\end{keywords}

\section{Introduction}\label{sec:intro}
For graphical models, there are several fundamental computational tasks
which are essential for utilizing these models.
These computational problems can be broadly labeled as follows:
sampling, counting, structure learning, and testing.
Our big picture aim is to understand the relationship between these problems.
The specific focus in this paper is on the computational complexity of the {\em identity testing} problem
for \emph{undirected} graphical models and its connections to the hardness of the counting problem.

Identity testing is a basic question in statistics for testing whether a given model fits a dataset.
Roughly speaking, given data $\D$ sampled from the posterior or likelihood distribution
of an unknown/hidden model $M^*$ and given an explicit model $M$, can we distinguish
whether $M=M^*$?
%or the corresponding distributions $\mu_M$ and $\mu_H$ are
%far from each other?

We study identity testing in the context of undirected graphical
models~\cite{Murphy}, which correspond to (pairwise) Markov random fields in
probability theory and computer vision~\cite{GG}
and to spin systems in statistical physics~\cite{Georgii}.
We focus attention on examples of graphical models of particular interest:
the Ising model, the Potts model, and Restricted Boltzmann Machines.
The Ising model is the simplest example of an undirected graphical model, and, in fact,
it is one of the most well-studied models in statistical physics where it is used to
study phase transitions.  The Potts model is the generalization of the
Ising model from a two state system to an integer $q\geq 3$ state system.
It is also well-studied in statistical physics as the nature of the phase transition changes
as $q$ increases~\cite{DCGHMT,DCST}.

Restricted Boltzmann Machines (RBMs) are a simple class of undirected
graphical models corresponding to the Ising model on bipartite graphs.
Originally introduced by Smolensky in 1986~\cite{Smolensky},
they have played an important role in the history of computational learning theory.
They have two layers of variables: one layer
corresponding to the observed variables and another layer corresponding to
the hidden/latent variables, and
no intralayer connections so that the underlying graph is bipartite.
Learning was shown to be practical in these
restricted models~\cite{Hinton,HOT}
and henceforth played a seminal role in the development of
deep learning~\cite{SalakhutdinovHinton,OsinderoHinton,SMH,topic}.

We define first the Potts model, as both the Ising model and RBMs may be viewed as special cases of this model.
The Potts model is specified by a graph $G=(V,E)$, a set of vertex labels or spins $[q] = \{1,\dots,q\}$,
a set of edge weights defined by $\beta: E \rightarrow \R$
and a set of vertex weights $h: V \times [q] \rightarrow \R $.
%$\beta=(\beta_e)_{e\in E}$ where $\beta_e\in\R$.
%Let $V=L\cup R$ denote the bipartition of the
%vertices, and let $|L|=|R|=n$.
%Note, the weights $\beta_e$ are allowed to be a function of $n$.
Configurations of the Potts model are the collection of vertex labelings
$\Omega=\{1,\dots,q\}^V$.
The \emph{Gibbs distribution} associated with the Potts model is a distribution over all configurations $\sigma\in\Omega$ such that:
%\vspace*{-3mm}
\[
\mu(\sigma) = \mu_{G,\beta,h}(\sigma) := \frac{1}{Z} \exp\left( \sum_{\{u,v\}\in E} \beta(\{u,v\})\1(\sigma(u)=\sigma(v)) + \sum_{v \in V} h(v,\sigma(v)) \right),
\]
where $Z=Z_{G,\beta,h}$ is the normalizing factor or \emph{partition function} 
given by:
\[
Z := \sum_{\sigma\in\Omega}  \exp\left( \sum_{\{u,v\}\in E} \beta(\{u,v\})\1(\sigma(u)=\sigma(v)) + \sum_{v \in V} h(v,\sigma(v)) \right).
\]

When $\beta(e) > 0$ for every $e \in E$, the model is called \emph{ferromagnetic} and neighboring vertices prefer to align to the same spin.
Conversely, when $\beta(e) <0$ for every $e \in E$ the model is called \emph{antiferromagnetic}.
Models where $\beta$ is allowed to be both positive or negative for distinct edges are called \emph{mixed} models.

The Ising model corresponds to the special case where there are only two spins; i.e., $q=2$.
RBMs are mixed Ising models restricted to bipartite graphs; that is, $G$ is bipartite with bipartition $V=L\cup R$.
Since the focus in
this paper is on lower bounds, we often consider the case  of no external field ($h=0$) in order to obtain stronger
hardness results.
%RBMs are the special case of the Ising model restricted to bipartite graphs.  The Potts
%model is the generalization of the Ising model to spins $\{1,2,\dots,q\}$ where $q\geq 2$ and
%weighted by the number of monochromatic edges (those where both endpoints receive the same spin),
%the case $q=2$ is the same as the Ising model.  The Potts and Ising model are more formally defined in Section~\ref{sec:definition}.

Given a model specification, that is, a graph $G=(V,E)$, an edge weight function $\beta$ and an external field $h$,
the goal in the {\em sampling problem} is to
generate samples from the Gibbs distribution $\mu=\mu_{G,\beta,h}$ (or from a distribution close to $\mu$
in total variation distance).
The corresponding {\em counting problem} is to compute the partition function $Z = Z_{G,\beta,h}$.
The (exact) counting problem is \#P-hard~\cite{Valiant} even for restricted classes of graphs~\cite{Greenhill,Vadhan}, and hence the
focus on the approximate counting problem of
obtaining an $\fpras$ (fully-polynomial randomized approximation scheme\footnote{A fully polynomial-time randomized approximation scheme ($\fpras$) for an optimization problem with optimal solution $Z$ produces an approximate solution $\hat{Z}$ such that, with probability at least $1-\delta$, $(1-\eps)\hat{Z} \leq Z \leq (1+\eps)\hat{Z}$ with running time polynomial in the instance size, $\eps^{-1}$ and $\log (\delta^{-1})$.}) for $Z$.
For a general class of models, the approximate counting and the approximate sampling problems
are equivalent, i.e., there are polynomial-time reductions between them~\cite{JVV,SVV,Kolmogorov}.
A seminal result of Jerrum and Sinclair~\cite{JS:ising} (see also~\cite{RandallWilson,worm,GuoJerrum}) presented an $\fpras$ for the partition function of the
ferromagnetic Ising model. %(with no or non-negative external field).

Another two fundamental problems for undirected graphical models are \emph{structure learning}
and \emph{identity testing}.  The structure learning problem is as follows: given oracle
access to samples from the Gibbs distribution $\mu_{M^*}$ for an unknown (i.e., ``hidden'') model $M^* = (G^*,\beta^*,h^*)$,
can we learn $G^*$ (i.e., the structure of the model) in polynomial-time with probability at least $2/3$?
In the case of no latent variables (so the samples from the Gibbs distribution reveal the
label of all vertices $V$ of $G$) recent work of Klivans and Meka~\cite{KM} (see also~\cite{Bresler,VML,HKM,VMLC,WSD}) learns
$n$-vertex graphs
with $O(\log{n})\times\exp(O(\beta d))$ samples and $O(n^2\log{n})\times\exp(O(\beta d))$ time
where $d$ is the maximum degree of $G$ and $\beta := \max_{e \in E} |\beta(e)|$ is the maximum edge weight in absolute value;
this bound has nearly-optimal sample complexity from an information-theory perspective~\cite{SanthanamWainwright}.

%The problem of learning G∗ from samples is known as structure learning and has received tremendous attention (see, e.g., Chow and Liu, 1968; Dasgupta, 1999; Lee et al., 2007; Anandkumar et al., 2012; Ravikumar et al., 2010; Bresler et al., 2013, 2014b; Bresler, 2015; Vuffray et al., 2016; Hamilton et al., 2017; Klivans and Meka, 2017). Once the graph G∗ is known, it is often a simpler task to estimate H∗ (Bresler, 2015); this is known as the parameter estimation problem.

For RBMs with latent variables (thus samples only reveal the labels for vertices on one side $R$),
structure learning can be done in time $O(n^{d_L+1})$ where $d_L$ is the maximum degree of the latent variables.
Recent work of Bresler, Koehler and Moitra~\cite{BKM} proves that there is no algorithm with running time
$n^{o(d_L)}$ assuming $k$-sparse
noisy parity on $n$ bits is hard to learn in time $n^{o(k)}$; they also show that for the special case of ferromagnetic RBMs with hidden variables there is a structure learning algorithm with $O(\log{n})\times\exp(O(\beta d^2))$ sample complexity and $O(n^2\log{n})\times\exp(O(\beta d^2))$ running time.

In the {\em identity testing} problem we are given oracle access to samples from the Gibbs distribution~$\mu_{M^*}$
for an unknown model $M^* = (G^*,\beta^*,h^*)$ (as in structure learning) and we are also given
an explicit model $M = (G,\beta,h)$. Our goal is to determine, with probability $\geq 2/3$, if either $M = M^*$ or
if the models are $(1-\eps)$-far apart; specifically, if the total variation distance between their Gibbs distributions
is at least $1-\eps$ for a given $\eps>0$.  (We note that previous works assumed separation $\geq\eps$ in the later case, whereas
we prove hardness even when we assume separation $\geq 1-\eps$.)

%In fact we prove hardness results even when assuming that
%either the models are the same or their associated Gibbs distributions are extremely far apart, namely
%total variation distance $>1-1/\poly(n)$, see Section~\ref{sec:definitions} for formal definitions.

It is known that identity testing cannot be solved in polynomial time for general graphical models in the presence of hidden variables
unless $\RP= \NP$~\cite{BMV} .
In this paper we assume there are no hidden variables and hence the samples from $\mu_{M^*}$
reveal the label of every vertex in the graph $G$; this setting is more interesting for hardness results.
We explore a more refined picture of hardness of identity testing vs. polynomial-time algorithms.

It is known that identity testing can be reduced to sampling~\cite{DDK} or structure learning~\cite{colt}:
given an efficient algorithm for the associated sampling problem or an efficient algorithm for structure learning, then one
can efficiently solve the identity testing problem.  Hence, identity testing is (computationally)
easier than sampling and structure learning.
(To be precise, one needs to solve both the structure learning and the parameter estimation problems to solve identity testing; the algorithm of Klivans and Meka~\cite{KM} does in fact provide this.)
%\eric{Good enough?}
This raises the question of whether identity testing can be efficiently solved in cases where sampling and
structure learning are known to be hard.  We prove (for the models studied here)
that when sampling and structure learning are hard, then identity testing is also hard.

\subsection{Our results}
%It is straightforward to see that the identity testing problem can be reduced to structure learning:
%we first run the structure learning algorithm  to learn the
%unknown model $G,\beta$, and then we simply compare $(G,\beta)$ and $(\Ghat,\betahat)$.  This
%implies a poly-time testing algorithm when $d|\beta| = O(\log{n})$ (and there are no latent variables),
%due to the work of~\cite{KM}.

%\medskip\noindent\textbf{Our results.} \ \
The $\eps$-identity testing problem for the Ising and Potts models is formally defined as follows.
For positive integers $n$ and $d$, and positive real numbers $\beta$ and $h$,
let $\MRBM(n,d,\beta,h)$ denote the family of RBMs
on $n$-vertex bipartite graphs $G=(V,E)$ of maximum degree at most $d$,
where
the absolute value of all edge interactions is at most $\beta$
and the field $|h(v,i)| \leq h$ for all $v \in V$ and $i \in [q]$; see Definition~\ref{dfn:potts:notation}.
%$h: V_G \times [q] \rightarrow \R$ satisfies $|h(v,i)| \le h$ for all $v \in V$ and $i \in [q]$.
We define $\MPotts(n,d,\beta,h)$ analogously
for the family of Potts models, without the restriction of $G$ being bipartite.
%\eric{Is it better to define it for RBMs since that's the focus and then say it generalizes to Ising/Potts in the same manner.}
%for $\mathcal{M}_{\textsc{bip}}^{\textsc{ferro}}(n,d,\beta,h)$ that all edge weights are positive.

\begin{center}
	\fbox{
		\parbox{0.88\textwidth}{
			Given an RBM $M \in \MRBM(n,d,\beta,h)$, and
			sample access to a distribution $\mu_{M^*}$ for an
			unknown RBM $M^* \in \MRBM(n,d,\beta,h)$,
			distinguish with probability at least $3/4$ between the cases:
			%\begin{enumerate}
			\begin{center}
			1.~$\mu_{M}=\mu_{M^*}$;~~~~~~~~~~~~2.~${\|\mu_{M}-\mu_{M^*}\|}_\textsc{tv} \ge 1 - \eps$.
			\end{center}
			%\end{enumerate}
}}
\end{center}
\noindent
The choice of $3/4$ for
the probability of success is arbitrary, and it can be replaced by any constant in the interval $(\frac 12,1)$ at the expense of a constant factor in the running time of the algorithm.
The $\eps$-identity testing problem for the Potts model is defined in the same manner, but assuming that both $M$ and $M^*$ belong to $\mathcal M_{\textsc{potts}}(n,d,\beta,h)$ instead.

%Our first result shows that for RBMs the above approach utilizing structure learning
%is essentially best possible.  In particular we prove that when $d|\beta| = \omega(\log{n})$ then there
%is no poly-time identity testing algorithm, unless $NP=RP$.
%
%
%\begin{thm}
%	For restricted Boltzmann machines, unless $NP=RP$,
%	there is no polynomial-time identity testing algorithm which applies for all models $(G,\beta)$
%	and $(\Ghat,\betahat)$ when $d|\beta| = \omega(\log{n})$ where $d$ is the maximum degree of $\Ghat$ and
%	$|\beta|$ is the maximum (in absolute value) of the edge weights $\betahat_e$.
%\end{thm}

Our first result concerns the identity testing problem on $\MRBM(n,d,\beta,0)$;
that is, (mixed) RBMs without external fields: $h(v,i)=0$ for all $v\in V$, $i \in [q]$.
We show that for RBMs the approach utilizing structure learning
is essentially best possible.  In particular we prove that when $\beta d = \omega(\log{n})$ there
is no poly-time identity testing algorithm, unless $\RP=\NP$. Note that when $\beta d = O(\log{n})$, the
algorithm of Klivans and Meka \cite{KM} for structure learning and parameter estimation
provides an identity testing algorithm with $\poly(n)$ sample complexity and running time.

\begin{thm}\label{thm:main-RBM-mixed}
	%	Suppose that $d\ge 3$ and $\eps \in (0,1)$.
	%	If $\beta d = \omega(\log n)$,
	%	then there is no polynomial-time identity testing algorithm for $\mathcal M_{\textsc{bip}}(n,d,\beta,0)$
	%	unless RP=NP.
	Suppose $n$, $d$ are positive integers such that $3 \le d \le n^\theta$ for constant $\theta \in (0, 1)$ and let $\eps \in (0,1)$. If $\RP \neq \NP$, then for all real $\beta > 0$ satisfying $\beta d = \omega(\log n)$  there is no polynomial running time algorithm to solve the $\eps$-identity testing problem for the class $\MRBM(n,d,\beta,0)$ of mixed RBMs without external fields.
\end{thm}

In contrast to the above result,
Daskalakis, Dikkala and Kamath~\cite{DDK} provided a poly-time identity testing algorithm for all \emph{ferromagnetic} Ising model
with \emph{consistent} fields
(the external field is consistent if it only favors the same unique spin at every vertex; otherwise it is called inconsistent; see Definition~\ref{dfn:consistent}).
Their algorithm crucially
utilizes the known poly-time sampling methods for the ferromagnetic Ising model~\cite{JS:ising,RandallWilson,worm,GuoJerrum}.
On the hardness side, super-polynomial lower bounds were recently established for identity testing for the \emph{antiferromagnetic} Ising model on general (not necessarily bipartite) graphs
when $\beta d=\omega(\log{n})$~\cite{colt}.
% this result holds for general graphs.
% and hence
%is an easier result).
This previous result utilizes the hardness of the maximum cut problem, since maximum cuts  correspond to the
``ground states'' (maximum likelihood configurations) of the antiferromagnetic model; this is not the case for RBMs, and new insights are required 
(see Section~\ref{subsect:intro:tech} for a more detailed discussion).
In particular we show a new approach to reduce from the counting problem.
%;see Section~\ref{subsec:proof-overview}.

Ferromagnetic and antiferromagnetic RBMs are equivalent models; that is, there is a
one-to-one correspondence between configurations with the same weight.
Hence, the results in~\cite{DDK} solve the identity testing problem for both ferromagnetic and antiferromagnetic RBMs with no latent variables,
even in the presence of a consistent external field.
Moreover,
Klivans and Meka's algorithm from~\cite{KM} together with the hardness results of
Theorem~\ref{thm:main-RBM-mixed}
provides a fairly complete picture of
the computational complexity of identity testing for (mixed) RBMs with no external field ($h=0$).

Our next result concerns the hardness of identity testing for purely \emph{ferromagnetic} RBMs with an \emph{inconsistent} magnetic field; that is, a field that favors one spin for some of the vertices and the other spin for the rest; see Definition~\ref{dfn:consistent}.
For this we utilize the complexity of {\BIS}, which is the problem of counting the independent sets in a bipartite graph. {\BIS} is believed not to have an $\fpras$, and it has achieved considerable interest in approximate counting as a tool for proving relative complexity hardness~\cite{DGGJ,GJ,DGJ,BDGJM,CDGJLMR,CGGGJSV,GGJ}.
Let $\MfRBM(n,d,\beta,h)$ be set of all ferromagnetic RBMs in
$\MRBM(n,d,\beta,h)$.

\begin{thm}\label{thm:main-RBM-ferro}
	Suppose $n$, $d$ are positive integers such that $3 \le d \le n^\theta$ for constant $\theta \in (0, 1)$ and let $\eps \in (0,1)$.
	If {\BIS} does not admit an $\fpras$, there exists $h = O(1)$ such that when $\beta d= \omega(\log n)$
	there is no polynomial running time algorithm that solves the $\eps$-identity testing problem for
	the class $\MfRBM(n,d,\beta,h)$ of \emph{ferrromagnetic} RBMs with \emph{inconsistent} external fields.
\end{thm}

Given the efficient identity testing algorithm for ferromagnetic Ising models~\cite{DDK,JS:ising},
we may ask whether there
are other (ferromagnetic) models that allow efficient testing algorithms.
A prime candidate is the ferromagnetic Potts model.
Both the ferromagnetic Ising and Potts models have
a rich structure; for instance, their random-cluster representation~\cite{Grimmett}
enables sophisticated (and widely-used)
sampling algorithms such as the Swendsen-Wang algorithm~\cite{SW}.
However, while there are efficient samplers for the ferromagnetic Ising model for all graphs $G$ and all edge interactions $\beta$~\cite{JS:ising,worm,GuoJerrum}, the case of the ferromagnetic
Potts model (i.e., $q>2$ spins) looks less promising. In fact, it is unlikely that there is an efficient sampling/counting algorithm for general
ferromagnetic Potts models
since this is a known {\BIS}-hard problem~\cite{GJ,GSVY}; this is due
to a phenomena called \emph{phase co-existence}, which we will  also exploit;
see Section~\ref{subsubsec:mf}.
Given the weaker hardness of sampling and approximate counting for the ferromagnetic Potts model, the hardness of the identity problem was less clear.

We prove that identity testing for the ferromagnetic Potts model is in fact hard
in the same regime of parameters where sampling and structure learning are known to be hard.
Specifically, we observe that the structure learning algorithm from~\cite{KM} applies to the Potts model, and hence implies
a testing algorithm when $\beta d = O(\log{n})$; we establish lower bounds when $\beta d = \omega(\log{n})$
that hold even for the simpler case of models with no external field.

\begin{thm}
	\label{thm:Potts}
	Suppose $n$, $d$, $q \ge 3$ are positive integers such that $3 \le d \le n^\theta$ for constant $\theta \in (0, 1)$ and let $\eps \in (0,1)$.
	If {\BIS} does not admit an $\fpras$, then there is no polynomial running time algorithm that solves the $\eps$-identity testing problem for the class $\MfPotts(n,d,\beta,0)$ of ferromagnetic $q$-state Potts models without an external field. Moreover,
	our lower bound applies restricted to the class of ferromagnetic Potts models on bipartite graphs in $\MfPotts(n,d,\beta,0)$.
\end{thm}

%
%For example, we prove that for the ferromagnetic Ising model once a (mixed) external field is present then
%the identity testing problem becomes hard.  This result also applies for bipartite graphs and hence
%holds for RBMs with all ferromagnetic edge weights.
%
%\begin{thm}\label{thm:main-RBM-mixed}
%	Suppose that $d\ge 3$.
%	If $\bmax d = \omega(\log n)$,
%	then there is no polynomial-time identity testing algorithm for all $(\bmax,d)$-bounded RBMs,
%	unless RP=NP.
%\end{thm}

\subsection{Our techniques}
\label{subsect:intro:tech}

%\noindent\textbf{Our techniques.} \ \
Our proof is a general approach that allows us to obtain hardness results for several models of interest. 
Specifically, we devise a novel methodology to reduce the problem of approximate counting (i.e., approximating partition functions) to identity testing.
For this we consider a decision version of approximate counting and 
prove that this variant is as hard as the standard approximation problem;
this first step of our reduction applies to many other models of interest (see Theorem~\ref{thn:decision-cnt:potts} and Section~\ref{app:app-cnt}).

In the second step of our reduction, given a hard counting instance, 
we 
use insights about the phase transition of the models
to construct a testing instance whose output allows us to solve the decision version of approximate counting.
The actual reduction is generic (see Theorem~\ref{thm:gen-red}), but the insights about each model are needed to build a suitable testing instance; this construction is the only part of our proof that is model specific, whereas every other step in the proof applies to more general spin systems.
Our approach is nicely illustrated in the context of the ferromagnetic 
Potts model; that is, in the proof of Theorem~\ref{thm:Potts} in Section~\ref{sec:Potts}.
There, we utilize the phase transition phenomenon in the associated mean-field Potts model which corresponds to the complete graph. In particular, there is a phase co-existence corresponding to a
first-order phase transition which we utilize to approximate the partition function 
of the input graph; see Section~\ref{sec:Potts}. 

In the third and final step of the reduction, we reduce the
maximum degree of the graph in the testing instance by using random bipartite graphs as gadgets, as has been done in seminal hardness results for approximate counting~\cite{Sly,SlySun}, and more recently in~\cite{colt} for the hardness of testing for the antiferromagnetic Ising model.
This step is also generic and applies to a large class of models; see Section~\ref{section:main-proof} and specifically Theorem~\ref{thm:deg-red}.
One interesting implication of our approach is that our gadget and reduction yields always bipartite graphs, and
hence we immediately get hardness results for bipartite graphs for all of the models studied in this paper.

We pause to briefly contrast the above proof approach with that in~\cite{colt}, where it was established hardness of identity testing for the antiferromagnetic Ising model. As mentioned earlier, in the antiferromagnetic Ising model, the configurations with the highest weight or likelihood (i.e., the ground states) correspond to the maximum cuts of the original graph. Hence, it is natural to prove hardness of identity testing for the antiferromagnetic Ising model using a reduction from the maximum cut problem. The ground states of ferromagnetic systems, on the other hand, correspond to the monochromatic configurations, so there is no hard optimization problem in the background to utilize in the reduction. 
(The similar obstacle for RBMs is that the maximum cut problem is trivial in bipartite graphs, so we cannot hope to use it to prove hardness.)
We  use the hardness of approximating the partition function instead, and consequently our
reduction is of a completely different flavor (than~\cite{colt}); we utilize 
the unique nature of the phase transition in these models in an essential way. 

To reduce the degree of the graphs in our construction we do utilize insights and certain technical lemmas from~\cite{colt}. Specifically, those concerning the expansion of random near-regular bipartite graphs.
We note that the models we consider on these random graphs are different than those in~\cite{colt}; in particular, we consider mixed models and allowed external fields, whereas in ~\cite{colt} these gadgets are purely antiferromagnetic and there is no external field.

We present our proof approach in the context of the ferromagnetic Potts model first, specifically in Section~\ref{sec:Potts}
we prove Theorem~\ref{thm:Potts}.  The proofs for RBMs, namely Theorems~\ref{thm:main-RBM-mixed}
and~\ref{thm:main-RBM-ferro} which follow the same approach, are provided in Sections~\ref{RBM-nofield} and \ref{RBM-field}, respectively.

\section{Testing ferromagnetic Potts models}
\label{sec:Potts}

In this section we prove Theorem~\ref{thm:Potts}, our lower bound for identity testing for the ferromagnetic Potts model.
To prove this theorem, we introduce a new methodology
to reduce approximate counting (i.e., the problem of finding an $\fpras$ for the partition function of a model), to identity testing.
We later use this framework to establish our lower bounds for
identity testing for RBMs (i.e., Theorems~\ref{thm:main-RBM-mixed} and~\ref{thm:main-RBM-ferro}); we believe our methods could be used
to establish the hardness of identity testing for other spin systems.

We introduce some useful notation next. Recall that in the introduction we define the families of models $\MRBM$, $\MfRBM$, $\MPotts$ and $\MfPotts$. We formalize and extend this notation as follows.

\begin{defn} 
	\label{dfn:potts:notation}
		For integers $n,d \ge 3$ and $\beta, h \in \R$,
		let $\MPotts(n,d,\beta,h,q)$ denote the family of $q$-state Potts models on $n$-vertex graphs $G=(V_G,E_G)$ of maximum degree at most $d$ with edge interactions and external field given by $\beta_G: E_G \rightarrow \R$ and $h_G: V_G \times [q] \rightarrow \R$, respectively, such that:
		\begin{enumerate}[(i)]
			\item for every edge $\{u,v\}\in E_G$, $|\beta_G(\{u,v\})|\le\beta$; and
			\item for every vertex $v\in V_G$ and spin $i\in [q]$, $|h_G(v,i)|\le h$.
		\end{enumerate}
\end{defn}

\begin{rmk}
	We omit $q$ from the notation above as it is usually clear from context. For the special case of $q=2$, 
	i.e., the Ising model, we use $\MIsing$; when $q=2$ and the underlying graph is bipartite we use $\MRBM$.
	In addition, we add ``$+$'' or ``$-$'' as a superscript to the notation to denote the corresponding ferromagnetic or antiferromagnetic
	subfamilies; e.g., $\MfPotts(n,d,\beta,h)$ denotes the subset of ferromagnetic Potts models in $\MPotts(n,d,\beta,h)$. Finally, we add a circumflex, e.g., $\hatMfPotts(n,d,\beta,h)$, for the subfamily of models where every edge weight is \emph{exactly} equal to $\beta$.
\end{rmk}

\subsection{Step 1: Decision version of approximate counting}\label{Potts-decision}

Our starting point is always a known hard approximate counting instance.
For the ferromagnetic Potts model,
we consider the problem of approximating its partition function on a graph $G$. As mentioned in the introduction,
this problem is known to be {\BIS}-hard, even under the additional assumptions that all edges have the same interaction parameter $0 < \beta_G = \Theta(1)$ and that there is no external field (i.e., $h=0$)~\cite{GJ,GSVY}.
Our goal is to design an $\fpras$ for the partition function $Z_{G,\beta_G} := Z_{G,\beta_G,0}$ using a polynomial-time algorithm for identity testing, 
thus establishing the {\BIS}-hardness of this problem.

Our first step is to reduce the problem of approximating $Z_{G,\beta_G}$ to a natural decision variant of the problem. This decision version will be more naturally solved by the testing algorithm and is more generally defined as follows:

%as outlined in the following subsection.
%
%Hence, consider an instance of the ferromagnetic Potts model which we denote as $G,\beta_G$ (we can assume all
%edges have the same weight).
%
%Our first step is to reduce the approximation of the partition function to a corresponding decision problem
%as outlined in the following subsection.  We then prove hardness of identity testing on general instances in
%Section~\ref{Potts-proof-general}.  Finally, in Section~\ref{Potts-degree} we apply our degree reducing gadget
%to obtain Theorem~\ref{thm:Potts}

%\subsection{Potts: Decision version}\label{Potts-decision}

\begin{center}
	\setlength{\fboxsep}{5pt}
	\noindent\fbox{
		\parbox{0.88\textwidth}{
			%\vspace{-5pt}
			%\begin{center}
			%	\textit{Decision version of $r$-approximate counting for the  ferromagnetic Potts Model}
			%\end{center}
			\begin{restatable}[Decision $r$-approximate counting]{defn}{defDecisionCnt}
				\label{def:decision-cnt}
			%\emph{Decision $r$-approximate counting:} 
			Given a Potts model ($G$,$\beta_G$,$h_G$),
			an approximation ratio $r>1$ and an input $\hat{Z} \in \R$, distinguish with probability at least $5/8$ between the following two cases:
			\[
			\text{(i)}~ Z_{G,\beta_G,h_G} \le \frac{1}{r} \hat{Z}
			\quad\quad \text{(ii)}~
			Z_{G,\beta_G,h_G} \ge r \hat{Z}
			\]
			\end{restatable}
		}
	}
\end{center}

We show that the decision version of approximate counting is as hard as the standard problem of approximating $Z_{G,\beta_G,h_G}$. 

\begin{thm}
	\label{thn:decision-cnt:potts}
	Let $n,d \ge 1$ be integers and let $\beta,h\ge0$ be real numbers.
	Suppose that there is no $\fpras$ for the counting problem for a family of Potts models 
	$\mathcal M$, where 
	$$
	\mathcal M \in \{\hatMfPotts(n,d,\beta,h),\hatMaIsing(n,d,\beta,h),\hatMfIsing(n,d,{\beta},{h})\}.
	$$
	Then, for any $c>0$ there is no polynomial-time algorithm for the decision version of $n^c$-approximate counting for $\mathcal M$.
\end{thm}

Our proof of this theorem is provided in Section~\ref{app:app-cnt}.
%, is fairly general and applies to most spin systems of interest, including RBMs, Ising models, Potts models, hard-core models, proper colorings, etc.
%In particular, it still holds if we replace $\MfPotts$ by $\MPotts$, $\MRBM$ or $\MfRBM$.

\subsection{Step 2: Testing instance construction}\label{Potts-proof-general}

We first construct a hard instance for the identity testing problem for the ferromagnetic Potts model on general graphs, with no restriction on the maximum degree and with a constant upper bound on the edge interactions.
We prove first that identity testing is {\BIS}-hard in this setting.

\begin{thm}
\label{thm:Potts-general}
	Consider a ferromagnetic Potts model with no external field ($h=0$) where the interaction on every edge
	is ferromagnetic and bounded from above by a constant $\beta_0 > 0$.
	Then, there is no polynomial-time
	identity testing algorithm for the model unless there is an $\fpras$ for \#BIS.
\end{thm}

To establish this theorem, we construct an identity testing instance that allows us to solve the decision variant of approximate counting (see Definition~\ref{def:decision-cnt}).
We note that this theorem does not immediately imply Theorem~\ref{thm:Potts} from the introduction because
we allow the degree to be unbounded; specifically, Theorem~\ref{thm:Potts-general}
establishes hardness for $\MfPotts(n,n,\beta,0)$.
The next step of the proof uses this result
and a degree-reducing gadget to
establish Theorem~\ref{thm:Potts} (see Section~\ref{Potts-degree}).
Our main gadget in the proof of Theorem~\ref{thm:Potts-general} will be a complete graph $H$ on $m$ vertices; this is known as the \emph{mean-field} case in statistical physics.
%
%Let $G,\beta_G$ be an instance of ferromagnetic Potts model.
%Our aim is to estimate the partition function $Z=Z_{G,\beta_G}$ using an algorithm
%for the identity testing problem.
%
%We will create a collection of instances of the
%testing problem (recall that an instance consists of a hidden model and a visible model) together with a polynomial-time
%sampler from the hidden model such that a testing algorithm will reveal an approximation
%of the partition function $Z$ of $G$. Note, the partition function $Z$ for $G$ is at least $1$ (taking the
%monochromatic configuration and using the fact that we are looking at a ferromagnetic system)
%and at most $q^N\exp(\beta_G\binom{N}{2})$.
%We will use parameter $m=N^{4}$ in our construction.

\subsubsection{The ferromagnetic mean-field $q$-state Potts model}
\label{subsubsec:mf}

Let $H = K_m$ be a complete graph on $m$ vertices and let $\beta_H$ be
the interaction strength on the edges of $H$.
By symmetry, 
the $q$-state Potts configurations on a complete graph can be described by their
``signature''---by ``signature'' we mean the vector $(\sigma_1,\dots,\sigma_q) \in \Z^q$ where $\sigma_i \ge 0$ is
the number of vertices that have spin $i$; note that $\sum_{i=1}^q \sigma_i = m$.

In the complete graph, the ferromagnetic Potts model is known to undergo an ``order-disorder'' phase transition. 
Specifically, 
there exists a critical value $\beta_H = \mathfrak{B}_o/m$
such that when $\beta_H < \mathfrak{B}_o/m$, long-range
correlations do not exist; the system is then said to be in a ``disordered'' state as the typical configurations
have signature $\approx (m/q,\dots,m/q)$ where each spin has roughly the same density (up to lower order terms).
In contrast, when $\beta_H >  \mathfrak{B}_o/m$, typical configurations have a dominant spin
and the remaining spins are uniformly distributed. These configurations are thus referred to as
``majority'' configurations.   More precisely there exists a constant $\alpha=\alpha(\beta_H)>1/q$ and, with
high probability,
configurations from the Gibbs distribution have signature 
$\approx \left(\alpha m,\frac{(1-\alpha)m}{q-1},\dots,\frac{(1-\alpha)m}{q-1}\right)$ up to permutations and lower order terms.
%Hence, the value of the partition function is  essentially determined by the weights of two types of configurations whose weights 

When $q \ge 3$, the phase transition is known to be of \emph{first-order}, which means that at the critical point $\beta_H = \mathfrak{B}_o/m$ both disordered and majority configurations occur
with constant probability.
This phenomena is referred to as \emph{phase co-existence}, and it
 is known (or conjectured) to be present in a variety of graphs, being
the root reason for the hardness of sampling and counting for the ferromagnetic Potts model.
In contrast, in the Ising model (i.e., when $q=2$),
there is a \emph{second-order phase transition} and the majority density $\alpha(\mathfrak{B}_o)$ is $1/q$ at the critical point; hence these two phases -- disordered and majority -- coincide at this point.

 We now formalize the notion of the majority phase $M$, the disordered phase $D$,
 and the remaining configurations $S$ with their corresponding partition functions $Z_H^{\MM}$, $Z^{\DD}_H$, and $Z_H^{\SS}$.
 The majority phase is defined with respect to a fixed constant $\hat{\alpha} = \hat{\alpha}(\mathfrak{B}_o)$ which is the density of the dominant color
at the phase coexistence point $\mathfrak{B}_o/m$.
%Let $A \subset \Z^q$ be the set of all possible signature vectors. That is,
%$$
%A :=\Bigg\{a=(a_1,\dots,a_q) \in \Z^q\,\big|\,\sum_{i=1}^q a_i = m \mbox{ and } a_1 ,\dots, a_q \ge 0\Bigg\}.
%$$
Let $\Omega_H$ denote the set of Potts configurations on $H$ and
for $\sigma \in \Omega_H$, let $(\sigma_1,\dots,\sigma_q) \in \Z^q$ denote its signature. 
Consider the following sets:
$$
M:=\Bigg\{\sigma \in \Omega_H\,\big|\, ~\exists j \in [q]:~|\sigma_j -\hat{\alpha} m|\leq m^{3/4} \mbox{ and } \left|\sigma_i-\frac{1-\hat{\alpha}}{q-1}m\right|\leq m^{3/4}\ \mbox{for}\ i\in [q]\setminus \{j\} \Bigg\},
$$
$$
D:= \Bigg\{\sigma \in \Omega_H\,\big|\, ~\forall i\in [q]:~|\sigma_i- m/q|\leq m^{3/4}\Bigg\},
$$
and
$
S:= \Omega_H \setminus (M\cup D).
$

For a configuration $\sigma$ on the complete graph $H = (E_H,V_H)$, let
$$w_H^\sigma(\beta_H) = \exp\left( \sum_{\{u,v\}\in E(H)} \beta_H\1(\sigma(u)=\sigma(v)) \right)$$
 denote the weight of $\sigma$ in the mean-field model $(H,\beta_H)$.
Consider the contributions of each type of configuration to the partition function. That is,
\begin{equation*}\label{jjj}
Z_H^{\MM}(\beta_H):= \sum_{\sigma\in M} w_H^\sigma(\beta_H),
\quad\quad
Z_H^{\DD}(\beta_H) := \sum_{\sigma\in D} w_H^\sigma(\beta_H),
\quad\quad
Z_H^{\SS}(\beta_H) := \sum_{\sigma\in S} w_H^\sigma(\beta_H).
\end{equation*}
Hence, the partition function of $(H,\beta_H)$ is given by $Z_H(\beta_H) = Z_H^{\MM}(\beta_H) + Z_H^{\DD}(\beta_H) + Z_H^{\SS}(\beta_H)$.
We note that 
in our reduction later, we will choose a specific $\beta_H>0$ depending on the instance of the approximate counting problem and the parameters of the identity testing algorithm; hence,
to emphasize the effect of $\beta_H$, we parameterize $Z_H^{\MM}$ (and other functions in this section) in terms of $\beta_H$. 

The following two lemmas detail the relevant behavior of the mean-field Potts model at and around the critical point $\mathfrak{B}_o/m$.
We note that as a consequence of the first-order phase transition, there is a critical window around $\mathfrak{B}_o/m$ where the non-dominant phase (i.e., disorder or majority) is still much more likely than any other type configurations; this phenomena is known as \emph{metastability} and will also be crucial for us. 

First we establish that in the critical window around $\mathfrak{B}_o/m$ the majority $M$ and disordered $D$ configurations are exponentially more likely than the remaining configurations $S$.  Several variants of this result have been proved in some fashion before, e.g.,~\cite{BGJ,LL,GJ,CDLLPS,GLP,GSVmf,BSmf}; however, the precise bound we require in our proofs does not seem to be available in the literature.
 
 \begin{lemma}\label{lem:binary-betaH-meta}
 	There exists constants $c,c'>0$ such that for any $\beta_H$ satisfying $|\beta_H-\mathfrak{B}_o/m|\leq c' m^{-3/2}$ we have
 	\begin{equation*}
 	Z_H^{\SS}(\beta_H) \leq \min\{Z_H^{\MM}(\beta_H),Z_H^{\DD}(\beta_H)\}\exp(-c \sqrt{m}).
 	\end{equation*}
 \end{lemma}

In addition, we show that we can find  in $\poly(m)$ time a value for the parameter $\beta_H$ in the critical window to achieve a specified ratio $R$ of the majority partition function $Z^\mathrm{M}_H(\beta_H)$ to the disordered partition function $Z^\mathrm{D}_H(\beta_H)$.

\begin{lemma}\label{lem:binary-betaH-ratio}
There exist constants $c,c'>0$ such that for any $R\in [{\e}^{-c\sqrt{m}},{\e}^{c\sqrt{m}}]$ and any constant $\delta\in(0,1)$, 
we can efficiently find $\beta_H>0$ in $\poly(m)$ time such that $|\beta_H-\mathfrak{B}_o/m|\leq c' m^{-3/2}$ and 
\begin{equation}\label{rrr}
(1-\delta)R \le \frac{Z_H^{\MM}(\beta_H)}{Z_H^{\DD}(\beta_H)} \le R.
\end{equation}
\end{lemma}

The proof of these two lemmas is provided in Appendix~\ref{app:mean-field-potts}.

\subsubsection{Identity testing reduction}
\label{subsubsec:potts-reduce}

\noindent\textbf{Visible Model Construction.}\ \ 
%We now construct the visible model  for the identity testing problem.
Let $(G,\beta_G)$ be the instance of the ferromagnetic Potts model with no external field (i.e., $h = 0$) for which we are trying to approximate the partition function $Z_{G,\beta_G}$; we shall assume $G=(V_G,E_G)$ is an $N$-vertex graph and that every edge has interaction strength
$0< \beta_G = \Theta(1)$.
Let $H = (V_H,E_H)$ be a complete graph on $m= N^{10}$ vertices.
The graph $F = (V_F,E_F)$ is the result of connecting the vertices of $H$ and $G$ with a complete bipartite graph $K_{m,N}$ with edges $E_{m,N}$; that is, $V_F = V_G \cup V_H$ and $E_F = E_H \cup E_G \cup E_{m,N}$.
We consider the Potts model on the graph $F$
with edge interactions
$\beta_F: E_F \rightarrow \R$ given by:
\[
\beta_F(e) = \left.
\begin{cases}
\beta_H & \text{if } e \in E_H\\
\beta_G & \text{if } e \in E_G\\
\beta & \text{if } e \in E_{m,N},
\end{cases}
\right.
\]
where $\beta_H,\beta > 0$ will be chosen later. 
We use $n:=N+m$ for the number of vertices of $F$, and,
with a slight abuse of notation, we use $F$ for the Potts model $(F,\beta_F)$ which will play the role of the visible model in our reduction;  
$\mu_F$ denotes the corresponding Gibbs distribution.

We study first the properties of ``typical'' configurations on $G$ conditional on a
configuration $\sigma$ on the complete graph $H$.  
For this, we introduce some additional notation.
Let $\Omega_F$, $\Omega_H$ and $\Omega_G$ be the set of Potts configuration on the graph $F$, $H$ and $G$ respectively; note that $\Omega_F= \Omega_H \times \Omega_G$.
For $\sigma \in \Omega_H$, define
\[
Z_{F}^\sigma(\beta_H) := \sum_{\eta \in \Omega_F: \eta(V_H) = \sigma} {w_F^\eta(\beta_H)}
\]
where the weight $w_F^\eta(\beta_H)$ of configuration $\eta$ is given by
$$
w_F^\eta(\beta_H) =  \exp\left( \sum_{\{u,v\}\in E_F} \beta_F(\{u,v\})\1(\eta(u)=\eta(v))\right);
$$
that is, $Z_{F}^\sigma(\beta_H)$ is the total contribution to the partition $Z_F(\beta_H)$ of $F$ 
of the configurations that agree with $\sigma$ on $H$.  

 If we fix a configuration $\sigma$ on $H$ and look at the configuration on $G$ (under the Gibbs distribution on $F$
 conditional on $\sigma$) then $\sigma$ will act as an external field on the vertices of $G$.
 We show that if $\sigma$ is in the majority phase (i.e., in the set $M$), then the configuration on $G$
 will be monochromatic with high probability as these configurations will maximize the number of monochromatic
 edges between $G$ and $H$. In contrast, when $\sigma$ is in the disordered phase (i.e., in $D$), then every configuration
on $G$ will have (roughly) the same number of monochromatic edges between $G$ and $H$; hence,
the partition function $Z_F^\sigma(\beta_H)$ in this case will be proportional to $Z_{G,\beta_G}$.

To formalize this,
we split the partition function of $F$ into three parts
depending on the signature on the complete graph $H$. Let
\[  Z_F^{\MM}(\beta_H) = \sum_{\sigma\in M} Z^\sigma_{F}(\beta_H), \ \ Z_F^{\DD} (\beta_H)= \sum_{\sigma\in D} Z^\sigma_F(\beta_H), \ \mbox{ and  }
 Z_F^{\SS}(\beta_H) = \sum_{\sigma\in S} Z^\sigma_F(\beta_H);
 \]
then, 
$
Z_F(\beta_H) = Z_{F}^{\MM}(\beta_H) +  Z_{F}^{\DD}(\beta_H) + Z_{F}^{\SS}(\beta_H).
$

The following lemma details the above description of the
properties of configurations on the original instance $G$ conditional on the configuration on the
complete graph $H$.

\begin{lemma}\label{lem:M-D-S}
%There exist constants $c_1,c_2,c_3>0$ such that the following holds. 
%For any $\beta_H$ such that $|\beta_H-\mathfrak{B}_o/m|\leq c' m^{-3/2}$, any $\beta\in [c_1 N/m , c_2/(Nm^{3/4})]$ and any constant $\delta\in(0,1)$:
For any constants $\delta \in (0,1)$ and $ c > 0$, and
any $\beta_H$ such that $|\beta_H-\mathfrak{B}_o/m|\leq c m^{-3/2}$,
there exists constants $c_1,c_2>0$ such that  for any $\beta\in \left[\frac{c_1 N}{m} , \frac{c_2}{Nm^{3/4}}\right]$:

\begin{enumerate}
\item 
When the configuration on $H$ is in the majority phase, the configuration on $G$ is likely to be monochromatic; more precisely,
\begin{equation}
\label{eq:maj}
e^{-\delta} \cdot Z_H^{\MM} \cdot \exp\left( \hat{\alpha} \beta N m + \beta_G |E_G| \right) \le 
Z_F^{\MM}(\beta_H) 
\le e^\delta \cdot Z_H^{\MM} \cdot \exp\left( \hat{\alpha} \beta N m + \beta_G |E_G| \right).
\end{equation}
\item When the configuration on $H$ is in the disordered phase, the configuration on $G$ will have very limited influence from the configuration on $H$; more precisely,
\begin{equation}
\label{eq:dis}
e^{-\delta} \cdot Z_H^{\DD} \cdot Z_G \cdot \exp\left( \beta N m/q \right) 
\le Z_F^{\DD} (\beta_H) 
\le  e^{\delta} \cdot Z_H^{\DD} \cdot Z_G \cdot \exp\left( \beta N m/q \right).
\end{equation}
\item The remaining configurations on $H$ have a small contribution to the partition function of the model $F$; more precisely, 
\begin{equation}\label{by2-main}
Z_F^{\SS}(\beta_H) \le Z_F(\beta_H)
\exp\left(-\Omega(\sqrt{m})\right).
\end{equation}
\end{enumerate}
\end{lemma}

We remark that the factors~$\exp(\hat{\alpha} \beta N m + \beta_G |E_G|)$ and $\exp(\beta N m /q)$
in~\eqref{eq:maj} and~\eqref{eq:dis}, respectively, account for the contribution of
all the monochromatic edges in $G$ and between $G$ and $H$ in each case.

\begin{proof}[Proof of Lemma~\ref{lem:M-D-S}]
%It will be convenient to separate the interaction of edges in $H$ (that captures the 
%phase coexistence in the mean-field model) and the interaction in $G$ and between $H$ and $G$
%(that captures the effect of different phases on $G$. Let
%\[
%Z_{F\setminus H}^\sigma(\beta_H) := \sum_{\eta:V(F)\to [q] \atop \eta(H)=\sigma} \frac{w_F^\eta(\beta_H)}{w^\sigma_H(\beta_H)},
%\]
%
%When the dependence of these quantities on $\beta_H$ is clear, it will be dropped from the notation for simplicity; e.g.,  $Z_{F\setminus H}^\sigma =Z_{F\setminus H}^\sigma(\beta_H)$. 
%Lemma~\ref{lem:M-D-S} will follow as a corollary of the following lemma.
%
%\begin{lemma}\label{rrrrrt}
%	There exist $c_1,c_2>0$ such that for $\beta \in [c_1 N/m, c_2/(N m^{3/4})]$ the following is true. Let $Z^\sigma_{F}$ be the conditional partition function of $F$ where we condition on the
%	configuration $\sigma$ on the complete graph $H$. Suppose we fix the configuration $\sigma$ on $H$ to an assignment that is in $M$.
%	Then
%	\begin{equation}\label{ehh1}
%	Z_{F\setminus H}^\sigma = \exp(\pm\eps) \exp(\beta \hat{\alpha} N m  ) \exp(\beta_G |E|).
%	\end{equation}
%	Suppose we fix the configuration $\sigma$ on $H$ to an assignment that is in $D$.
%	\begin{equation}\label{ehh2}
%	Z_{F\setminus H}^\sigma = \exp(\pm\eps) \exp(\beta N m /q) Z_G.
%	\end{equation}
%\end{lemma}

We fix $\beta_H$ and, for ease of notation, we drop the dependence on $\beta_H$ throughout the proof; i.e., $Z_F^\MM(\beta_H)$ becomes $Z_F^\MM$,
$w_H^\sigma(\beta_H)$ becomes $w_H(\sigma)$ for $\sigma \in \Omega_H$ and
$w_F^\eta(\beta_H)$ becomes $w_F(\eta)$ for $\eta \in \Omega_F$.

Let $\sigma \in \Omega_H$ and $\tau \in \Omega_G$.
When computing weight for configuration $\sigma \cup \tau$ (i.e., the configuration of $F$ that results from combining the spins assignment of $\sigma$ and $\tau$ in $H$ and $G$, respectively,
it will be convenient to separate the interaction of edges in $H$ (that captures the 
phase coexistence in the mean-field model) and the interaction in $G$ and between $H$ and $G$
(that captures the effect of different phases on $G$). 
Thus, let
$$
w_{F\setminus H}(\sigma \cup \tau) := \frac{w_F(\sigma \cup \tau)}{w_H(\tau)}.
$$
Then,
$$
Z_F^\mathrm{M} = \sum_{\sigma\in M} \sum_{\tau \in \Omega_G} w_H(\sigma) w_{F \setminus H}(\sigma \cup \tau).
$$

For $\sigma \in M$, let 
$(\sigma_1,\dots,\sigma_q) \in \Z^q$ be its signature;
suppose w.l.o.g.\ that $\sigma_1$ is such that
$|\sigma_1 -\hat{\alpha} m| \leq m^{3/4}$ and $|\sigma_i-\frac{1-\hat{\alpha}}{q-1}m| \le m^{3/4}$ for all $i \in \{2,\dots,q\}$. 
Consider the configuration $\eta_1$ on $G$
that assigns spin $1$ to every vertex of $G$
and let $a = \max_{i \in \{2,\dots,q\}} \sigma_i$. 
For any other configuration $\tau \neq \eta_1$ on $G$ with $t \ge 1$ vertices not assigned spin $1$, we have that 
	\begin{equation}
	\label{eq:unlikely-config}
	w_{F\setminus H}(\sigma \cup \tau)\leq \exp\left( \beta_G |E_G|+  \beta (\sigma_1 (N-t) + at) \right) \leq \exp\left( \beta_G |E_G|+  \beta (\sigma_1 (N-1) + a) \right),
	\end{equation}
	since there are at most $|E_G|$ monochromatic edges in $G$ and at least one vertex in $G$
	has a vertex assigned a spin different from $1$ (thus there are at most $\sigma_1 (N-1) + a$
	monochromatic edges between $G$ and $H$).	
	Hence, we get
	$$
	\frac{w_{F\setminus H}(\sigma \cup \tau)}{w_{F\setminus H}(\sigma \cup \eta_1)} 
	\le \frac{\exp\left( \beta_G |E_G|+  \beta (\sigma_1 (N-1) + a) \right)}{\exp\left( \beta_G |E_G| + \sigma_1 N \beta \right)} 
	\le e^{(a-\sigma_1) \beta } \leq
	e^{ \left( - \alpha' m + 2m^{3/4}\right) \beta }
	\leq e^{ - \alpha'' m \beta},
	$$
	where $\alpha' = \hat{\alpha}  - (1-\hat{\alpha})/(q-1) > 0$ and the rightmost inequality
	is true for some $\alpha''>0$ and sufficiently large $m$. For $c_1=(2\log q)/\alpha''$ we have
	for $\beta\geq c_1 N/m$
	$$
	\frac{w_{F\setminus H}(\sigma \cup \tau)}{w_{F\setminus H}(\sigma \cup \eta_1)}  \leq q^{-2N}.
	$$
	Hence
	\begin{equation}\label{none}
	\sum_{\tau \neq \eta_1 \in \Omega_G} w_{F\setminus H}(\sigma \cup \tau) \leq q^{-N} w_{F\setminus H}(\sigma \cup \eta_1).
	\end{equation}
	Now,  
	\begin{align}
	w_{F\setminus H}(\sigma \cup \eta_1) &= \exp\left( \beta_G |E_G| + \sigma_1 N \beta \right) \nonumber\\
	&\le  \exp\left( \beta_G |E_G| + \hat{\alpha} m  N \beta + m^{3/4} N \beta \right) \nonumber\\ 
	&\le e^{\delta/2} \exp\left( \beta_G |E_G|+ \hat{\alpha} m  N \beta \right) \label{eq:weight:FH},
	\end{align}
		where in the last equality we take $c_2=\delta/2$ and use the fact that $\beta\leq c_2/(N m^{3/4})$.
		Therefore, when $\sigma \in M$ is such that $|\sigma_1 - \hat \alpha m| < m^{3/4}$, we have
			$$
			\sum_{\tau \in \Omega_G} w_{F\setminus H}(\sigma \cup \tau) \leq (1+q^{-N}) w_{F\setminus H}(\sigma \cup \eta_1) \le e^{\delta} \exp\left( \beta_G |E_G|+ \hat{\alpha} m  N \beta \right),
			$$
			for $N$ sufficiently large.
%		Similarly, we can deduce that 
%			\begin{align}
%			w_{F\setminus H}(\sigma \cup \eta_1) \ge e^{-\delta/2} e^{\beta_G |E_G|+ \hat{\alpha} m  N \beta} \label{eq:weight:LowerFH}.
%		\end{align}
	By symmetry, we then get that
	$$
Z_F^\MM \le \sum_{\sigma\in M} w_H(\sigma) e^{\delta} \exp\left( \beta_G |E_G|+ \hat{\alpha} m  N \beta \right) = e^{\delta}  Z^\MM_H  \exp\left( \beta_G |E_G|+ \hat{\alpha} m  N \beta \right).
	$$
	The lower bound in~\eqref{eq:maj} can be derived in similar fashion and part 1 of the lemma follows.
	
	For part 2, suppose that $\sigma\in D$ and let $\tau \in \Omega_G$. Let $\tau_i$ be the number of vertices of $G$ assigned spin $i$ in $\tau$ and let $w_G(\tau)$ denote the weight of $\tau$ for the Potts model $(G,\beta_G)$.	
	Then,
	\begin{align}\label{zz2}
	w_{F\setminus H}(\sigma \cup \tau) &= w_G(\tau) \exp\left( \beta\sum_{i=1}^q \sigma_i \tau_i \right) \nonumber\\ 
	&\le
	w_G(\tau) \exp\left( m^{3/4} N \beta +  \beta m N/q \right) \nonumber\\ 
	&\le
	e^\delta w_G(\tau) \exp\left( \beta m N/q \right),
	\end{align}
	since recall we set $c_2 = \delta/2$. Hence,
	$$
	Z_F^\DD = \sum_{\sigma\in D} \sum_{\tau \in \Omega_G} w_H(\sigma) w_{F \setminus H}(\sigma \cup \tau) \le e^\delta Z_H^\DD Z_G \exp\left( \beta m N/q \right).
	$$
	The lower bound for $Z_F^\DD$ can be derived analogously and part 2 of the lemma follows. 

  Finally for part 3, note that 
	\begin{align*}
	Z_F^\mathrm{S} &= \sum_{\sigma\in S} \sum_{\tau \in \Omega_G} w_H(\sigma) w_{F \setminus H}(\sigma \cup \tau)  
	\le q^N \exp\left( \beta_GN^2 + \beta N m \right) Z^\SS_H 
	\le  \min\{Z_H^\mathrm{M},Z_H^\mathrm{D}\} \exp\left( -\Omega(\sqrt{m}) \right), 
	\end{align*}
	where the last inequality follows for sufficiently large $N$ and $m$ from Lemma~\ref{lem:binary-betaH-meta} and the fact that $\beta < c_2/(Nm^{3/4})$. Then, 
	\begin{equation*}
		\label{by2-app}
	\frac{Z_F^\mathrm{S}}{Z_F}\leq  \frac{Z_F^\mathrm{S}}{Z_F^\mathrm{M}}\leq \exp\left(-\Omega(\sqrt{m})\right),
	\end{equation*}
	and the result follows.
\end{proof}

\medskip
\noindent\textbf{Hidden Model Construction.}\ \ 
We now construct our hidden model and show that we can efficiently generate
samples from its Gibbs distribution.
Let $F^*$ be the graph obtained by our construction above where we replace the  graph $G$ by a complete graph
on $N$ vertices.
More precisely, let $K = K_N$ be a complete graph on $N$ vertices and let $F^*$ be 
the graph that results from connecting the vertices of $K$ and $H$ with a complete bipartite graph $K_{N,m}$.

The edges of $K$ have parameter $\beta_K =  \beta_G+4\log q$, whereas the remaining edges
have the same interaction strength as in $F$; that is,
edges between $K$ and $H$ will have parameter $\beta$ and those in $H$ parameter $\beta_H$.
This Potts model on $F^*$, which again with a slight abuse of notation we denote by $F^*$, will act as the hidden model.
We choose $\beta_K = \beta_G+4\log q$, so that $K$ is more likely to be monochromatic than $G$. 
%\begin{rmk}\label{rmk:lem:M-D-S}
%For the hidden model $F^*$, 
%Lemma~\ref{lem:M-D-S} holds as well (with $F,G$ replaced by $F^*,K$), without any change in the proof. 
%\end{rmk}
Let $\mu_{F^*}$ the corresponding Gibbs distribution on $F^*$.
We show next that we can efficiently generate samples from $\mu_{F^*}$.

\begin{lemma}\label{les}
There is an exact sampling algorithm for the distribution $\mu_{F^*}$ with running time $\poly(n)$. %\zongchen{Changed to exact sampling.}
%There exists
%a sampling
%algorithm with running time $\poly(n,\delta^{-1})$
%such that
%the distribution $\mu_{F^*}^{\normalfont\textsc{alg}}$ of its output satisfies:
%$
%\TV{\mu_{ F^*}}{\mu_{F^*}^{\normalfont\textsc{alg}}} \le \delta.
%$
\end{lemma}

\begin{proof}
Because of symmetry there are at most $n^{2q}$ types of configurations---described by their signatures
on $H$ and $K$; recall that $n = m+N$. We can then enumerate every signature,
explicitly compute its probability and sample from the resulting distribution.
This involves computing multinomial coefficients, but they can each be expressed as product of $q$ binomial coefficients which can be easily computed in $\poly(n)$ time.
Once the signature is generated from the correct distribution, we can simply take a random permutation of the vertices to assign their spins.
\end{proof}

\medskip
\noindent\textbf{Proof Overview.}\ \ 
We provide the high-level idea of the reduction next. 
Recall that our goal  is to provide a polynomial-time algorithm for the decision version of the $r$-approximate counting problem for the ferromagnetic Potts model $(G,\beta_G)$.
That is, for a real number $\hat Z$ we want to determine whether $Z_G \le \frac{1}{r}\hat{Z}$ or $Z_G \ge r \hat{Z}$,
where $Z_G := Z_{G,\beta_G}$ is the partition function of the model $(G,\beta_G)$ .

For any ``reasonable'' $\hat Z \in \R$ (i.e., $\hat Z$ that is not too small or too large, in which case the approximate counting problem becomes trivial), we can find a value of the parameter $\beta_H$ for our construction such that
$$ 
\frac{Z_F^{\DD}(\beta_H)}{Z_F^{\MM}(\beta_H)} 
\approx
\frac{1}{\sqrt{\varepsilon L}} \frac{Z_G}{\hat{Z}},
$$
where $L=L(n)$ and $\eps = \eps(n)$ are the sample complexity and accuracy parameter of the testing algorithm, respectively.
This is possible because of the first-order phase transition of the ferromagnetic mean-field $q$-state Potts model for $q \ge 3$, and the associated
phase coexistence and metastability phenomena discusses earlier; see Section~\ref{subsubsec:mf}.
(Specifically, by Lemma~\ref{lem:binary-betaH-ratio} 
we can find $\beta_H$ so that ${Z_H^{\MM}(\beta_H)}/{Z_H^{\DD}(\beta_H)} \approx R$ for any target $R$, and then we can use Lemma~\ref{lem:M-D-S}
to translate this value to a value for $Z_G \cdot Z_F^{\MM}(\beta_H)/Z_F^{\DD}(\beta_H)$.)

Now, for this choice of $\beta_H$ and setting $r \approx \sqrt{L/\eps}$, note that if $Z_G \le \frac{1}{r}\hat{Z}$, then 
${Z_F^{\DD}(\beta_H)}/{Z_F^{\MM}(\beta_H)}$ is small ($\lesssim 1/L$). Conversely, when $Z_G \ge r \hat{Z}$, the ratio is large ($\gtrsim 1/\eps$). 
Therefore, to distinguish whether $Z_G \le \frac{1}{r}\hat{Z}$ or $Z_G \ge r \hat{Z}$ it is sufficient to determine whether the ratio ${Z_F^{\DD}(\beta_H)}/{Z_F^{\MM}(\beta_H)}$ is small or large. 
For this we can use the identity testing algorithm. 
In particular, when the ratio is small ($\lesssim 1/L$), the majority phase of $H$ is dominant in $F$, and $G$ will likely be monochromatic. Since this is also the case in $F^*$ (i.e., $K$ is monochromatic with high probability), then the models $F$ and $F^*$ will be close in total variation distance ($\lesssim 1/L$), and the testing algorithm using only $L$ samples would output $\textsc{Yes}$. 
Otherwise, when ${Z_F^{\DD}(\beta_H)}/{Z_F^{\MM}(\beta_H)}$ is large ($\gtrsim 1/\eps$), the disorder phase is dominant, so $F$ and $F^*$ are likely to disagree on the spins of $G$ and $K$; this implies that their total variation distance is large ($\gtrsim 1-\eps$), and so the tester would output $\textsc{No}$. 
We proceed to flesh out the technical details next.

\begin{lemma}
	\label{lem:F-to-TV}
	Let $\varepsilon \in (0,1)$ be a constant, $L = L(n) = \poly(n)$ and $r = 96\eps^{-1}\sqrt{\eps L +1}$. 
	Suppose $\hat Z \in \R$ is such that $rq \exp(\beta_G |E_G|) \le \hat{Z} \le \frac{1}{r} q^N \exp(\beta_G |E_G|)$.
	Then, there exists constants $c,c_1,c_2 >0$ such that the following holds. 
	For any $\beta\in \left[\frac{c_1 N}{m} , \frac{c_2}{Nm^{3/4}}\right]$, 
    we can find $\beta_H>0$ in the range $|\beta_H-\mathfrak{B}_o/m|\leq c m^{-3/2}$ in $\poly(n)$ time such that all of the following holds:
\begin{enumerate}[(i)]
    %\begin{align}
    \item $  \frac{1}{4\sqrt{\eps L+1}} \frac{Z_G}{\hat{Z}} \le 
    \frac{Z_F^\mathrm{D}(\beta_H)}{Z_F^\mathrm{M}(\beta_H)} \le \frac{1}{\sqrt{\eps L+1}} \frac{Z_G}{\hat{Z}}$, and  
    $\frac{Z_F^\mathrm{S}(\beta_H)}{Z_F(\beta_H)} \le e^{-c_3 \sqrt{m}};$
    \item 
    $
    \frac{Z_{F^*}^\mathrm{D}(\beta_H)}{Z_{F^*}^\mathrm{M}(\beta_H)} \le
    \frac{2}{r\sqrt{\eps L+1}}$, and $\frac{Z_{F^*}^\mathrm{S}(\beta_H)}{Z_{F^*}(\beta_H)} \le e^{-c_3 \sqrt{m}}; 
    $
    %\end{align}
    %Furthermore, for $\beta_H$ satisfying \eqref{eqn:D/M-F} and \eqref{eqn:D/M-F*} we have:
\item If $Z_G \le \frac{1}{r} \hat{Z}$, then 
$ \TV{\mu_F}{\mu_{F^*}} \le \frac{1}{16L} $;
\item If $Z_G \ge r \hat{Z}$, then 
$ \TV{\mu_F}{\mu_{F^*}}  \ge 1-\eps $.
\end{enumerate}
\end{lemma}

\begin{proof}
	Let $\alpha_0 = \hat{\alpha} - 1/q$.
	By parts 1 and 2 of Lemma~\ref{lem:M-D-S}, for
	any $\beta_H$ such that $|\beta_H-\mathfrak{B}_o/m|\leq c m^{-3/2}$
	and any $\beta\in \left[\frac{c_1 N}{m} , \frac{c_2}{Nm^{3/4}}\right]$ for suitable constants $c_1,c_2 > 0$, we have 
	\[
	\frac{2}{3} \cdot \exp\left( - \alpha_0 \beta Nm - \beta_G|E_G| \right) \cdot \frac{Z_H^\mathrm{D}}{Z_H^\mathrm{M}} \cdot Z_G
	\le 
	\frac{Z_F^\mathrm{D}}{Z_F^\mathrm{M}} 
	\le 
	\frac{4}{3} \cdot \exp\left( - \alpha_0 \beta Nm - \beta_G|E_G| \right) \cdot \frac{Z_H^\mathrm{D}}{Z_H^\mathrm{M}} \cdot Z_G,
	\]
	where for ease of notation we dropped the dependence on $\beta_H$ and set $Z_G = Z_{G,\beta_G}$.
	Moreover, part 3 of the same lemma implies that there exists a constant $c_3 > 0$ such that
	\begin{equation}
	\label{eq:s-small}
	%\frac{Z_H^\mathrm{S}}{Z_H^\mathrm{D}} \le e^{-c_3 \sqrt{m}}, \qquad
	\frac{Z_F^\mathrm{S}}{Z_F} \le e^{-c_3 \sqrt{m}}.
	\end{equation}
	Recall that $n = m + N$ and $m = N^{10}$.
	By Lemma~\ref{lem:binary-betaH-ratio}, 
	 we can find $\beta_H>0$ in $\poly(m)$ time such that $|\beta_H-\mathfrak{B}_o/m|\leq c m^{-3/2}$ and
	 %\antonio{The bounds for $R$ seem quite loose. Check.} 
	 \begin{equation}
	 \label{eq:sec2:ml-1}
	\frac{3}{8\sqrt{\eps L+1}} \cdot \exp\left( \alpha_0 \beta Nm + \beta_G |E_G| \right) \cdot \frac{1}{\hat{Z}} \le
	\frac{Z_H^\mathrm{D}}{Z_H^\mathrm{M}} \le \frac{3}{4\sqrt{\eps L+1}} \cdot \exp\left( \alpha_0 \beta Nm + \beta_G |E_G| \right) \cdot \frac{1}{\hat{Z}} ;
    \end{equation}
	note that the assumptions $rq \exp(\beta_G |E_G|) \le \hat{Z} \le \frac{1}{r} q^N \exp(\beta_G |E_G|)$ and $r = \poly(n)$ ensure that $Z_H^\mathrm{D}/Z_H^\mathrm{M}$ is in the desired range. 
	Thus, for this choice of $\beta_H$ we get
     \begin{equation}
     \label{eq:ratio}
	\frac{1}{4\sqrt{\eps L+1}} \frac{Z_G}{\hat{Z}} \le
	\frac{Z_F^\mathrm{D}}{Z_F^\mathrm{M}} \le \frac{1}{\sqrt{\eps L+1}} \frac{Z_G}{\hat{Z}}.
	 \end{equation}
	This establishes part (i) of the lemma. 
	
	For part (ii), we  note that
	Lemma~\ref{lem:M-D-S} holds for the hidden model $F^*$ (with $F$ and $G$ replaced by $F^*$ and $K$, respectively), without any change in the proof. Hence, we get
	 \begin{equation}
	\label{eq:sec2:ml-2}
	\frac{Z_{F^*}^\mathrm{D}}{Z_{F^*}^\mathrm{M}} 
		\le 
		\frac{4}{3} \cdot \exp\left( - \alpha_0 \beta Nm - \beta_K|E_K| \right) \cdot \frac{Z_H^\mathrm{D}}{Z_H^\mathrm{M}} \cdot Z_K
	 \end{equation}
	and
	 \begin{equation}
	\label{eq:sec2:ml-3}
	\frac{Z_{F^*}^\mathrm{S}}{Z_{F^*}} \le e^{-c_3 \sqrt{m}}. 
	 \end{equation}
	Thus, for our choice of $\beta_H$ we deduce from \eqref{eq:sec2:ml-1}, \eqref{eq:sec2:ml-2} and \eqref{eq:sec2:ml-3} that
	\[
	\frac{Z_{F^*}^\mathrm{D}}{Z_{F^*}^\mathrm{M}} 
	\le \frac{1}{\sqrt{\eps L+1}} \exp(\beta_G |E_G| - \beta_K |E_K|) \cdot \frac{Z_K}{\hat{Z}} 
	\le \frac{2}{r\sqrt{\eps L+1}},
	\]
	where the last inequality follows from $q \exp(\beta_K |E_K|)/Z_K \ge 1/2$ when $\beta_K \ge 4\log q$ and the assumption that $\hat{Z} \ge r q \exp(\beta_G |E_G|)$.  
	
	We prove part (iii) next. 
	Suppose that $Z_G \le \frac{1}{r} \hat{Z}$ and let $\nu_F$ be the conditional distribution of $\mu_F$ conditioned on the configuration on $H$ being in the majority phase (i.e., in the set $M$). That is, for $\sigma \in \Omega_H$ and $\tau \in \Omega_G$, 
	$$\nu_F(\sigma \cup \tau) = \1(\sigma \in M) \frac{\mu_F(\sigma \cup \tau) Z_F}{Z^\MM_F}.$$
	From the definition of total variation distance we have
	\[
	\TV{\mu_F}{\nu_F} = \sum_{\eta \in \Omega_F: \mu_F(\eta) \ge \nu_F(\eta)}  \mu_F(\eta) - \nu_F(\eta) = \frac{Z_F^\mathrm{D} + Z_F^\mathrm{S}}{Z_F}.
	\]
	From~\eqref{eq:s-small},~\eqref{eq:ratio}  and the assumption that $Z_G \le \frac{1}{r}\hat{Z}$, we get
	\[
	\TV{\mu_F}{\nu_F} \le \frac{Z_F^\mathrm{D}}{Z_F^\mathrm{M}} + \frac{Z_F^\mathrm{S}}{Z_F} \le \frac{1}{\sqrt{\eps L+1}} \frac{Z_G}{\hat{Z}} + e^{-c_3\sqrt{m}} \le \frac{1}{r\sqrt{\eps L +1}} + e^{-c_3\sqrt{m}}.
	\]
	Since $r = 96\eps^{-1}\sqrt{\eps L +1}$, it follows that
	\begin{equation}\label{eq:mu-nu-TV}
	\TV{\mu_F}{\nu_F} \le \frac{\eps}{96(\eps L+1)} + e^{-c_3\sqrt{m}} \le \frac{1}{96L} + e^{-c_3 \sqrt{m}}.
	\end{equation}
	Similarly, for the distribution $\mu_{F^*}$ and the conditional distribution $\nu_{F^*}$ of the majority phase, we also have
	\begin{equation}\label{eq:mu*-nu*-TV}
	\TV{\mu_{F^*}}{\nu_{F^*}} \le \frac{Z_{F^*}^\mathrm{D}}{Z_{F^*}^\mathrm{M}} + \frac{Z_{F^*}^\mathrm{S}}{Z_{F^*}} \le \frac{2}{r\sqrt{\eps L +1}} + e^{-c_3\sqrt{m}} \le \min\left\{\frac{1}{48L}, \frac{\eps}{48}\right\} + e^{-c_3\sqrt{m}}.
	\end{equation}
	
	Let $\mathcal{A}$ be the event that all vertices of $G$ are assigned the same spin.  
	By drawing a sample from $\nu_F$ and sequentially resampling the spin of each vertex of $G$,
	we deduce from a union bound and the fact $\hat \alpha > 1/q$ that
	\[
	1-\nu_F(\mathcal{A}) \le N \cdot \frac{\exp \left( \beta_G N + \beta (\frac{1-\hat{\alpha}}{q-1} m + m^{3/4}) \right)}{\exp \left( \beta(\hat{\alpha} m - m^{3/4}) \right)} \le e^{-\gamma\beta m}
	\]
	for a suitable constant $\gamma >0$; 
	similarly
	\[
		1 - \nu_{F^*}(\mathcal{A}) \le e^{-\gamma\beta m}.
	\]
	Let $\rho = \nu_F(\,\cdot\, | \mathcal{A})$ denote the conditional distribution of $\nu_F$ given $\mathcal{A}$. 
	Observe that $\rho$ does not depend on the graph $G$, because we condition on the event that all vertices from $G$ receive the same spin, and thus the structure of $G$ does not affect the conditional distribution $\rho$.
	In particular, we have $\rho = \nu_F(\,\cdot\, | \mathcal{A}) = \nu_{F^*}(\,\cdot\, | \mathcal{A})$. 
	Thus, we get 
	\begin{equation}\label{eq:nu-nu*-TV}
	\TV{\nu_F}{\nu_{F^*}} \le 
	\TV{\nu_F}{\rho} + \TV{\nu_{F^*}}{\rho} 
	= 1-\nu_F(\mathcal{A}) + 1-\nu_{F^*}(\mathcal{A}) 
	\le 2e^{-\gamma\beta m}. 
	\end{equation}
	%Observe that $\nu_F(\cdot\,|\, \mathcal{A}) = \nu_{F^*}(\cdot\,|\, \mathcal{A})$.
	%Therefore, we get
	%\begin{align}
	%	\TV{\nu_F}{\nu_{F^*}} &\le \min\{\nu_F(\mathcal{A}),\nu_{F^*}(\mathcal{A})\} \TV{\nu_F(\cdot\,|\, \mathcal{A})}{\nu_{F^*}(\cdot\,|\, \mathcal{A})} + 1 - \min\{\nu_F(\mathcal{A}),\nu_{F^*}(\mathcal{A})\}\nonumber\\ 
	%	&\le e^{-\gamma\beta m}.
	%	\label{eq:nu-nu*-TV}
	%\end{align}
	
	From \eqref{eq:mu-nu-TV}, \eqref{eq:mu*-nu*-TV}, \eqref{eq:nu-nu*-TV} and the triangle inequality, we conclude that 
	\begin{align*}
		\TV{\mu_F}{\mu_{F^*}} &\le \TV{\mu_F}{\nu_F} + \TV{\mu_{F^*}}{\nu_{F^*}} + \TV{\nu_F}{\nu_{F^*}}\\ 
		&\le \frac{1}{32L} + 2e^{-c_3\sqrt{m}} + 2e^{-\gamma\beta m} \le \frac{1}{16L}. 
	\end{align*}
	and part (i) follows.
	
	Finally, for part (iv), suppose that
	$Z_G \ge r\hat{Z}$. Then,
	\begin{equation}\label{eq:2-nu-nu*-TV}
	\TV{\mu_F}{\nu_F} = 1 - \frac{Z_F^\mathrm{M}}{Z_F} \ge 1 - \frac{Z_F^\mathrm{M}}{Z_F^\mathrm{D}} \ge 1 - 4\sqrt{\eps L +1} \frac{\hat{Z}}{Z_G} \ge 1- \frac{4}{r}\sqrt{\eps L+1} = 1-\frac{\eps}{24}.
	\end{equation}
	Thus, equations \eqref{eq:2-nu-nu*-TV}, \eqref{eq:mu*-nu*-TV}, \eqref{eq:nu-nu*-TV} and the triangle inequality imply that
	\begin{align*}
		\TV{\mu_F}{\mu_{F^*}} &\ge \TV{\mu_F}{\nu_F} - \TV{\mu_{F^*}}{\nu_{F^*}} - \TV{\nu_F}{\nu_{F^*}}\\
		&\ge 1- \frac{\eps}{16} - e^{-c_3\sqrt{m}} - 2e^{-\gamma\beta m} \ge 1-\eps,
	\end{align*}
	and the result follows.
\end{proof}

\subsubsection{A generic reduction from counting to testing}

Theorem~\ref{thm:Potts-general} will follow from Lemmas~\ref{les} and~\ref{lem:F-to-TV} using the following general reduction from the decision version of $r$-approximate counting to testing. 

\begin{thm}
	\label{thm:gen-red}
	Let $(G,\beta_G,h_G)$ be a Potts model on an $N$-vertex graph $G$ with partition function $Z_G$ and let $\hat Z \in \R$.
	Let $\varepsilon \in (0,1)$ be a constant, $n = \poly(N)$ and 
	suppose there exists an $\varepsilon$-identity testing algorithm for 
	a family of Potts models $\mathcal M$ on $n$-vertex graphs 
	with sample complexity $L = L(n) = \poly(n)$ and $\poly(n)$ running time.
	Suppose that given $(G,\beta_G,h_G)$, $\hat Z$, $\varepsilon$ and $L$, 
	there exists $r = \poly(L,\varepsilon^{-1})$ 
	such that
	we can construct  
	two models $ F,  F^* \in \mathcal M$ in $\poly(n)$ time
	satisfying:
	\begin{enumerate}[(i)]
		\item If $Z_G \le \frac 1r \hat Z$, then  
			$ \TV{\mu_{ F}}{\mu_{ F^*}} \le \frac{1}{16L} $;
		\item If $Z_G \ge r \hat{Z}$, then 
		$ \TV{\mu_{ F}}{\mu_{ F^*}}  \ge 1-\eps$; and
		\item We can generate samples from a distribution $\mu_{F^*}^{\textsc{alg}}$ such that
		$\TV{\mu_{ F^*}}{\mu_{ F^*}^{\textsc{alg}}}  \le \delta$ in time $\poly(n,\delta^{-1})$.
	\end{enumerate}
Then, there is a $\poly(N)$ running time algorithm for the decision version of $r$-approximate counting for $(G,\beta_G,h_G)$ that succeeds with probability at least $5/8$.
\end{thm}

\begin{proof}
	Recall that the input to the decision version of $r$-approximate counting is the model $(G,\beta_G,h_G)$ and a real number $\hat{Z}>0$; the goal is to determine whether $Z_G \le \frac{1}{r} \hat{Z}$ or $Z_G \ge r \hat{Z}$. The algorithm proceeds as follows:
	\begin{enumerate}
		\item Construct the Potts models $F$ and $F^*$ in $\mathcal M$.
		\item Generate $L = L(n)$ $\delta$-approximate samples $\mathcal{S} = \{\sigma_1,\dots,\sigma_L\}$ from $\mu_{F^*}^{\textsc{alg}}$, setting $\delta = \frac{1}{16L}$.
		\item The input to the testing algorithm, henceforth called the \textsc{Tester}, is $F$, which plays the role of the visible model, and the samples $\mathcal{S}$.
		\item If the \textsc{Tester} outputs \textsc{Yes}, then return $Z_G \le \frac{1}{r} \hat{Z}$.
		\item If the \textsc{Tester} outputs \textsc{No}, then return $Z_G \ge r \hat{Z}$.
	\end{enumerate}
	
	We show next that our output for decision version of $r$-approximate counting is correct with probability at least $5/8$. 
	Consider first the case when $Z_G \le \frac{1}{r} \hat{Z}$.
	If this is the case, then by assumption we have
	$
	\TV{\mu_{F}}{\mu_{F^* }} \le \frac{1}{16L}
	$
	and 
	\begin{equation}
	\label{eq:F*-alg}
	\TV{\mu_{F^* }}{\mu_{F^* }^{\textsc{alg}}} \le \frac{1}{16L}.
	\end{equation}
	So, by the triangle inequality,
	$$
	\TV{\mu_F}{\mu_{F^* }^{\textsc{alg}}} \le \frac{1}{8L}.
	$$
	Let $(\mu_{F})^{\otimes L}$, $(\mu_{F^* })^{\otimes L}$ and $(\mu_{F^* }^{\textsc{alg}})^{\otimes L}$ be the product distributions corresponding to $L$ independent samples from $\mu_{F}$, $\mu_{F^* }$ and $\mu_{F^* }^{\textsc{alg}}$ respectively.
	We have
	$$
	\TV{(\mu_{F})^{\otimes L}}{(\mu_{F^* }^{\textsc{alg}})^{\otimes L}} \leq L \TV{\mu_F}{\mu_{F^* }^{\textsc{alg}}} \le \frac 18. 
	$$
	Hence, if $\pi$ is the optimal coupling of the distributions $(\mu_{F^* }^{\textsc{alg}})^{\otimes L}$ and $(\mu_{F})^{\otimes L}$, and $(\mathcal{S},\mathcal{S}')$ is sampled from $\pi$, then 
	$\mathcal{S} \sim (\mu_{F^* }^{\textsc{alg}})^{\otimes L}$, $\mathcal{S}' \sim (\mu_{F})^{\otimes L}$ and $\pi(S\neq S' ) \le \frac{1}{8}$. Therefore,
	\begin{align}
	\label{eq:coupling-main}
	\Pr[ \textsc{Tester}&~\text{outputs}~\textsc{No}~\text{when given samples}~\mathcal{S}~\text{where}~\mathcal{S}\sim(\mu_{F^* }^{\textsc{alg}})^{\otimes L}]\notag\\
	={}& \Pr[  \textsc{Tester}~\text{outputs}~\textsc{No}~\text{when given samples}~\mathcal{S}~\text{where}~(\mathcal{S},\mathcal{S}')\sim\pi]\notag \\
	\leq{}&	\Pr[ \textsc{Tester}~\text{outputs}~\textsc{No}~\text{when given samples}~\mathcal{S}'~\text{where}~(\mathcal{S},\mathcal{S}')\sim\pi]+\pi(S\neq S' )\notag\\
	={}&\Pr[\textsc{Tester}~\text{outputs}~\textsc{No}~\text{when given samples}~\mathcal{S}'~\text{where}~\mathcal{S}'\sim(\mu_{F})^{\otimes L}]+\pi(S\neq S' )\notag\\
	\le{}&  \frac{1}{4} + \frac{1}{8} \le \frac{3}{8}.
	\end{align}	
	Hence, the \textsc{Tester} returns \textsc{Yes} (and our output is correct) with probability at least $5/8$.
	%We note that by Lemma~\ref{les} the sample sequence $\mathcal{S}$ can be generated in $\poly(n)$ time since $L = \poly(n)$.
	
	Now, if $Z_G \ge r \hat{Z}$, then by assumption 
	$
	\TV{\mu_{F}}{\mu_{F^* }} > 1 - \eps.
	$
	Moreover, by~ \eqref{eq:F*-alg} 
	$$
	\TV{(\mu_{F^* })^{\otimes L}}{(\mu_{F^* }^{\textsc{alg}})^{\otimes L}} \leq L \TV{\mu_{F^* }}{\mu_{F^* }^{\textsc{alg}}} \le \frac 18.
	$$
	Thus, analogously to~\eqref{eq:coupling-main} (i.e., using the optimal coupling for $(\mu_{F^* }^{\textsc{alg}})^{\otimes L}$ and $(\mu_{F^* })^{\otimes L}$), we get
	$$
	\Pr\left[  \textsc{Tester}~\text{outputs}~\textsc{Yes}~\text{when given samples}~\mathcal{S}~\text{where}~\mathcal{S}\sim(\mu_{F^* }^{\textsc{alg}})^{\otimes L}\right] \leq \frac{3}{8}.
	$$
	Hence, the \textsc{Tester} returns \textsc{No} with probability at least $5/8$.	Therefore, we can conclude that our algorithm for decision $r$-approximate counting succeeds with probability at least $5/8$.
	The result then follows from the fact that the running time of the algorithm is $\poly(N)$, as each step of the algorithm takes at most $\poly(N)$ time by our assumptions.
\end{proof}

\subsubsection{Proof of Theorem~\ref{thm:Potts-general}}

We can now prove Theorem~\ref{thm:Potts-general} which states hardness of identity testing for the ferromagnetic Potts model on
general graphs. 

\begin{proof}[Proof of Theorem~\ref{thm:Potts-general}]

Consider the ferromagnetic Potts model on an $N$-vertex graph $G = (V_G,E_G)$
with constant edge weight $\beta_G$ in every edge and no external field. 
Let $\hat Z > 0$ be a real number and let $n = N^{10}+N$.
Suppose there is an $\varepsilon$-identity testing algorithm for $\MfPotts(n,n,\beta_G,0)$
with sample complexity $L = L(n) = \poly(n)$ and running time $\poly(n)$.  
Let $r = 96\eps^{-1}\sqrt{\eps L +1}$; our goal is to determine whether $Z_G \le \frac{1}{r} \hat{Z}$ or $Z_G \ge r \hat{Z}$ where $Z_G := Z_{G,\beta_G}$.

We construct the Potts models  $F$ and $F^*$ as describe in Section~\ref{subsubsec:potts-reduce} with corresponding Gibbs distributions $\mu_F$ and $\mu_{F^*}$ using the values of $\beta$ and $\beta_H$ supplied by Lemma~\ref{lem:F-to-TV}; hence the models $F$ and $F^*$ belong to $\MfPotts(n,n,\beta_G,0)$, since $\beta_G > \max\{\beta,\beta_H\}$.

Lemmas~\ref{lem:F-to-TV} ensures that when 
\begin{equation}
\label{eq:Z-bound}
r q e^{\beta_G |E_G|} \le \hat{Z} \le \frac{q^N}{r}  e^{\beta_G |E_G|},
\end{equation}
conditions (i) and (ii) in Theorem~\ref{thm:gen-red} are satisfied.
Moreover, Lemma~\ref{les} gives condition (iii). Thus, Theorem~\ref{thm:gen-red} implies that we have an algorithm for the decision version of $r$-approximate counting for the Potts model on $G$ when $\hat Z$ satisfies~\eqref{eq:Z-bound}.
Meanwhile, we can bound $Z_G$ crudely by
\[
q e^{\beta_G|E_G|} \le Z_G \le q^N {e}^{\beta_G|E_G|}. 
\]
Thus, if $\hat{Z} < r q \exp(\beta_G|E_G|) \le r Z_G$, we can output $\hat{Z} \le \frac{1}{r} Z_G$.
Similarly, when $\hat{Z} > \frac{1}{r} q^N \exp(\beta_G|E_G|) \ge \frac{1}{r} Z_G$ we can output $\hat{Z} \ge r Z_G$.
Therefore, we have a $\poly(N)$ algorithm for the decision version of $r$-approximate counting for $\hatMfPotts(N,N,\beta_G,0)$ where $N =  \Theta(n^{1/10})$, $r = \poly(N)$  and $\beta_G= \Theta(1)$.
The result then follows from Theorem~\ref{thn:decision-cnt:potts}
and the fact that there is no $\fpras$ for $\hatMfPotts(N,N,\beta_G,0)$ unless 
there is one for \#BIS~\cite{GJ,GSVY}. 
\end{proof}

\subsection{Step 3: Degree reduction}\label{Potts-degree}

The following result provides a reduction from identity testing in the family $\MPotts(\hat{n},d,\hat \beta,\hat h)$ to identity testing in $\MPotts(n,n,\beta,{h})$, under some mild assumptions on the model parameters; this allows us to deduce the hardness of identity problem on graphs of bounded degree as
stated in Theorem~\ref{thm:Potts} using the main result Theorem~\ref{thm:Potts-general} from the previous section.

\begin{thm}
	\label{thm:deg-red-main}
	Let $\hat{n},d \in \N^+$ be such that $3 \le d \le \hat{n}^{1-\rho}$ for some constant $\rho \in (0,1)$.
	Suppose that
	%for all sufficiently large $N \in \N$,
	for some constants $\beta,h\ge0$
	%and $0 < h \le n^{\min\{5,4/\rho-1\}} \log N$,
	there is no $\poly(n)$ running time $\varepsilon$-identity testing algorithm for $\MPotts(n,n,\beta,h)$.
	%Then, there exists $\varepsilon_n \rightarrow \varepsilon$ such that for all sufficiently large $n$
	Then there exists a constant $c\in(0,1)$ such that, for any constant $\hat\eps > \eps$
	%and all sufficiently large $\hat{n}$,
	there is no $\poly(\hat{n})$ running time $\hat \eps$-identity testing algorithm for $\MPotts(\hat{n},d,\hat \beta, \hat h)$
	provided $\hat\beta d = \omega(\log \hat{n})$ and
	%$\hat h \le \frac{h}{n^{\min\{5,4/\rho-1\}}}$.
	$\hat h \le h \hat{n}^{-c}$.
%	Moreover, the hardness is true even if we restrict to a subclass of $\MPotts(\hat{n},d,\hat \beta, \hat h)$ that contains only Potts models whose underlying graph is bipartite or if we replace $\MPotts$ by $\MfPotts$, $\MRBM$ or $\MfRBM$.
\end{thm}

This theorem is a special case of our more general result in Theorem~\ref{thm:deg-red}, which 
 we prove in Section~\ref{section:main-proof}. We conclude with the proof of Theorem~\ref{thm:Potts}.
%, which is a direct consequence of Theorems~\ref{thm:Potts-general} and~\ref{thm:deg-red}.

\begin{proof}[Proof of Theorem~\ref{thm:Potts}]
	Follows from Theorems~\ref{thm:Potts-general} and~\ref{thm:deg-red-main}.
\end{proof}

\section{Testing mixed RBMs with no external fields}
\label{RBM-nofield}

In this section, we show that identity testing for RBMs with arbitrary edges interactions is computationally hard,
even in the absence of an external field (i.e., $h=0$); specifically, we prove Theorem~\ref{thm:main-RBM-mixed} from the introduction.
For this, we establish first the hardness of the identity testing problem for antiferromagnetic Ising models with bounded edge interactions. We then reduce this problem to identity testing for mixed RBMs using our degree reduction machinery (see Sections and \ref{Potts-degree} and ~\ref{section:main-proof}) which conveniently also turns our instance into a bipartite graph.

We start by reducing the problem of approximating the partition function of the antiferromagnetic Ising models to identity testing. Hence, the following well-known result concerning the hardness of approximate counting in the antiferromagnetic setting plays an important role for us.

\begin{thm}[\!\cite{SlySun,GSV:Ising}]
	\label{thm:antiferro-Ising-hard}
	Let $d \ge3$ be an integer and let $\beta_0 > \beta_c(d) := \mathrm{arctanh}(1/(d-1))$ be a real number. Then, for a sufficiently large integer $N$, there is no $\fpras$ for the partition function of
	the antiferromagnetic Ising model on $d$-regular $N$-vertex graphs
	with interaction $\beta_0$ on every edge,
	unless $\RP=\NP$.
\end{thm}

The next step in our proof is a reduction from the decision version of approximate counting (see Definition~\ref{def:decision-cnt}) to identity testing.

\begin{thm}\label{thm:reduction-antiferro-Ising}
	Let $\varepsilon \in (0,1)$ be any constant.
	There exists $0 < \beta_0 = O(1)$ such that
	an $\eps$-identity testing algorithm for $\MaIsing(n,n,\beta_0,0)$ with $\poly(n)$ sample complexity and running time
	can be used to solve the decision $r$-approximate counting problem for $\hatMaIsing(N,3,-0.6,0)$ in $\poly(N)$ time, where $N = \Theta(\sqrt{n})$ and $r = \poly(N)$.
\end{thm}

We can now provide the proof of Theorem~\ref{thm:main-RBM-mixed}.

\begin{proof}[Proof of Theorem~\ref{thm:main-RBM-mixed}]
From Theorems~\ref{thm:antiferro-Ising-hard} and~\ref{thn:decision-cnt:potts}, it follows that for any $c>0$ there is no $\poly(N)$ running time algorithm for the decision version of $N^c$-approximate counting for $\hatMaIsing(N,3,-0.6,0)$ unless $\RP=\NP$.
Theorem~\ref{thm:reduction-antiferro-Ising} then implies that, under the same assumption that $\RP=\NP$, there is no
$\eps$-identity testing algorithm for $\MaIsing(n,n,\beta_0,0)$ with $\poly(n)$ sample complexity and running time for constant $\varepsilon \in (0,1)$, $n = \Theta(\sqrt{n})$ and a suitable constant $\beta_0 > 0$. The result then follows from Theorem \ref{thm:deg-red}. 
\end{proof}

We provide in the next section the missing proof of Theorem~\ref{thm:reduction-antiferro-Ising}

\subsection{Reducing counting to testing for the antiferromagnetic Ising model: proof of Theorem~\ref{thm:reduction-antiferro-Ising}}

\noindent
\textbf{Testing instance construction.} \ \ Consider an antiferromagnetic Ising model on an $N$-vertex $3$-regular graph $G = (V_G,E_G)$ with the same inverse temperature parameter $\beta_G = -0.6$ on every edge an no external field.
We provide an algorithm for the decision version of $r$-approximate counting for $Z_G := Z_{G,\beta_G,0}$, using the presumed identity testing algorithm.

Define $F$ to be a graph with the vertex set
\[
V_{F} = V_G \cup \{s_1,s_2\} \cup \left\{u^{(i)}_{v,j}: v\in V_G, 1\le i\le N, j\in\{1,2\} \right\} \cup \left\{ w^{(i)}_j: 1\le i \le N^2,j\in\{1,2\} \right\}
\]
and the edge set
\begin{align*}
E_{F} = E_G &\cup \left\{\{u_{v,j}^{(i)},v\},\{u_{v,j}^{(i)},s_j\}: v\in V_G,1\le i\le N,j\in\{1,2\}\right\}\\
&\cup \left\{ \{w_j^{(i)},s_j\}: 1\le i\le N^2,j\in\{1,2\} \right\}\\
& \cup \left\{ \{w_1^{(i)}, w_2^{(i)}\}: 1\le i\le N^2 \right\};
\end{align*}
\begin{figure}[th]
	\centering
  	\includegraphics[width=0.7\textwidth]{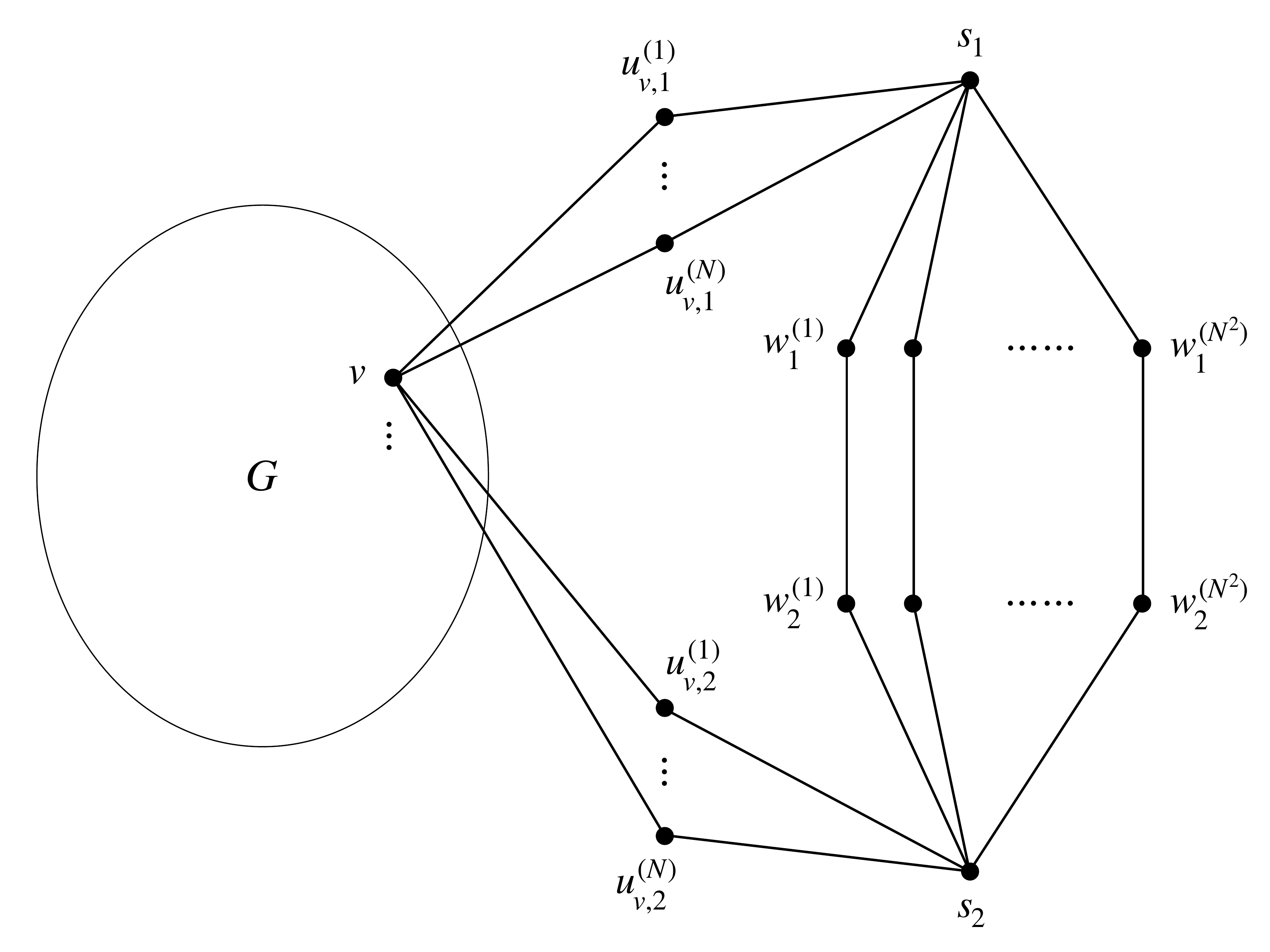}
  	\caption{The graph $F$. For every vertex $v\in V_G$ and $j\in\{1,2\}$, $v$ and $s_j$ are connected by $N$ disjoint paths of length $2$. Also, $s_1$ and $s_2$ are connected by $N^2$ disjoint paths of length $3$.}
  	\label{fig:F-antiferro}
\end{figure}
see Figure~\ref{fig:F-antiferro}. 
Observe that $F$ has $n = 4N^2+N+2$ vertices. 
Given two real numbers $\beta_1,\beta_2 > 0$, we then define an antiferromagnetic Ising model on the graph $F$ as follows:
\begin{enumerate}
\item Every edge $\{u,v\}\in E_G$ has weight $-0.6$.
\item For every $v\in V_G$, $1\le i\le N$ and $j=1,2$, the two edges $\{u_{v,j}^{(i)},v\}$ and $\{u_{v,j}^{(i)},s_j\}$ have weight $-\beta_1$;
\item For every $1\le i\le N^2$, the edges $\{w_1^{(i)},s_1\}$, $\{w_2^{(i)},s_2\}$ and $\{w_1^{(i)}, w_2^{(i)}\}$ have weight $-\beta_2$.
\end{enumerate}
We slight abuse of notation, we use $F$ for the resulting Ising model on $F$ and $\mu := \mu_F$ for the corresponding Gibbs distribution. $F$ will be the visible model of our testing instance. 

For the hidden model $F^*$, we consider the same construction above but replacing $G$ with an independent set $I_N$ on $V_G$. 
Let $\mu^* := \mu_{F^*}$ be the corresponding the Gibbs distribution.
We note first that we can efficiently sample from $\mu^*$.

\begin{lemma}\label{lem:I-sample-antiferro}
There is an exact sampling algorithm for the distribution $\mu^*$ with running time $\poly(n)$. 
\end{lemma}

\begin{proof}%[Proof of Lemma~\ref{lem:I-sample-antiferro}]
Configurations in $\Omega_{F^*}$ can be classified by their type, which is given by the spins of $s_1,s_2$ and the number of spin $1$'s in the independent set $I_N$. 
There are $4(N+1)$ types in total. 
Observe that configurations of each type have the same weight by symmetry, and this weight can be computed efficiently since given the spins of $s_1$, $s_2$ and $I_N$ the remaining graph has only isolated vertices and edges. 
Also it is easy to get the number of configurations of each type. 
Thus, to sample from $\mu^*$, we can first sample a type from the induced distribution on types, and then sample a configuration of the given type uniformly at random. 
\end{proof}

Our hidden and visible models $F$ and $F^*$ are related as follows.

\begin{lemma}\label{lem:Z-to-TV-antiferro}
Let $\varepsilon \in (0,1)$ be a constant, $L = L(n) = \poly(n)$ and $r = 96 \eps^{-1} \sqrt{\eps L+1}$. 
Suppose $\hat Z \in \R$ is such that $r2^N e^{-0.9N} \le \hat{Z} \le \frac{1}{r} 2^N$.
Then, for any $\beta_1 \ge 3$, we can find $0 < \beta_2 < \beta_1+2$ in $\poly(n)$ time such that all of the following holds:
\begin{enumerate}[(i)]
%\label{eqn:D/M-antiferro}
\item $\frac{1}{4\sqrt{\eps L+1}} \frac{Z_G}{\hat{Z}} 
\le 
\frac{Z_{F}^\mathrm{D}}{Z_{F}^\mathrm{M}} 
\le 
\frac{1}{\sqrt{\eps L+1}} \frac{Z_G}{\hat{Z}}$;
\item $\frac{Z_{F^*}^\mathrm{D}}{Z_{F^*}^\mathrm{M}} \le \frac{1}{r \sqrt{\eps L+1}};$
\item If $Z_G \le \frac{1}{r} \hat{Z}$, then
$\TV{\mu}{\mu^*} \le \frac{1}{16L}$;
\item If $Z_G \ge r \hat{Z}$, then
$\TV{\mu}{\mu^*} \ge 1-\eps$. 
\end{enumerate}
\end{lemma}

The proof of Lemma~\ref{lem:Z-to-TV-antiferro} is provided in Section~\ref{subsec:mixed-rbm-lemmas}.
We proceed first with the proof of Theorem~\ref{thm:reduction-antiferro-Ising} which follows along the lines of the proof of Theorem~\ref{thm:Potts-general}.

\begin{proof}[Proof of Theorem~\ref{thm:reduction-antiferro-Ising}]
	
Consider an antiferromagnetic Ising model on an $N$-vertex $3$-regular graph $G = (V_G,E_G)$ with edge weight $\beta_G = -0.6$ on every edge and no external field; note that this model belongs to the family $\hatMaIsing(N,3,-0.6,0)$.
Let $\hat Z > 0$ be a real number, let $n = 4N^2+N+2$ and
suppose there is an $\varepsilon$-identity testing algorithm for $\MaIsing(n,n,\beta_0,0)$
with sample complexity $L = L(n) = \poly(n)$ and running time $\poly(n)$, where $\beta_0 >0$ is a suitable constant we choose later.
Let $r = 96\eps^{-1}\sqrt{\eps L +1}$; we want to check whether $Z_G \le \frac{1}{r} \hat{Z}$ or $Z_G \ge r \hat{Z}$ where $Z_G := Z_{G,\beta_G}$.

We construct the Ising models  $F$ and $F^*$ with Gibbs distribution $\mu$ and $\mu^*$, respectively as described above. We set $\beta_1=3$ and use the $\beta_2$ supplied by Lemma~\ref{lem:Z-to-TV-antiferro}; hence the models $F$ and $F^*$ belong to $\MaIsing(n,n,\beta_0,0)$, provided $\beta_0 \ge \max\{\beta_1,\beta_2\}$.

By Lemma~\ref{lem:Z-to-TV-antiferro} when 
$r 2^N e^{-0.9N} \le \hat{Z} \le \frac{1}{r} 2^N$,
conditions (i) and (ii) from Theorem~\ref{thm:gen-red} are satisfied; condition (iii) is given by Lemma~\ref{lem:I-sample-antiferro}. Hence, we have an algorithm for the decision version of $r$-approximate counting for the Ising model on $G$ for $\hat Z$ in this range.
Otherwise, 
observe that the weight of every configuration is at least $e^{-0.9N}$, which corresponds to the weight of the monochromatic configuration, and at most $1$.
Thus, $2^N e^{-0.9N} \le Z_G \le 2^N$. 
If $\hat{Z} < r 2^N e^{-0.9N} \le r Z_G$, then we can output $\hat{Z} \le \frac{1}{r} Z_G$. Similarly, $\hat{Z} > \frac{1}{r} 2^N \ge \frac{1}{r} Z_G$ and we output $\hat{Z} \ge r Z_G$.

Therefore, we have a $\poly(N)$ running time algorithm for the decision version of $r$-approximate counting for $\hatMaIsing(N,3,-0.6,0)$ where $N =  \Theta(n^{1/2})$ and $r = \poly(N)$, as desired.
\end{proof}

%\begin{claim}\
%	\label{claim:extremal_cases-antiferro}
%	\begin{enumerate}
%		\item If $\hat{Z} < r 2^N e^{-0.9N}$, then $Z_G \ge r \hat{Z}$.
%		\item If $\hat{Z} > \frac{1}{r} 2^N$, then $Z_G \le \frac{1}{r} \hat{Z}$.
%		%\item If $L(n) > e^N$, then computing $Z_G$ exactly by brutal force takes $O(L(n))$ time.
%		%\item Assume that $\eps(n) < e^{-N}$.
%		%If $\hat{Z} \ge 2^N$, then $Z_G \le \frac{1}{r} \hat{Z}$; if $\hat{Z} \le 1$, then $Z_G \ge r \hat{Z}$.
%	\end{enumerate}
%\end{claim}

%\begin{proof}[Proof of Claim~\ref{claim:extremal_cases-antiferro}]
%Observe that the weight of every configuration is at least $e^{-0.9N}$, the weight of the monochromatic configuration, and at most $1$.
%Thus, we have the following crude bounds on $Z_G$:
%\[
%2^N e^{-0.9N} \le Z_G \le 2^N.
%\]
%If $\hat{Z} < r 2^N e^{-0.9N} \le r Z_G$, then we have $\hat{Z} \le \frac{1}{r} Z_G$ by our assumption on the testing instances.
%Similarly, $\hat{Z} > \frac{1}{r} 2^N \ge \frac{1}{r} Z_G$ implies that $\hat{Z} \ge r Z_G$.
%This shows the first and the second claim.
%For the third claim, computing $Z_G$ exactly takes $O(2^N N^2) = O(e^N) = O(L(n))$ time when $L(n) > e^N$.
%Finally, notice that $r\ge 1/\eps \ge e^N$ when $\eps < e^{-N}$.
%Hence, if $Z_G \le \frac{1}{r}\hat{Z}$ then
%\[
%\hat{Z} \ge r Z_G \ge e^N 2^N e^{-0.9N} \ge 2^N.
%\]
%Meanwhile, if $Z_G \ge r\hat{Z}$ then
%\[
%\hat{Z} \le \frac{1}{r} Z_G \le e^{-N} 2^N \le 1.
%\]
%This shows the last claim.
%\end{proof}

\subsection{Proof of Lemma~\ref{lem:Z-to-TV-antiferro}}
\label{subsec:mixed-rbm-lemmas}

Our construction of the visible and hidden models is inspired by our construction in Section~\ref{Potts-proof-general} for the ferromagnetic Potts model. 
In particular, the two vertices $\{s_1,s_2\}$ play the role of the complete graph $H$ in our construction in Section~\ref{subsubsec:potts-reduce}. 
We partition $\Omega_{F} = \{+,-\}^{V_{F}}$ into two disjoint subsets $\Omega_{F} = \Omega_{F}^\mathrm{M} \cup \Omega_{F}^\mathrm{D}$, depending on whether $\sigma(s_1) = \sigma(s_2)$ (the \emph{majority phase}) or $\sigma(s_1) \neq \sigma(s_2)$ (the \emph{disordered phase}); 
more precisely, the set of \emph{majority} configurations is given by
\[
\Omega_{F}^\mathrm{M} = \left\{ \sigma\in\Omega_{F}: \sigma(s_1) = \sigma(s_2) \right\}
\]
and the set of \emph{disordered} configurations is
\[
\Omega_{F}^\mathrm{D} = \left\{ \sigma\in\Omega_{F}: \sigma(s_1) \neq \sigma(s_2) \right\}.
\]
The partition function for the majority phase is defined naturally as
\[
Z_{F}^\mathrm{M} = \sum_{\sigma\in\Omega_{F}^\mathrm{M}} \exp\left( \sum_{\{u,v\}\in E_{F}} \beta_{F}(\{u,v\}) \Ind\{\sigma(u) = \sigma(v)\} \right),
\]
and similarly for $Z_{F}^\mathrm{D}$. 
%where $\wt_{F}(\sigma)$ is the weight of the configuration $\sigma$ for the Ising model on $F$.
%\[
%\wt(\sigma) = \exp\left( \sum_{(u,v)\in E} \beta(u,v) \sigma(u) \sigma(v) + \sum_{v\in V} h(v) \sigma(v) \right)
%\]
Therefore, we have $Z_{F} = Z_{F}^\mathrm{M} + Z_{F}^\mathrm{D}$. 
In the same way, we also define the partition functions $Z_{F^*}^\mathrm{M}$ and $Z_{F^*}^\mathrm{D}$ for the hidden model on the graph $F^*$ (notice that $\Omega_{F^*}^\mathrm{M} = \Omega_{F}^\mathrm{M}$ and $\Omega_{F^*}^\mathrm{D} = \Omega_{F}^\mathrm{D}$). 
%These partition functions are related by the following lemma. 

%Let $\sigma^+\in\Omega_G$ (resp., $\sigma^-\in\Omega_G$) to be the monochromatic configuration where all vertices are labeled $+$ (resp., $-$), and define
%\[
%Z_G^\mathrm{mo} = \sum_{\sigma\in\{\sigma^+,\sigma^-\}} \exp\left( \sum_{(u,v)\in E_G} \beta_G(u,v) + \sum_{v\in V_G} h_G(v) \sigma(v) \right)
%\]
%to be the partition function for monochromatic configurations.
%The restricted partition functions $Z_{F}^i$ and $Z_G$ are related by the following lemma.
%These partition functions $Z_{F}^\mathrm{M}$, $Z_{F}^\mathrm{D}$ and $Z_G$ are related by the following lemmas.

%\begin{lemma}\label{lem:ratio_1/2-antiferro}
%Suppose $\hat Z \in \R$ is such that $r2^N e^{-0.9N} \le \hat{Z} \le \frac{1}{r} 2^N$.
%Then, for any $\beta_1 \ge 3$ we can find $0 < \beta_2 < \beta_1 + 2$ in $\poly(n)$ time such that 
%\[
%\frac{1}{4\sqrt{\eps L+1}} \frac{Z_G}{\hat{Z}} 
%\le 
%\frac{Z_{F}^\mathrm{D}}{Z_{F}^\mathrm{M}} 
%\le 
%\frac{1}{\sqrt{\eps L+1}} \frac{Z_G}{\hat{Z}} 
%\qquad\text{and}\qquad
%\frac{Z_{F^*}^\mathrm{D}}{Z_{F^*}^\mathrm{M}} \le \frac{1}{r %\sqrt{\eps L+1}}.
%\]
%Moreover, such $\beta_2>0$ can be computed in $\poly(n)$ time. 
%\[
%	\frac{1}{4\sqrt{\eps L+1}}\, \frac{e^{-0.9N}}{\hat{Z}}
%	\le
%	\left( \frac{g(\beta_2)}{\cosh \beta_1} \right)^{N^2}
%	\le
%	\frac{1}{\sqrt{\eps L+1}}\, \frac{e^{-0.9N}}{\hat{Z}},
%\]
%If $\beta_2$ satisfies
%\[
%\frac{1}{4\sqrt{\eps L+1}}\, \frac{e^{-0.9N}}{\hat{Z}}
%\le
%\left( \frac{g(\beta_2)}{\cosh \beta_1} \right)^{N^2}
%\le
%\frac{1}{\sqrt{\eps L+1}}\, \frac{e^{-0.9N}}{\hat{Z}},
%\]
%\end{lemma}

\begin{proof}[Proof of Lemma~\ref{lem:Z-to-TV-antiferro}]
Consider the following subset of configurations in $\Omega_F^\mathrm{M}$ given by
\[
\Omega_{F}^{\mathrm{M}_0} = \left\{ \sigma \in \Omega_{F}^\mathrm{M}: \forall v\in V_G, \sigma(v) = \sigma(s_1) \right\}. 
\]
We also define the corresponding partition function $Z_F^{\mathrm{M}_0}$ and $Z_{F^*}^{\mathrm{M}_0}$ in the same way as above. 
We claim that $Z_F^{\mathrm{M}_0}$ (resp., $Z_{F^*}^{\mathrm{M}_0}$) is a good approximation (with only exponentially small error) of the partition function $Z_F^\mathrm{M}$ (resp., $Z_{F^*}^\mathrm{M}$) that we are interested in. 
\begin{claim}\label{claim:Z-and-Z0}
If $\beta_1 \ge 3$, then 
$(1-e^{-2N}) Z_F^\mathrm{M} \le Z_F^{\mathrm{M}_0} \le Z_F^\mathrm{M}$ 
and 
$(1-e^{-2N}) Z_{F^*}^\mathrm{M} \le Z_{F^*}^{\mathrm{M}_0} \le Z_{F^*}^\mathrm{M}$. 
\end{claim}
The proof of the following claim is postponed to the end of the section. 
We then derive explicit formula for $Z_{F}^\mathrm{D}$ and $Z_{F}^{\mathrm{M}_0}$. 
For configurations in $\Omega_{F}^\mathrm{D}$, every spin assignment to the vertices of $G$, $s_1$ and $s_2$ is multiplied by a $2e^{-\beta_1} (e^{-2\beta_1}+1)$ factor, corresponding to the weight of the edges $\{u_{v,j}^{(i)},v\}$, $\{u_{v,j}^{(i)},s_j\}$, $j\in\{1,2\}$ for every vertex $v\in V_G$ and every $1\le i \le N$, 
and by a $3e^{-2\beta_2} + 1$ factor for the edges $\{w_1^{(i)},s_1\}$, $\{w_2^{(i)},s_2\}$ and $\{w_1^{(i)}, w_2^{(i)}\}$ for every $1\le i \le N^2$. 
For configurations in $\Omega_{F}^{\mathrm{M}_0}$, each monochromatic configuration on $G$ is multiplied by a $(e^{-2\beta_1}+1)^2$ factor for the edges $\{u_{v,j}^{(i)},v\}$, $\{u_{v,j}^{(i)},s_j\}$, $j\in\{1,2\}$ for every vertex $v\in V_G$ and every $1\le i \le N$, 
and by $e^{-3\beta_2} + 3e^{-\beta_2}$ for the edges $\{w_1^{(i)},s_1\}$, $\{w_2^{(i)},s_2\}$ and $\{w_1^{(i)}, w_2^{(i)}\}$ for every $1\le i \le N^2$. 
Thus, we obtain
\begin{align*}
Z_{F}^\mathrm{D} &= 2 \left( 3e^{-2\beta_2} + 1 \right)^{N^2} \left( 2 e^{-\beta_1} \left(e^{-2\beta_1} + 1\right) \right)^{N^2} Z_G;\\
Z_{F}^{\mathrm{M}_0} &= 2 \left( e^{-3\beta_2} + 3e^{-\beta_2} \right)^{N^2} \left(e^{-2\beta_1} + 1 \right)^{2N^2} e^{-0.9 N}.
\end{align*}
Let $g(x) = (3e^{-2x}+1)/(e^{-3x}+3e^{-x})$
and recall that $\cosh x = \frac{1}{2}(e^x+e^{-x})$. We then deduce that
\begin{equation}\label{eq:D/M0}
\frac{Z_{F}^\mathrm{D}}{Z_{F}^{\mathrm{M}_0}} = \left( \frac{g(\beta_2)}{\cosh \beta_1} \right)^{N^2} e^{0.9N} Z_G.
\end{equation}

Now for $\beta_1 \ge 3$, we show that we can pick $\beta_2 >0$ such that
\begin{equation}\label{eq:beta2}
\frac{1}{2\sqrt{\eps L+1}} \frac{e^{-0.9N}}{\hat{Z}} \le 
\left(\frac{g(\beta_2)}{\cosh \beta_1}\right)^{N^2} \le 
\frac{1}{\sqrt{\eps L+1}} \frac{e^{-0.9N}}{\hat{Z}}. 
\end{equation}
Such $\beta_2>0$ always exists and satisfies $\beta_2 < \beta_1 + 2$. 
To see this, we note that the function $g(x)$ is a continuous increasing function for $x\ge 0$ with $g(0)=1$ and $g(\infty) = \infty$. 
Since $\hat{Z} \le \frac{1}{r} 2^N$, we get
\begin{align*}
\frac{1}{N^2} \log \left( \frac{1}{4\sqrt{\eps L+1}}\, \frac{e^{-0.9N}}{\hat{Z}} \right) + \log(\cosh \beta_1) &\ge
\frac{1}{N^2} \log \left( \frac{1}{4\sqrt{\eps L+1}}\, r 2^{-N} e^{-0.9N} \right) + \beta_1 - 1\\ 
&\ge
- \frac{2}{N} + 3 - 1 > 0,
\end{align*}
where the second inequality follows from $r/(4\sqrt{\eps L +1}) = 6/\eps \ge 1$.
This shows that 
\[
\left( \frac{1}{2\sqrt{\eps L+1}}\, \frac{e^{-0.9N}}{\hat{Z}} \right)^{\frac{1}{N^2}} \cosh \beta_1 \ge 1
\]
and thus implies the existence of $\beta_2>0$. 
Meanwhile, since $\hat{Z} \ge r 2^N e^{-0.9 N}$ we have
\begin{align*} 
\frac{1}{N^2} \log \left( \frac{1}{\sqrt{\eps L+1}}\, \frac{e^{-0.9N}}{\hat{Z}} \right) + \log(\cosh \beta_1)
\le \frac{1}{N^2} \log \left( \frac{1}{\sqrt{\eps L+1}}\, \frac{1}{r} 2^{-N} \right) + \beta_1 < \beta_1,
\end{align*}
where the second inequality follows from $r \sqrt{\eps L +1} = 96\eps^{-1} (\eps L +1) \ge 1$.
This shows that 
\[
e^{\beta_1} > g(\beta_2) = \frac{3e^{-2\beta_2} + 1}{e^{-3\beta_2} + 3e^{-\beta_2}} \ge \frac{1}{4 e^{-\beta_2}} \ge e^{\beta_2 - 2}
\]
and thus $\beta_2 < \beta_1 + 2$. Finally, we can compute a $\beta_2$ satisfying \eqref{eq:beta2} in $\poly(n)$ time by, for example, the binary search algorithm. 

Combining Claim~\ref{claim:Z-and-Z0} and equations \eqref{eq:D/M0} and \eqref{eq:beta2}, we deduce that
\[
\frac{1}{4\sqrt{\eps L+1}} \frac{Z_G}{\hat{Z}} 
\le (1-e^{-2N})\frac{Z_F^\mathrm{D}}{Z_F^{\mathrm{M}_0}} 
\le \frac{Z_F^\mathrm{D}}{Z_F^\mathrm{M}} 
\le \frac{Z_F^\mathrm{D}}{Z_F^{\mathrm{M}_0}} 
\le \frac{1}{\sqrt{\eps L+1}} \frac{Z_G}{\hat{Z}}. 
\]
This shows the first part of the lemma.
For part (ii), we can compute $Z_{F^*}^\mathrm{D}$ and $Z_{F^*}^{\mathrm{M}_0}$ in a similar fashion and obtain
\begin{align*}
Z_{F^*}^\mathrm{D} &= 2 \left( 3e^{-2\beta_2} + 1 \right)^{N^2} \left( 2 e^{-\beta_1} \left(e^{-2\beta_1} + 1\right) \right)^{N^2} 2^N;\\
Z_{F^*}^{\mathrm{M}_0} &= 2 \left( e^{-3\beta_2} + 3e^{-\beta_2} \right)^{N^2} \left(e^{-2\beta_1} + 1 \right)^{2N^2}.
\end{align*}
This gives 
\begin{equation}\label{eq:F*-D/M0}
\frac{Z_{F^*}^\mathrm{D}}{Z_{F^*}^{\mathrm{M}_0}} = \left( \frac{g(\beta_2)}{\cosh \beta_1} \right)^{N^2} 2^N.
\end{equation}
Therefore, by equations~\eqref{eq:F*-D/M0} and \eqref{eq:beta2} we obtain
\[
\frac{Z_{F^*}^\mathrm{D}}{Z_{F^*}^\mathrm{M}} 
\le \frac{Z_{F^*}^\mathrm{D}}{Z_{F^*}^{\mathrm{M}_0}} 
\le \frac{1}{\sqrt{\eps L+1}}\, \frac{e^{-0.9N}}{\hat{Z}}\, 2^N \le \frac{1}{r\sqrt{\eps L+1}},
\] 
where the last inequality follows from the assumption $\hat{Z} \ge r2^N e^{-0.9N}$;
thus, part (ii) follows.

Next, we derive parts (iii) and (iv). 
We define $\nu = \mu(\,\cdot\, | \Omega_{F}^\mathrm{M})$ to be the distribution conditioned on $\Omega_{F}^\mathrm{M}$, and similarly $\nu^* = \mu^*(\,\cdot\, | \Omega_{F^*}^\mathrm{M})$. 
By the definition of total variation distance we have
\[
\TV{\mu}{\nu} = \TV{\mu}{\mu(\,\cdot\, | \Omega_{F}^\mathrm{M})} = \frac{Z_{F}^\mathrm{D}}{Z_{F}} = 1 - \frac{Z_{F}^\mathrm{M}}{Z_{F}}.
\]

For part (iii), if $Z_G \le \frac{1}{r} \hat{Z}$, then we deduce from part (i) that
\[
\TV{\mu}{\nu} \le \frac{Z_{F}^\mathrm{D}}{Z_{F}^\mathrm{M}} 
\le \frac{1}{\sqrt{\eps L+1}} \frac{Z_G}{\hat{Z}} 
\le \frac{1}{r\sqrt{\eps L+1}} = \frac{\eps}{96(\eps L + 1)} \le \frac{1}{96L}.
\]
Similarly, part (ii) implies
\[
\TV{\mu^*}{\nu^*} \le \frac{Z_{F^*}^\mathrm{D}}{Z_{F^*}^\mathrm{M}} 
\le \frac{1}{r \sqrt{\eps L+1}} = \frac{\eps}{96(\eps L + 1)} \le \min\left\{\frac{1}{96L}, \frac{\eps}{96}\right\}. 
\]
Let $\rho = \nu(\,\cdot\, | \Omega_{F}^{\mathrm{M}_0})$ denote the conditional distribution of $\nu$ on $\Omega_F^{\mathrm{M}_0}$. 
Observe that $\rho$ does not depend on the graph $G$, because we condition on the event that all vertices from $G$ receive the same spin, and thus the structure of $G$ does not affect the conditional distribution $\rho$.
In particular, we have $\rho = \nu(\,\cdot\, | \Omega_{F}^{\mathrm{M}_0}) = \nu^*(\,\cdot\, | \Omega_{F^*}^{\mathrm{M}_0})$. 
Then, Claim~\ref{claim:Z-and-Z0} implies that
\[
\TV{\nu}{\rho} = 1 - \frac{Z_F^{\mathrm{M}_0}}{Z_F^\mathrm{M}} \le e^{-2N}
\]
and similarly $\TV{\nu^*}{\rho} \le e^{-2N}$. 
Therefore, we obtain from the triangle inequality that
\[
\TV{\nu}{\nu^*} \le \TV{\nu}{\rho} + \TV{\nu^*}{\rho} \le 2 e^{-2N}. 
\]
We conclude again from the triangle inequality that
\[
\TV{\mu}{\mu^*} \le \TV{\mu}{\nu} + \TV{\mu^*}{\nu^*} + \TV{\nu}{\nu^*} \le \frac{1}{96L} + \frac{1}{96L} + 2 e^{-2N} \le \frac{1}{16L}. 
\]

Finally, for part (iv), if $Z_G \ge r \hat{Z}$, then by part (ii) we have
\[
\TV{\mu}{\nu} \ge 1 - \frac{Z_{F}^\mathrm{M}}{Z_{F}^\mathrm{D}} \ge 1 - 4\sqrt{\eps L+1}\, \frac{\hat{Z}}{Z_G} \ge 1 - \frac{4}{r}\sqrt{\eps L+1} = 1 - \frac{\eps}{24}.
\]
Hence, 
\[
\TV{\mu}{\mu^*} \ge \TV{\mu}{\nu} - \TV{\mu^*}{\nu^*} - \TV{\nu}{\nu^*} \ge 1 - \frac{\eps}{24} - \frac{\eps}{96} - 2e^{-2N} \ge 1-\eps,
\]
as claimed. 
\end{proof}

\begin{proof}[Proof of Claim~\ref{claim:Z-and-Z0}]
%\begin{align*}
%\Omega_{F}^{\mathrm{M}_0} &= \left\{ \sigma\in\Omega_{F}: \sigma(s_1) = \sigma(s_2) ~\text{and}~ \forall v\in V_G, \sigma(v) = \sigma(s_1) \right\};\\
%\Omega_{F}^{\mathrm{M}_1} &= \left\{ \sigma\in\Omega_{F}: \sigma(s_1) = \sigma(s_2) ~\text{and}~ \exists v\in V_G, \sigma(v) \neq \sigma(s_1) \right\}.
%\end{align*}
%\begin{lemma}\label{lem:ratio_3/2-antiferro}
%If $\beta_1 \ge 3$, then $Z_{F}^{\mathrm{M}_1} \le e^{-2N} (Z_{F}^{\mathrm{M}_0} + Z_{F}^{\mathrm{M}_1})$. 
%and $Z_{I'}^{\mathrm{M}_1} \le e^{-2N} (Z_{I'}^{\mathrm{M}_0} + Z_{I'}^{\mathrm{M}_1})$.
%\end{lemma}
For the first inequality, 
note that $Z_F^{\mathrm{M}_0} \le Z_F^\mathrm{M}$. 
A union bound implies
%First we sample a configuration conditioned on that $\sigma(s_1) = \sigma(s_2)$ and then we resample all other vertices. Then the probability that there exists a vertex $v$ having the opposite spin as $s_1$ and $s_2$ is
\begin{align*}
1 - \frac{Z_{F}^{\mathrm{M}_0}}{Z_{F}^\mathrm{M}}
&= \Pr\Big(\exists v \in V_G: \sigma(v) \neq \sigma(s_1) \Big| \sigma(s_1) = \sigma(s_2)\Big)\le \sum_{v\in V_G} \Pr(\sigma(v) \neq \sigma(s_1) | \sigma(s_1) = \sigma(s_2)). 
\end{align*}
For every $\sigma \in \Omega_{F}^\mathrm{M}$ and $v\in V_G$, if $\sigma(v) \neq \sigma(s_1)$, then the total weight of edges incident to $v$ is at most $(2 e^{-\beta_1})^{2N}$; 
and if $\sigma(v) = \sigma(s_1)$, then it is at least $(e^{-2\beta_1} + 1)^{2N} \exp(\beta_G \deg_G(v)) \ge (e^{-2\beta_1} + 1)^{2N} e^{-1.8}$. 
Thus, we get
\begin{align*}
\Pr(\sigma(v) \neq \sigma(s_1) | \sigma(s_1) = \sigma(s_2)) 
&\le \frac{(2e^{-\beta_1})^{2N}}{(2e^{-\beta_1})^{2N} + (e^{-2\beta_1}+1)^{2N} e^{-1.8}}\\
&\le e^{1.8} \left(\frac{2 e^{-\beta_1}}{e^{-2\beta_1}+1}\right)^{2N} \le 10 e^{-2(\beta_1-1)N} \le 10 e^{-4N},
\end{align*}
where the last inequality follows from the assumption $\beta_1 \ge 3$. 
Therefore,
\[
\frac{Z_{F}^{\mathrm{M}_0}}{Z_{F}^\mathrm{M}} \ge 1 - 10N e^{-4N} \ge 1 - e^{-2N}. 
\]
The bound for $F^*$ can be derived analogously.
\end{proof}

%\subsection{Hardness of testing $(\bmax,d)$-bounded RBMs when $\bmax d = \omega(\log n)$}

%In this section we prove Theorem~\ref{thm:main-RBM-mixed}.
%\begin{thm}[\textcolor{Green}{Green arrow}]
%Let $d\ge 3$ be an integer. Suppose that we are given an identity testing algorithm for all $(\bmax,d)$-bounded RBMs when $\bmax d=\omega(\log n)$, and that the testing algorithm requires $L(n)$ samples and runs in $T(n)$ time. Then there is an identity testing algorithm for all $n$-vertex $(\beta_0,n)$-bounded antiferromagnetic Ising models that uses $\poly(n,L(n))$ samples and runs in $\poly(n,T(n))$ time.
%\end{thm}

%\noindent
%Use gadgets.

%\newpage
\section{Testing ferromagnetic RBMs with inconsistent fields}
\label{RBM-field}

In this section, we establish our lower bound for identity testing for ferromagnetic RBMs with inconsistent fields; specifically, we prove Theorem~\ref{thm:main-RBM-ferro} from the introduction. Let us formally define first the notions of consistent and inconsistent external fields. 

\begin{defn}
	\label{dfn:consistent}
	Consider an Ising model on a graph $G=(V_G,E_G)$ with external field $h_G: V_G \times \{1,2\} \rightarrow \R$. We say that the external field $h_G$ is \emph{consistent} if~~$\forall v \in V_G$, $h_G(v,1) \ge 0$ and $h_G(v,2) = 0$ or~$\forall v \in V_G$, $h_G(v,2) \ge 0$ and $h_G(v,1) = 0$. 
	\end{defn}
%(Recall that we say that the external field is inconsistent when it is allowed to favor different spins for different vertices.)
We use once again our reduction strategy from $r$-approximate counting to testing.
We start from the following well-known result. 

\begin{thm}[\!\cite{GJfield}]
		\label{thm:inconsistent-Ising-hard}
		There is no $\fpras$ for the partition function of ferromagnetic Ising models with inconsistent fields, unless \#BIS admits an $\fpras$.
\end{thm}

The next step is the reduction from the decision version of approximate counting to identity testing.

\begin{thm}\label{thm:reduction-ferro-Ising}
	Let $\varepsilon \in (0,1)$ be any constant.
	For every $\hat{\beta}, \hat{h}>0$
	there exist $\beta_0, h_0>0$ such that an $\eps$-identity testing algorithm for $\MfIsing(n,n,\beta_0,h_0)$ with $\poly(n)$ sample complexity and running time
	can be used to solve the decision $r$-approximate counting problem for 
	$\hatMfIsing(N,N,\hat{\beta},\hat{h})$ in $\poly(N)$ time, where $N = \Theta(\sqrt{n})$ and $r = \poly(N)$.
\end{thm}

We can now provide the proof of Theorem~\ref{thm:main-RBM-ferro}.

\begin{proof}[Proof of Theorem~\ref{thm:main-RBM-ferro}]
Follows from Theorems~\ref{thm:inconsistent-Ising-hard},~\ref{thn:decision-cnt:potts}, \ref{thm:reduction-ferro-Ising} and \ref{thm:deg-red}. 
\end{proof}

\subsection{Reducing counting to testing for the ferromagnetic Ising model with an inconsistent field: proof of Theorem~\ref{thm:reduction-ferro-Ising}}

\noindent\textbf{Testing instance construction.}\ \ Consider an instance $(G,\beta_G,h_G)$ of ferromagnetic Ising models with an inconsistent field,
where $G = (V_G,E_G)$ is the underlying graph with $N = |V_G|$, $\beta_G(e) =  \hat{\beta}>0$ for every $e \in E_G$, and
at every vertex the external field is either $h_G = (\hat{h},0)$ or $h_G = (0,\hat{h})$ for $\hat{h}>0$; that is, $\forall v\in V_G,~h_G(v,j) = \1(j=1) \hat h$ for $j=\{1,2\}$ or $h_G(v,j) = \1(i=2) \hat h$ for $i=\{1,2\}$. Note that for consistency with the notation in the previous sections we use spins $\{1,2\}$ for the Ising model, instead of the usual ``$+$'' and ``$-$'' spins.
Our goal is to give a $r$-approximate counting algorithm for the partition function $Z_G := Z_{G,\beta_G,h_G}$ for some $r = \poly(N)$ using an identity testing algorithm.

Define $F$ to be a graph with the vertex set
\[
V_{F} = V_G \cup \{s_1,s_2\} \cup \left\{u^{(i)}_{v,j}: v\in V_G, 1\le i\le N, j\in\{1,2\} \right\} \cup \left\{ w^{(i)}_j: 1\le i \le N^2, j\in\{1,2\} \right\}
\]
and the edge set
\begin{align*}
E_{F} ={}& E_G \cup \left\{\{u_{v,j}^{(i)},v\},\{u_{v,j}^{(i)},s_j\}: v\in V_G,1\le i\le N, j\in\{1,2\}\right\}\\ 
&\cup \left\{ \{w_j^{(i)},s_j\}: 1\le i\le N^2, j\in\{1,2\} \right\};
\end{align*}
\begin{figure}[th]
	\centering
  	\includegraphics[width=0.7\textwidth]{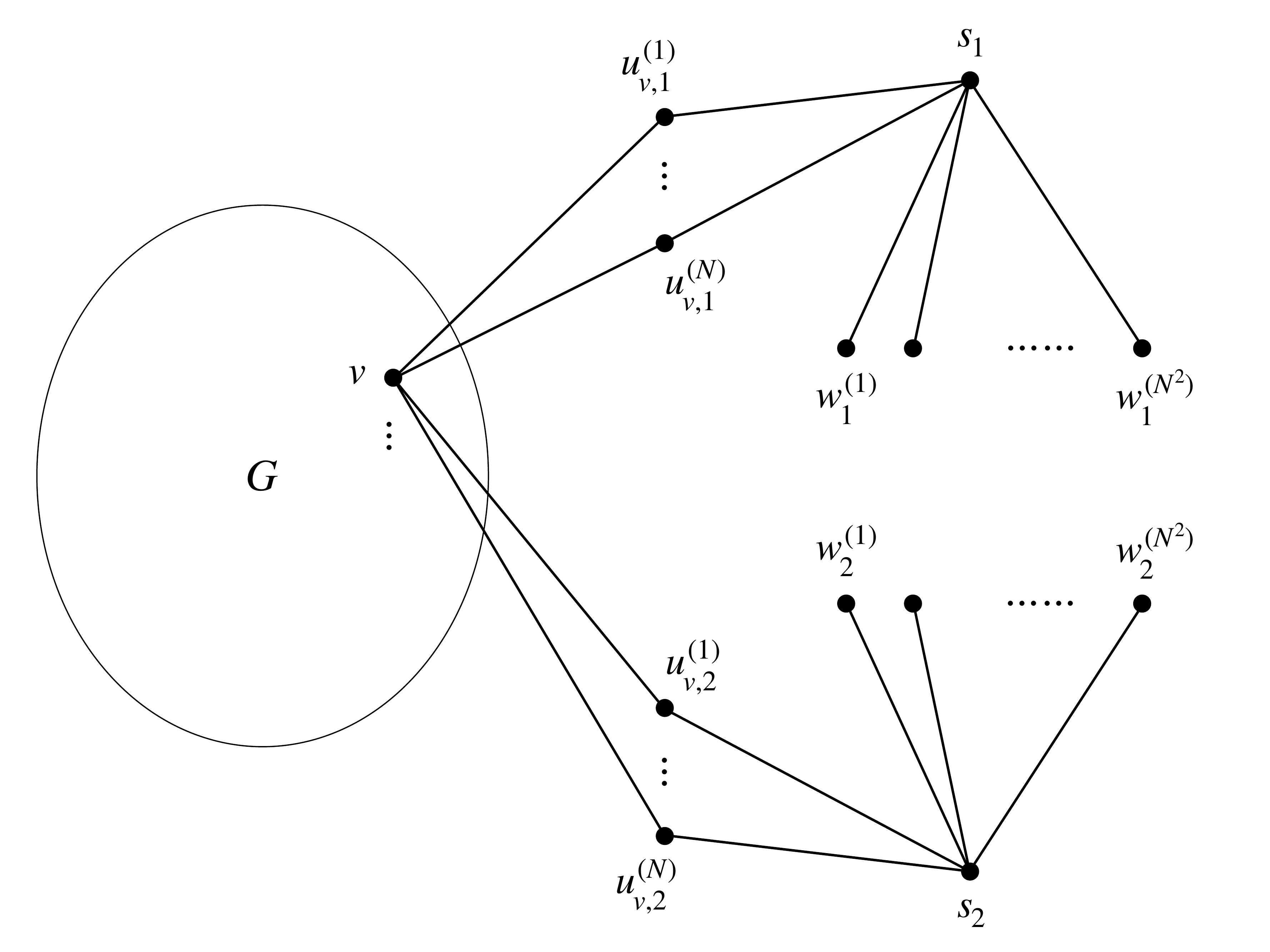}
  	\caption{The graph $F$. For every vertex $v\in V_G$ and $j\in\{1,2\}$, $v$ and $s_j$ are connected by $N$ disjoint paths of length $2$. Also, each of $s_1$ and $s_2$ is adjacent to $N^2$ vertices with nonzero fields.}
  	\label{fig:F-ferro}
\end{figure}
see Figure~\ref{fig:F-ferro}.

Given three real numbers $\beta_1,\beta_2,h > 0$, we then define a ferromagnetic Ising model on the graph $F$ as follows:
\begin{enumerate}
\item Every edge $\{u,v\}\in E_G$ has weight $\hat\beta$ and every vertex $v\in V_G$ has external field given by $h_G$.
\item For every $v\in V_G$, $1\le i\le N$ and $j\in\{1,2\}$, the two edges $\{u_{v,j}^{(i)},v\}$ and $\{u_{v,j}^{(i)},s_j\}$ have weight $\beta_1$;
\item For every $1\le i\le N^2$ and $j\in\{1,2\}$, the edge $\{w_j^{(i)},s_j\}$ has weight $\beta_2$;
\item For every $1\le i\le N^2$, the vertex $w_1^{(i)}$ has external field $(h,0)$ and the vertex $w_2^{(i)}$ has external field $(0,h)$; that is, $h_F(w_1^{(i)},j) = \1(j=1)h$ and $h_F(w_2^{(i)},j) = \1(j=2)h$.
\end{enumerate}
Thus, $F$ is a graph on $n = 4N^2+N+2$ vertices and the Ising model on $F$ is ferromagnetic with an inconsistent external field. Let $\mu := \mu_F$ denote the corresponding Gibbs distribution.

For the hidden model $F^*$, we consider the same construction above but replacing $G$ with a complete graph $K = K_N$ on $N$ vertices where
every edge has weight $\beta_K = \hat{\beta}+4\log 2>0$ and every vertex has the same field $h_G$ as the Ising model on $G$.
Let $\mu^* := \mu_{F^*}$ be the corresponding the Gibbs distribution.
We note first that we can efficiently sample from $\mu^*$.

\begin{lemma}\label{lem:K-sample}
	There is an exact sampling algorithm for the distribution $\mu^*$ with running time $\poly(n)$. 
\end{lemma}

\begin{proof}%[Proof of Lemma~\ref{lem:K-sample}]
Configurations in $\Omega_{F^*}$ are classified by their type, which is given by the spins of the vertices $s_1$ and $s_2$ and the number of vertices with spin $1$ in the complete graph $K_N$. 
There are $4(N+1)$ types in total. 
Observe that configurations of each type have the same weight by symmetry, and this weight can be computed efficiently since given the spins of $s_1$, $s_2$ and $K_N$ the remaining graph has only isolated vertices and edges. 
It is also straightforward to get the number of configurations of each type. 
Thus, to sample from $\mu^*$, we first sample a type from the induced distribution on the types, and then sample a configuration of the given type uniformly at random. 
\end{proof}

Denote the sum of weights of two monochromatic configurations on $G$ by
\[
Z_G^\mathrm{mo} := \sum_{i \in \{1,2\}} \exp \left( \hat{\beta} |E_G| + \sum_{v\in V_G} h_G(v,i) \right). 
%= \sum_{\tau = +,-} \exp\left( \sum_{(u,v)\in E_G} \beta_G(u,v) + \sum_{v\in V_G} h_G(v,\tau) \right)
\]
The hidden and visible models $F$ and $F^*$ are related as follows.

\begin{lemma}\label{lem:Z-to-TV}
	Let $\varepsilon \in (0,1)$ be a constant, $L = L(n) = \poly(n)$ and $r = 96 \eps^{-1} \sqrt{\eps L+1}$. 
	Suppose $\hat Z \in \R$ is such that $rZ_G^\mathrm{mo} \le \hat{Z} \le \frac{1}{r} \exp(\frac{1}{2} (\hat{\beta} + \hat{h} + 1) N^2 )$.
	Then, for any $\beta_1 \ge \frac{1}{2}(\hat \beta + \hat h + 5)$, 
	we can find $\beta_2 \in(0,\beta_1)$ in $\poly(n)$ time such that by setting $h = \beta_2$ al of the following holds:
	\begin{enumerate}[(i)]
  %\label{eqn:D/M-ferro}
	\item $\frac{1}{4\sqrt{\eps L+1}} \frac{Z_G}{\hat{Z}} \le \frac{Z_{F}^\mathrm{D}}{Z_{F}^\mathrm{M}} \le \frac{1}{\sqrt{\eps L+1}} \frac{Z_G}{\hat{Z}}$;
	\item $\frac{Z_{F^*}^\mathrm{D}}{Z_{F^*}^\mathrm{M}} \le \frac{2}{r \sqrt{\eps L+1}}$;
	\item If $Z_G \le \frac{1}{r} \hat{Z}$, then
		$\TV{\mu}{\mu^*} \le \frac{1}{16L}$;
	\item If $Z_G \ge r \hat{Z}$, then
		$\TV{\mu}{\mu^*} \ge 1-\eps$. 
	\end{enumerate}
\end{lemma}

The proof of Lemma~\ref{lem:Z-to-TV} is provided in Section~\ref{subsec:inconsistent-aux}.
We provide next the proof of Theorem~\ref{thm:reduction-ferro-Ising}.

\begin{proof}[Proof of Theorem~\ref{thm:reduction-ferro-Ising}]
Consider the ferromagnetic Ising model $(G,\beta_G,h_G)$, where 
$G=(V_G,E_G)$ is an $N$-vertex graph, $\beta_G(e) = \hat \beta$ for all $e \in E_G$
and $h_G(v,j) = \1(j=1) \hat h$ for $j=\{1,2\}$ or $h_G(v,j) = \1(j=2) \hat h$ for $i=\{1,2\}$ for all $v\in V_G$; note that this model belongs to $\hatMfIsing(N,N,\hat \beta ,\hat h)$.
Let $\hat Z > 0$ be a real number, let $n = 4N^2+N+2$ and
suppose there is an $\varepsilon$-identity testing algorithm for $\MfIsing(n,n,\beta_0,h_0)$
with sample complexity $L = L(n) = \poly(n)$ and running time $\poly(n)$, where $\beta_0,h_0 >0$ are a suitable constants.
Let $r = 96\eps^{-1}\sqrt{\eps L +1}$; we want to check whether $Z_G \le \frac{1}{r} \hat{Z}$ or $Z_G \ge r \hat{Z}$ where $Z_G := Z_{G,\beta_G,h_G}$.

We construct the Ising models  $F$ and $F^*$ with Gibbs distribution $\mu$ and $\mu^*$, respectively as described above, setting $\beta_1= \frac{1}{2}(\hat \beta + \hat h + 5)$, using the $\beta_2$ supplied by Lemma~\ref{lem:Z-to-TV}, and taking $h = \beta_2$; hence the models $F$ and $F^*$ belong to $\MfIsing(n,n,\beta_0,h_0)$, provided $\beta_0 \ge \max\{\hat\beta,\beta_K,\beta_1,\beta_2\}$ and $h_0 \ge \max \{\hat h,h\}$.

By Lemma~\ref{lem:Z-to-TV} when 
$rZ_G^\mathrm{mo} \le \hat{Z} \le \frac{1}{r} \exp(\frac{1}{2} (\hat{\beta} + \hat{h} + 1) N^2 )$,
conditions (i) and (ii) of Theorem~\ref{thm:gen-red} are satisfied; condition (iii) is given by Lemma~\ref{lem:K-sample}. Therefore, we have an algorithm for the decision version of $r$-approximate counting for the Ising model on $G$ for $\hat Z$ in this range.
When $\hat Z$ is not in this range, 
note that  we have the following crude bounds on $Z_G$:
\[
Z_G^\mathrm{mo} \le Z_G \le 2^N \cdot \exp\left( \hat{\beta} \frac{N^2}{2} + \hat{h} N \right) \le \exp\left( \frac{1}{2} (\hat{\beta} + \hat{h} + 1) N^2 \right).
\]
Thus, if $\hat{Z} < r Z_G^\mathrm{mo} \le r Z_G$ we can output $\hat{Z} \le \frac{1}{r} Z_G$. 
Similarly, $\hat{Z} > \frac{1}{r} \exp(\frac{1}{2} (\hat{\beta} + \hat{h} + 1) N^2) \ge \frac{1}{r} Z_G$ we output $\hat{Z} \ge r Z_G$.

Therefore, we have a $\poly(N)$ running time algorithm for the decision version of $r$-approximate counting for $\hatMfIsing(N,N,\hat \beta,\hat h)$ where $N =  \Theta(n^{1/2})$ and $r = \poly(N)$, as desired.
\end{proof}

\subsection{Proof of Lemma~\ref{lem:Z-to-TV}}
\label{subsec:inconsistent-aux}
We reuse the notation introduce in Section~\ref{subsec:mixed-rbm-lemmas}. Recall that 
$\Omega_F^\mathrm{M} = \{ \sigma\in\Omega_F: \sigma(s_1) = \sigma(s_2) \}$ and
$ \Omega_F^\mathrm{D} = \{ \sigma\in\Omega_F: \sigma(s_1) \neq \sigma(s_2) \}$.
Also the partition function for the majority phase is given by
\[
Z_F^\mathrm{M} = \sum_{\sigma\in \Omega_F^\mathrm{M}} \exp\left( \sum_{\{u,v\}\in E_F} \beta_F(\{u,v\}) \Ind\{\sigma(u) = \sigma(v)\} + \sum_{v\in V_F} h_F(v,\sigma(v)) \right)
\]
and $Z_F^\mathrm{D}$ is defined similarly. 
The corresponding partition functions for the hidden model are denoted by $Z_{F^*}^\mathrm{M}$ and $Z_{F^*}^\mathrm{D}$.

\begin{proof}
Let
$
\Omega_{F}^{\mathrm{M}_0} = \left\{ \sigma \in \Omega_{F}^\mathrm{M}: \forall v\in V_G, \sigma(v) = \sigma(s_1) \right\} 
$
and consider restrictions of partition functions $Z_F^{\mathrm{M}_0}$ and $Z_{F^*}^{\mathrm{M}_0}$, as in the proof of Lemma~\ref{lem:Z-to-TV-antiferro}. 
The following claim, whose proof is provided at the end of the section, has the same flavor as Claim~\ref{claim:Z-and-Z0}. 

\begin{claim}\label{claim:Z-and-Z0-ferro-RBM}
If $\beta_1 \ge \frac{1}{2}(\hat{\beta} + \hat{h} + 5)$, 
then 
$(1-e^{-2N}) Z_F^\mathrm{M} \le Z_F^{\mathrm{M}_0} \le Z_F^\mathrm{M}$ 
and 
$(1-e^{-2N}) Z_{F^*}^\mathrm{M} \le Z_{F^*}^{\mathrm{M}_0} \le Z_{F^*}^\mathrm{M}$. 
\end{claim}

We then derive explicit formulae for the two partition functions $Z_{F}^\mathrm{D}$ and $Z_{F}^{\mathrm{M}_0}$.
For configurations $\sigma \in \Omega_{F}^\mathrm{D}$ with $\sigma(s_1) = 1$ and $\sigma(s_2) = 2$ (resp., $\sigma(s_1) = 2$ and $\sigma(s_2) = 1$), the weight of the configuration on $G$ is multiply by a factor of $2e^{\beta_1}(e^{2\beta_1} + 1)$ for each edge $\{u_{v,j}^{(i)},v\}$, $\{u_{v,j}^{(i)},s_j\}$, $j\in\{1,2\}$ for every $v\in V_G$ and every $1\le i \le N$; it is also multiply by a 
$(e^{\beta_2+h}+1)^2$ (resp., $(e^{\beta_2} + e^h)^2$) factor for each edge $\{w_j^{(i)},s_j\}$ and the vertex $w_j^{(i)}$, $j=\{1,2\}$ and $1\le i\le N^2$. 
For configurations in $\Omega_F^{\mathrm{M}_0}$, both monochromatic configurations on $G$ receive additional weight $(e^{2\beta_1}+1)^2$ for the edges $\{u_{v,j}^{(i)},v\}$, $\{u_{v,j}^{(i)},s_j\}$, $j=\{1,2\}$ for every vertex $v\in V_G$ and every $1\le i \le N$, 
and a $(e^{\beta_2+h}+1)(e^{\beta_2} + e^h)$ factor for each edge $\{w_j^{(i)},s_j\}$ and the vertex $w_j^{(i)}$, $j=\{1,2\}$ for every $1\le i\le N^2$. 
Thus, we obtain that
\begin{align*}
Z_{F}^\mathrm{D} &= \left[ \left( e^{\beta_2+h} + 1 \right)^{2N^2} + \left( e^{\beta_2} + e^h \right)^{2N^2} \right] \left( 2 e^{\beta_1} \left(e^{2\beta_1} + 1\right) \right)^{N^2} Z_G;\\
Z_{F}^{\mathrm{M}_0} &= \left( e^{\beta_2+h} + 1 \right)^{N^2} \left( e^{\beta_2} + e^h \right)^{N^2} \left(e^{2\beta_1} + 1 \right)^{2N^2} Z_G^\mathrm{mo}.
\end{align*}
Recall that $\cosh x = \frac{1}{2}(e^x+e^{-x})$. We then deduce that
\begin{align*}
\frac{Z_{F}^\mathrm{D}}{Z_{F}^{\mathrm{M}_0}} = \left[ \left( \frac{\cosh(\frac{\beta_2+h}{2})}{\cosh(\frac{\beta_2-h}{2})} \right)^{N^2} + \left( \frac{\cosh(\frac{\beta_2-h}{2})}{\cosh(\frac{\beta_2+h}{2})} \right)^{N^2} \right] \left( \frac{1}{\cosh \beta_1} \right)^{N^2} \frac{Z_G}{Z_G^\mathrm{mo}}.
\end{align*}
Since $\cosh x \ge 1$ for all $x\in\R$, let $h=\beta_2 > 0$ and then we get
\begin{equation}\label{eq:D/M0-ferro}
\left( \frac{\cosh \beta_2}{\cosh \beta_1} \right)^{N^2} \frac{Z_G}{Z_G^\mathrm{mo}}
\le \frac{Z_{F}^\mathrm{D}}{Z_{F}^{\mathrm{M}_0}}
\le 2\left( \frac{\cosh \beta_2}{\cosh \beta_1} \right)^{N^2} \frac{Z_G}{Z_G^\mathrm{mo}}.
\end{equation}

Now for $\beta_1 \ge \frac{1}{2}(\hat \beta + \hat h + 5)$, we pick $\beta_2 > 0$ such that
\begin{equation}\label{eq:beta2-ferro}
\frac{1}{3\sqrt{\eps L+1}}\, \frac{Z_G^\mathrm{mo}}{\hat{Z}}
\le
\left( \frac{\cosh \beta_2}{\cosh \beta_1} \right)^{N^2}
\le
\frac{1}{2\sqrt{\eps L+1}}\, \frac{Z_G^\mathrm{mo}}{\hat{Z}},
\end{equation}
Such $\beta_2>0$ always exists and satisfies $\beta_2 < \beta_1$. 
To see this, we note that since  $\hat{Z} \le \frac{1}{r} \exp(\frac{1}{2} (\hat{\beta} + \hat{h} + 1) N^2 )$ and $Z_G^\mathrm{mo} \ge 2$, we have
\begin{align*}
\frac{1}{N^2} \log \left( \frac{1}{3\sqrt{\eps L+1}}\, \frac{Z_G^\mathrm{mo}}{\hat{Z}} \right) + \log(\cosh \beta_1)
&\ge
\frac{1}{N^2} \log \left( \frac{1}{3\sqrt{\eps L+1}}\, 2r e^{ - \frac{1}{2} (\hat{\beta} + \hat{h} + 1) N^2} \right) + \beta_1 - 1\\
&\ge 
- \frac{1}{2} (\hat{\beta} + \hat{h} + 1) + \frac{1}{2}(\hat{\beta}+\hat{h}+5) - 1 = 1 > 0,
\end{align*}
where the second inequality follows from $2r/(3\sqrt{\eps L +1}) = 64/\eps \ge 1$. 
This is equivalent to  
\[
\left( \frac{1}{3\sqrt{\eps L+1}}\, \frac{Z_G^\mathrm{mo}}{\hat{Z}} \right)^{\frac{1}{N^2}} \cosh \beta_1 \ge 1,
\]
and hence $\beta_2>0$ satisfying \eqref{eq:beta2-ferro} always exists and can be computed in $\poly(n)$ time. 
Note also that since $\hat{Z} \ge r Z_G^\mathrm{mo}$ we have
\[
\frac{1}{2\sqrt{\eps L+1}}\, \frac{Z_G^\mathrm{mo}}{\hat{Z}} \le \frac{1}{2r\sqrt{\eps L+1}} = \frac{\eps}{192(\eps L+1)} < 1.
\]
This shows that $\cosh \beta_2/\cosh \beta_1 < 1$ and thus $\beta_2 < \beta_1$. 

Combining Claim~\ref{claim:Z-and-Z0-ferro-RBM} and inequalities~\eqref{eq:D/M0-ferro} and \eqref{eq:beta2-ferro}, we deduce that
\[
\frac{1}{4\sqrt{\eps L+1}} \frac{Z_G}{\hat{Z}} 
\le (1-e^{-2N})\frac{Z_F^\mathrm{D}}{Z_F^{\mathrm{M}_0}} 
\le \frac{Z_F^\mathrm{D}}{Z_F^\mathrm{M}} 
\le \frac{Z_F^\mathrm{D}}{Z_F^{\mathrm{M}_0}} 
\le \frac{1}{\sqrt{\eps L+1}} \frac{Z_G}{\hat{Z}}. 
\]
This shows part (i). 
For part (ii), we can compute $Z_{F^*}^\mathrm{D} / Z_{F^*}^{\mathrm{M}_0}$ in a similar fashion and obtain
\begin{equation}\label{eq:F*-D/M0-ferro}
\left( \frac{\cosh \beta_2}{\cosh \beta_1} \right)^{N^2} \frac{Z_K}{Z_K^\mathrm{mo}}
\le \frac{Z_{F^*}^\mathrm{D}}{Z_{F^*}^{\mathrm{M}_0}}
\le 2\left( \frac{\cosh \beta_2}{\cosh \beta_1} \right)^{N^2} \frac{Z_K}{Z_K^\mathrm{mo}}.
\end{equation}
Therefore, by inequalities~\eqref{eq:F*-D/M0-ferro} and \eqref{eq:beta2-ferro} we obtain
\[
\frac{Z_{F^*}^\mathrm{D}}{Z_{F^*}^\mathrm{M}} 
\le \frac{Z_{F^*}^\mathrm{D}}{Z_{F^*}^{\mathrm{M}_0}} 
\le \frac{1}{\sqrt{\eps L+1}}\, \frac{Z_G^\mathrm{mo}}{\hat{Z}} \frac{Z_K}{Z_K^\mathrm{mo}} \le \frac{2}{r\sqrt{\eps L+1}},
\] 
where the last inequality follows from the assumption $\hat{Z} \ge r Z_G^\mathrm{mo}$ and the fact that $Z_K^\mathrm{mo}/Z_K \ge 1/2$ when $\beta_K \ge 4\log 2$. 
Thus, part (ii) is established.
%\end{proof}

%\begin{proof}[Proof of Lemma~\ref{lem:Z-to-TV}]
To establish part (iii), let us
define $\nu = \mu(\,\cdot\, | \Omega_{F}^\mathrm{M})$ to be the distribution conditioned on $\Omega_{F}^\mathrm{M}$, and similarly $\nu^* = \mu^*(\,\cdot\, | \Omega_{F^*}^\mathrm{M})$. 
By the definition of total variation distance we have
\[
\TV{\mu}{\nu} = \TV{\mu}{\mu(\,\cdot\, | \Omega_{F}^\mathrm{M})} = \frac{Z_{F}^\mathrm{D}}{Z_{F}} = 1 - \frac{Z_{F}^\mathrm{M}}{Z_{F}}.
\]

For part (iii), if $Z_G \le \frac{1}{r} \hat{Z}$, we deduce from part (i) that
\[
\TV{\mu}{\nu} \le \frac{Z_{F}^\mathrm{D}}{Z_{F}^\mathrm{M}} 
\le \frac{1}{\sqrt{\eps L+1}} \frac{Z_G}{\hat{Z}} 
\le \frac{1}{r\sqrt{\eps L+1}} = \frac{\eps}{96(\eps L + 1)} \le \frac{1}{96L}.
\]
Similarly, by part (ii) we have
\[
\TV{\mu^*}{\nu^*} \le \frac{Z_{F^*}^\mathrm{D}}{Z_{F^*}^\mathrm{M}} 
\le \frac{2}{r \sqrt{\eps L+1}} = \frac{\eps}{48(\eps L + 1)} \le \min\left\{\frac{1}{48L}, \frac{\eps}{48}\right\}. 
\]
Let $\rho = \nu(\,\cdot\, | \Omega_{F}^{\mathrm{M}_0})$ denote the conditional distribution of $\nu$ on $\Omega_F^{\mathrm{M}_0}$. 
Observe that $\rho$ does not depend on the graph $G$, because we condition on the event that all vertices from $G$ receive the same spin, and thus the structure of $G$ does not affect the conditional distribution $\rho$.
In particular, we have $\rho = \nu(\,\cdot\, | \Omega_{F}^{\mathrm{M}_0}) = \nu^*(\,\cdot\, | \Omega_{F^*}^{\mathrm{M}_0})$. 
Then, Claim~\ref{claim:Z-and-Z0-ferro-RBM} implies that
\[
\TV{\nu}{\rho} = 1 - \frac{Z_F^{\mathrm{M}_0}}{Z_F^\mathrm{M}} \le e^{-2N}
\]
and similarly $\TV{\nu^*}{\rho} \le e^{-2N}$. 
Therefore, we obtain from the triangle inequality that
\[
\TV{\nu}{\nu^*} \le \TV{\nu}{\rho} + \TV{\nu^*}{\rho} \le 2 e^{-2N}. 
\]
We conclude again from the triangle inequality that
\[
\TV{\mu}{\mu^*} \le \TV{\mu}{\nu} + \TV{\mu^*}{\nu^*} + \TV{\nu}{\nu^*} \le \frac{1}{96L} + \frac{1}{48L} + 2 e^{-2N} \le \frac{1}{16L}. 
\]

For part (iv), if $Z_G \ge r \hat{Z}$, then by part (i)
\[
\TV{\mu}{\nu} \ge 1 - \frac{Z_{F}^\mathrm{M}}{Z_{F}^\mathrm{D}} \ge 1 - 4\sqrt{\eps L+1}\, \frac{\hat{Z}}{Z_G} \ge 1 - \frac{4}{r}\sqrt{\eps L+1} = 1 - \frac{\eps}{24}.
\]
Hence, 
\[
\TV{\mu}{\mu^*} \ge \TV{\mu}{\nu} - \TV{\mu^*}{\nu^*} - \TV{\nu}{\nu^*} \ge 1 - \frac{\eps}{24} - \frac{\eps}{48} - 2e^{-2N} \ge 1-\eps,
\]
as claimed. 
\end{proof}

\begin{proof}[Proof of Claim~\ref{claim:Z-and-Z0-ferro-RBM}]
Note first that $Z_F^{\mathrm{M}_0} \le Z_F^\mathrm{M}$ and
from a union bound we get 
\begin{align*}
1 - \frac{Z_{F}^{\mathrm{M}_0}}{Z_{F}^\mathrm{M}}
&= \Pr\Big(\exists v \in V_G: \sigma(v) \neq \sigma(s_1) \Big| \sigma(s_1) = \sigma(s_2)\Big)\le \sum_{v\in V_G} \Pr(\sigma(v) \neq \sigma(s_1) | \sigma(s_1) = \sigma(s_2)). 
\end{align*}
For every $\sigma \in \Omega_{F}^\mathrm{M}$ and $v\in V_G$, if $\sigma(v) \neq \sigma(s_1)$, then the total weight of edges incident to $v$ is at most $(2e^{\beta_1})^{2N} \exp(\hat{\beta}(N-1) + \hat{h})$; 
and if $\sigma(v) = \sigma(s_1)$, then it is at least $(e^{2\beta_1} + 1)^{2N}$. 
Thus, we get
\begin{align*}
\Pr(\sigma(v) \neq \sigma(s_1) | \sigma(s_1) = \sigma(s_2)) 
&\le \frac{(2e^{\beta_1})^{2N} \exp(\hat{\beta}(N-1) + \hat{h})}{(2e^{\beta_1})^{2N} \exp(\hat{\beta}(N-1) + \hat{h}) + (e^{2\beta_1} + 1)^{2N}}\\
&\le \left(\frac{2 e^{\beta_1}}{e^{2\beta_1}+1}\right)^{2N} \exp\left( \hat{\beta}(N-1) + \hat{h} \right)\\ 
&\le \exp\left(-2(\beta_1-1)N\right) \cdot \exp\left( (\hat{\beta}+\hat{h})N \right)\\
&\le e^{-3N}. 
\end{align*}
where the last inequality follows from the assumption $\beta_1 \ge \frac{1}{2}(\hat{\beta} + \hat{h} + 5)$. 
Therefore,
$
\frac{Z_{F}^{\mathrm{M}_0}}{Z_{F}^\mathrm{M}} \ge 1 - Ne^{-3N} \ge 1 - e^{-2N}. 
$ The bound for $F^*$ is proved analogously. 
\end{proof}

%\subsection{Hardness of testing $(\bmax,\hmax,d,\Dt)$-bounded RBMs when $\bmax d=\omega(\log n)$ and $\hmax \Dt=\omega(\log n)$}

%In this section we prove Theorem~\ref{thm:main-RBM-ferro}.

%\begin{thm}[\textcolor{Green}{Green arrow}]
%Let $d\ge 3$ and $\Delta\ge 2$ be integers. Suppose that we are given an identity testing algorithm for all $(\bmax,\hmax,d,\Dt)$-bounded RBMs when both $\bmax d = \omega(\log n)$ and $\hmax \Dt=\omega(\log n)$, and that the testing algorithm requires $L(n)$ samples and runs in $T(n)$ time. Then there is an identity testing algorithm for all $n$-vertex $(\beta_0,h_0,n,n)$-bounded ferromagnetic Ising models with inconsistent fields that uses $\poly(n,L(n))$ samples and runs in $\poly(n,T(n))$ time.
%\end{thm}

\section{Hardness of testing in bounded degree graphs}
\label{section:main-proof}

In this section, we provide a reduction from identity testing in bounded degree graphs to identity testing in general graphs. We introduce some convenient notation first.
Recall that  we use $\MPotts(n,d,\beta,h)$ for the family of Potts models on $n$-vertex graphs with
maximum degree at most $d$ with the absolute value of the edge and vertex weights bounded by $\beta$ and $h$, respectively; see Definition~\ref{dfn:potts:notation}.
We add ``\textsc{-Bip}'' to the subscript of this notation to denote the restriction to bipartite graphs; that is,
$\MBiPotts(n,d,\beta,h)$ denotes the set of models in $\MPotts(n,d,\beta,h)$, where the underlying graphs is bipartite; note that $\MBiIsing = \MRBM$.
Our reduction will also apply to Ising and Potts models with certain kinds of external fields, and so it is useful then to introduce the notion of \emph{$h$-vertex-monochromatic} external fields. 

\begin{defn}
	\label{dfn:field}
	Consider a Potts model on a graph $G=(V_G,E_G)$ with external field $h_G: V_G \times [q] \rightarrow \R$. 
	For $h \in \R$, we call $h_G$ \emph{$h$-vertex-monochromatic} if 
	$|h_G(v,i)| \le h$ for all $v \in V_G$, $i \in [q]$ and 
	$|\{i\in [q]: h_G(v,i) \neq 0\}| \le 1$ for all $v \in V_G$. 
	%if $h_G(v,\kappa) \neq 0$ for $v \in V_G$, then $h_G(v,j) = 0$ for all $j \neq \kappa$.
\end{defn}

In words, an $h$-vertex-monochromatic field is one that allows $h_G$ to be non-zero (and at most $h$) for at most one spin at each vertex.
We add ``\textsc{-Mono}'' to the subscript of $\MPotts$ to denote the subfamily of models where the external field is $h$-vertex-monochromatic; namely, 
$\MPottsMono(n,d,\beta,h)$ and $\MBiPottsMono(n,d,\beta,h)$ respectively denote the subfamilies of models from $\MPotts(n,d,\beta,h)$ and from $\MBiPotts(n,d,\beta,h)$ with $h$-vertex-monochromatic fields.

\begin{thm}
	\label{thm:deg-red}
	Let $\hat{n},d \in \N^+$ be such that $3 \le d \le \hat{n}^{1-\rho}$ for some constant $\rho \in (0,1)$.
	Suppose that
	%for all sufficiently large $N \in \N$,
	for some constants $\beta,h\ge0$
	%and $0 < h \le n^{\min\{5,4/\rho-1\}} \log N$,
	there is no $\poly(n)$ running time $\varepsilon$-identity testing algorithm for $\MPottsMono(n,n,\beta,h)$.
	%Then, there exists $\varepsilon_n \rightarrow \varepsilon$ such that for all sufficiently large $n$
	Then there exists a constant $c\in(0,1)$ such that, for any constant $\hat\eps > \eps$
	%and all sufficiently large $\hat{n}$,
	there is no $\poly(\hat{n})$ running time $\hat \eps$-identity testing algorithm for $\MBiPottsMono(\hat{n},d,\hat \beta, \hat h)$
	provided $\hat\beta d = \omega(\log \hat{n})$ and
	%$\hat h \le \frac{h}{n^{\min\{5,4/\rho-1\}}}$.
	$\hat h \le h \hat{n}^{-c}$.
	
	Moreover, our reduction preserves ferromagnetism; that is, the statement remains true if we replace the family $\MPottsMono$ by $\MfPottsMono$
	and $\MBiPottsMono$ by $\MfBiPottsMono$.
	
%	
%	Moreover, the hardness is true even if we restrict to a subclass of $\MPotts(\hat{n},d,\hat \beta, \hat h)$ that contains only Potts models whose underlying graph is bipartite.
%	Finally, the statement remains true if we replace $\MPotts$ by $\MfPotts$. In particular, hardness of testing for $\MfIsing$ implies hardness of testing for $\MfRBM$.
\end{thm}

The proof of this theorem is fleshed out in the following sections. First in Section~\ref{subsection:potts-gadget},
we introduce our degree reducing gadget, which consists of a random bipartite graph of maximum degree $d$.
In Section~\ref{subsection:testing-instance-potts}, we describe the construction of the testing instance (i.e., the reduction) and the actual proof of Theorem~\ref{thm:deg-red} is then finalized in Section~\ref{subsec:proof-deg}. 

\subsection{A degree reducing gadget for the Potts model}
\label{subsection:potts-gadget}
Suppose $b,p,d,\din,\dout$ are positive integers such that
$b \ge p$, $d \ge 3$ and $\din + \dout = d$.
Let $B = (V_B,E_B)$ be the random bipartite graph defined as follows:
\begin{enumerate}
	\item Set $V_B=L\cup R$, where $|L|=|R|=b$ and $L\cap R=\emptyset$;
	\item Let $P$ be subset of $V_B$ chosen uniformly at random among all the subsets such that $|P\cap L|=|P\cap R|=p$;
	\item Let $M_1,\dots,M_{\din}$ be $\din$ random perfect matchings between $L$ and $R$;
	\item Let $M_1',\dots,M_{\dout}'$ be $\dout$ random perfect matchings between $L\backslash P$ and $R\backslash P$;
	\item Set $E_B = \left(\bigcup_{i=1}^{\din} M_i \right) \cup \left(\bigcup_{i=1}^{\dout} M_i'\right)$;
	\item Make the graph $B$ simple by replacing multiple edges with single edges.
\end{enumerate}
We use $\Gr(b,p,\din,\dout)$ to denote the resulting distribution; that is, $B \sim \Gr(b,p,\din,\dout)$.
Vertices in $P$ are called \textit{ports}. Every port has degree at most $\din$ while every non-port vertex has degree at most $d$.
The set of ports $P$ is chosen uniformly at random following~\cite{colt}, in order to use the expansion properties of $B \sim \Gr(b,p,\din,\dout)$ proved there.

%In our proofs, we use instances of this random graph model
%with two different choices of parameters.
%For the case when $d$ is such that $3 \le d = O(1)$, we choose
%$p = \floor{m^{1/4}}$, $\din = d-1$ and $\dout=1$; otherwise
%we take $p = m$ (i.e., every vertex is a port), $\din = \floor{\theta d}$ and $\dout = d -  \floor{\theta d}$ for a suitable constant $\theta \in (0,1)$.
%For both parameter choices we establish that
%the random graph $B$ is a good expander with high probability; see Section~\ref{subsection:key-gadget-fact}. Using this, we can show that there are only two ``typical''
%configurations for the Ising model on $B$, even in the presence of an external configuration (i.e., a boundary condition) exerting influence on the configuration of $B$ via its ports.

%We present some notation next that will allow us to formally state these facts.
To capture the notion of an
external configuration for the bipartite graph $B$,
we assume that
$B$ is an induced subgraph of a larger graph $\mathbb{B} = (V_{\BB},E_{\BB})$; i.e., $V_B \subset V_{\BB}$ and $E_B \subset E_{\BB}$.
Let $\partial P  = V_{\BB} \setminus V_{B}$. We assume that
every vertex in $P \subseteq V_B$ is connected to up to $\dout$ vertices in $\partial P$
and that
there are no edges between $V_B \setminus P$ and $\partial P$ in $\BB$.
Given a real number $\beta_B > 0$, we consider the Potts model on the graph $\BB$ with:
\begin{enumerate}
	\item edge interactions given by $\beta_{\BB}: E_{\BB} \rightarrow \R$,
	where $\max_{e \in E_{\BB}\setminus E_B} |\beta_{\BB}(e)| \le \beta_B$ and 
	$\beta_{\BB}(e) = \beta_B$ for every $e \in E_B$;
	\item an external field given by $h_{\BB}: V_{\BB} \times [q] \rightarrow \R$, where there exists $\kappa \in [q]$ and $h \in \R$ such that $h_{\BB}(v,i) = h\cdot \1(i=\kappa) \cdot  \1(v \in V_B)$.
%	$$
%	h_{\BB}(v,i)=
%	\begin{cases}
%	h & v \in \J~\text{and}~i = k; \\
%	0 & v \in V_{\BB} \setminus \J~\text{or}~i \neq k.
%	\end{cases}
%	$$
\end{enumerate}
We remark that the field $h_{\BB}$ is $h$-vertex-monochromatic, but we also require that the spin for which the field is allowed to be not zero to be the same for all vertices.

Let $\sigma^i(B)$ be the configuration of $B = (V_B,E_B)$ where every vertex in $V_B$
is assigned color $i \in [q]$.
Let $\{\partial P = \tau\}$ denote the event that the configuration on $\partial P$ is $\tau \in [q]^{\partial P}$.
For certain choices of the random graph parameters
we can show that for any $\tau$, with high probability over~$B$,
the Potts configuration of $V_B$ on $\BB$ conditioned on $\{\partial P = \tau\}$ will likely be $\sigma^i(B)$  for some $i \in [q]$.

\begin{thm}
	\label{thm:gadget:fact-ld}
	Suppose $3 \le d = O_b(1)$, $\din = d-1$, $\dout = 1$ and $p = \floor{b^\alpha}$, where $\alpha \in (0,\frac{1}{4}]$ is a constant independent of $b$.
	Then, there exists a constant $\delta > 0$
	such that
	with probability $1-o(1)$ over the choice of the random graph $B$
	the following holds for every configuration $\tau$ on $\partial P$:
	$$
	\mu_{\BB} \left( \bigcup_{i \in [q]} \{\sigma^i(B)\} \;\middle|\; \partial P = \tau \right) \ge \left(1 -    \frac{q^2{\e}^{2h}}{{\e}^{\delta \beta_B d }} \right)^{2b}.
	$$
\end{thm}

\begin{thm}
	\label{thm:gadget:fact-hd}
	Suppose $p = b$ and $4 + \frac{1200}{\rho} \le d \le b^{1-\rho}$
	for some constant $\rho \in (0,1)$ independent of $b$.
	Then, there exist constants $\delta = \delta(\rho) > 0$ and $\theta = \theta(\rho) \in (0,1)$
	such that
	when $\din = \floor{\theta d}$ and $\dout = d - \floor{\theta d}$
	the following holds for every configuration $\tau$ on $\partial P$ with probability $1-o(1)$ over the choice of the random graph $B$:
	$$
	\mu_{\BB} \left( \bigcup_{i \in [q]} \{\sigma^i(B)\} \;\middle|\; \partial P = \tau \right)  \ge \left(1 -    \frac{q^2{\e}^{2h}}{{\e}^{\delta \beta_B d }} \right)^{2b}.
	$$
\end{thm}
\noindent
These theorems are extensions of Theorems 4.1 and 4.2 in \cite{colt},
where similar bounds are established for the case when every edge of $\BB$ has the same weight $\beta < 0$; i.e., the antiferromagnetic setting. In this new setting, there is an external field, every edge in $E_B$ has weight $\beta_B > 0$, and edges between $P$ and $\partial P$ are allowed to have either negative or positive weights bounded in absolute value by $\beta_B$.

%\subsection{Proof of Theorems \ref{thm:gadget:fact-ld} and \ref{thm:gadget:fact-hd}}

\begin{proof}[Proof of Theorems~\ref{thm:gadget:fact-ld} and~\ref{thm:gadget:fact-hd}]

    Let $E(S,T)$ denote the set of edges between $S$ and $T$ in $E_\BB$.
	For ease of notation, we set $\beta = \beta_B$.
	Let $P_i \subseteq \partial P$ be the set of vertices of $\partial P$ that are assigned color $i \in [q]$ in $\tau$.
	The weight of $\sigma^i(B)$ in $\BB$ conditional on $\tau$ is then given by
	\begin{align}
	\label{eq:ground-state}
	w^i := w_{\BB}^\tau(\sigma^i(B)) =  \exp\left[ \beta d b +  2bh  \1(i = \kappa)+ \sum_{e \in E(P,P_i)} \beta_{\BB} (e)  \right].
	\end{align}

	Let $\Omega_B$ be the set of Potts configurations of the graph $B$.
	For $\sigma\in\Omega_B$, let $S_\sigma(i) \subseteq V_B$ be the set of vertices that are assigned color $i \in [q]$ in $\sigma$.
	We let $S_\sigma$ denote the set of maximum cardinality among $S_\sigma(1),\dots,S_\sigma(q)$.
	Let
	$\Omega_B^i \subseteq \Omega_B$ be the set of configurations $\sigma$ such that $S_\sigma = S_\sigma(i)$.
	For $\sigma\in\Omega_B$, we use $w^\tau(\sigma)$ for the weight of the configuration
	on $\BB$ that agrees with
	$\sigma$ on $V_B$ and with $\tau$ on $V_\BB\setminus V_B$.
	By definition, the partition function $Z_{\BB}^\tau$
	for the conditional distribution $\mu_{\BB}(\cdot \mid \partial P = \tau)$
	satisfies
	\begin{equation}
	\label{eq:partition:main}
	Z_{\BB}^\tau = \sum_{\sigma \in \Omega_B} w^\tau(\sigma) = \sum_{\sigma \in \Omega_B: |S_\sigma|  > b} w^\tau(\sigma) + \sum_{\sigma \in \Omega_B: |S_\sigma|  \le b} w^\tau(\sigma).
	\end{equation}
	
	We bound each term in the right-hand side of \eqref{eq:partition:main} separately.
	For $\sigma \in \Omega_B$, let $r(\sigma,i) = |S_\sigma(\kappa)|h - 2bh  \1(i = \kappa)$.
	We will show that in the regimes of parameters in Theorems~\ref{thm:gadget:fact-ld}	and~\ref{thm:gadget:fact-hd},
	with probability $1-o(1)$ over the choice of
	the random graph $B \sim \Gr(b,p,\din,\dout)$, there exists a constant $\delta > 0$ such that for every $\sigma \in \Omega_B$:
	\begin{align}
	w^\tau(\sigma) &\le w^i \cdot e^{ -\delta \beta d |V_B\setminus S_\sigma(i)| + r(\sigma,i)}~~\text{when}~~\sigma\in\Omega_B^i,~|S_\sigma|>b;~\text{and}\label{eq:weight-bound:1}\\
	w^\tau(\sigma) &\le w^i \cdot e^{ -\delta \beta d b + r(\sigma,i)}~~\text{when}~~\sigma\in\Omega_B^i,~|S_\sigma|\le b. \label{eq:weight-bound:2}
	\end{align}
	Before proving these two bounds, we show how to use them to complete the proofs of the theorems.
	From \eqref{eq:weight-bound:1}, we get
	\begin{align*}
	\sum_{\sigma \in \Omega_B:|S_\sigma|  >  b} w^\tau(\sigma)
	= \sum_{i=1}^q \sum_{\sigma \in \Omega_B^i:|S_\sigma|  >  b} w^\tau(\sigma)
	\leq  \sum_{i=1}^q \sum_{\sigma \in \Omega_B^i:|S_\sigma|  >  b} w^{i} \!\cdot\! e^{ -\delta \beta d |V_B \setminus S_\sigma(i)| + r(\sigma,i)}.
	%\\
	%&=  \sum_{i=1}^q   w^{i} \sum_{k=0}^b  \binom{2b}{k} (q-1)^k {\e}^{ -\delta \beta d k }
	%\leq \left(1+\frac{q-1}{\e^{\delta\beta d}}\right)^{2b} \sum_{i=1}^q  w^{i} .
	\end{align*}
	If $i = \kappa$,
	\begin{align*}
	\sum_{\sigma \in \Omega_B^\kappa:|S_\sigma|  >  b} e^{-\delta \beta d |V_B \setminus S_\sigma(\kappa)|  + r(\sigma,\kappa)} &= \sum_{\sigma \in \Omega_B^\kappa:|S_\sigma|  >  b} e^{ -\delta \beta d |V_B \setminus S_\sigma(\kappa)|  - h |V_B \setminus S_\sigma(\kappa)|} \\	
	&= \sum_{x=0}^{b}  \binom{2b}{x}  (q-1)^{x} {\e}^{ -(\delta \beta d +h)x }
	\le \left(1+\frac{q-1}{\e^{\delta\beta d+h}}\right)^{2b}.
		%&= \sum_{x_1=0}^\Delta \sum_{x_2=0}^{b-x_1} \binom{\Delta}{x_1} \binom{2b - \Delta}{x_2}  (q-1)^{x_1+x_2} {\e}^{ -\delta \beta d (x_1+x_2)-h x_1 } \\
	%&\le \left(1+\frac{q-1}{\e^{\delta\beta d}}\right)^{2b-\Delta}\left(1+\frac{q-1}{\e^{\delta\beta d+h}}\right)^{\Delta}
	\end{align*}
	If $i \neq \kappa$,
	\begin{align*}
	\sum_{\sigma \in \Omega_B^i:|S_\sigma|  >  b} e^{-\delta \beta d |V_B \setminus S_\sigma(i)| + r(\sigma,i)}
	&= \sum_{\sigma \in \Omega_B^i:|S_\sigma|  >  b} e^{-\delta \beta d |V_B \setminus S_\sigma(i)| + h |S_\sigma(\kappa)| }\\
	&= \sum_{x=0}^b \sum_{y=0}^{x}  \binom{2b}{x} \binom{x}{y}  (q-2)^{x-y}  {\e}^{ -\delta \beta d x+h y } \\
	&\le \sum_{x=0}^b \binom{2b}{x}  (q-2)^{x}  {\e}^{ -\delta \beta d x} \left(1 + \frac{{\e}^h}{q-2}\right)^x \\
	&\le  \left(1+\frac{q-2+{\e}^h}{\e^{\delta\beta d}}\right)^{2b}.
	\end{align*}
	Hence, 	letting $W = \sum_{i = 1}^q w^{i}$, we obtain
	$$
	\sum_{\sigma \in \Omega_B:|S_\sigma|  >  b}  w^\tau(\sigma) \le  w^\kappa \left(1+\frac{q-1}{\e^{\delta\beta d+h}}\right)^{2b} + (W-w^\kappa) \left(1+\frac{q-2+{\e}^h}{\e^{\delta\beta d}}\right)^{2b}.
	$$
	
	To bound the second summand from~\eqref{eq:partition:main},
	note that from \eqref{eq:weight-bound:2} we get
	\begin{align*}
	\sum_{\sigma:|S_\sigma|  \le  b} w^\tau(\sigma)
	&= \sum_{i=1}^q \sum_{\sigma \in \Omega_B^i:|S_\sigma|  \le   b} w^\tau(\sigma)
	\leq  \sum_{i=1}^q \sum_{\sigma \in \Omega_B^i:|S_\sigma|  \le b} w^{i} \!\cdot\!{\e}^{-\delta \beta d b + r(\sigma,i)}\\
	&\le  \sum_{i=1}^q w^{i} \!\cdot\!{\e}^{-\delta \beta d b} \sum_{\sigma \in \Omega_B^i:|S_\sigma|  \le b} {\e}^{|S_\sigma(\kappa)| h}\\
		&\le  W \!\cdot\!{\e}^{-\delta \beta d b} \sum_{x=0}^{2b} \binom{2b}{x} {\e}^{x h} (q-1)^{2b-x}\\
		&\le  W \!\cdot\!\left( \frac{(q-1+{\e}^h)^2}{{\e}^{\delta \beta d }}\right)^{b}.
	\end{align*}
	Thus,
	\begin{align*}
	Z_{\BB}^\tau &\leq w^\kappa \left(1+\frac{q-1}{\e^{\delta\beta d+h}}\right)^{2b} + (W-w^\kappa) \left(1+\frac{q-2+{\e}^h}{\e^{\delta\beta d}}\right)^{2b} + \left( \frac{(q-1+{\e}^h)^2}{{\e}^{\delta \beta d }}\right)^{b} W \\
	 &\le W\left[ \left(1+\frac{q-2+{\e}^h}{\e^{\delta\beta d}}\right)^{2b}	 + 	\left( \frac{(q-1+{\e}^h)^2}{{\e}^{\delta \beta d }}\right)^{b}\right].
	 \end{align*}
	Setting $x = \frac{q-2+{\e}^h}{\e^{\delta\beta d}}$, $y = \frac{(q-1+{\e}^h)^2}{{\e}^{\delta \beta d }}$ and $z = \frac{q^2{\e}^{2h}}{{\e}^{\delta \beta d }}$
	\begin{align*}
	\mu_{\BB} \left( \bigcup_{i \in [q]} \{\sigma^i(B)\} \mid \partial P = \tau \right)  = \frac{W }{Z_{\BB}^\tau} \geq \frac{1}{(1+x)^{2b}+y^b}
	\ge \frac{1}{(1+2z)^{2b}} \ge (1-2z)^{2b}
%= \frac{1}{(1+x)^{2b}\left(1 + \left(\frac{y}{(1+x)^2}\right)^b\right)} \\
%	&\ge (1 -x)^{2b} \left(1 - y^b\right)
%	\ge (1-2bx)(1-by) \\
%	&\ge 1-2b(x+y) \ge 1- \frac{2bq^3}{\e^{\delta\beta d-h}},
	%\geq 1 - 3bz,
	\end{align*}
	as claimed.
%	where the third inequality follows from $(1-x)^{a} \geq 1-ax$ for all $x\in[0,1)$ and the second to last from the fact that $\left(\frac{z}{(1+z)^2}\right)^b \le bz$ for interger $b > 0$.
	
	It remains for us to establish \eqref{eq:weight-bound:1} and \eqref{eq:weight-bound:2};
	we start with \eqref{eq:weight-bound:1}.
	For $S,T \subseteq V_B \cup \partial P$, let $[S,T]$ denote the number of edges between $S$ and $T$ in the graph $\BB$.
	%Observe that if $i \neq i^*$, then $|S_\sigma(i)| \le b$.
	Then,
	\begin{align}
	\label{eq:weight-first}
	w^\tau(\sigma) &= \exp\left[ \beta\sum_{j=1}^q [S_\sigma(j),S_\sigma(j)] + \sum_{j=1}^q \sum_{e \in E(P_j,S_\sigma(j) \cap P) }  \beta_{\BB} (e) + h|S_\sigma(\kappa)|\right].
	\end{align}
	Now, $\sum_{j=1}^q [S_\sigma(j),S_\sigma(j)]  \le d b -  [S_\sigma,V_B\setminus S_\sigma]$ and for any $i \in [q]$
	\begin{align*}
	\sum_{j=1}^q \sum_{e \in E(P_j,S_\sigma(j) \cap P) }  \beta_{\BB} (e)  -  \sum_{e \in E(P,P_i)} \beta_{\BB} (e)
	&= \sum_{j \neq i} \sum_{e \in E(P_j,S_\sigma(j) \cap P) }  \beta_{\BB} (e)  -  \sum_{e \in E(P \setminus S_\sigma(i),P_i)} \beta_{\BB} (e)
	\\
	& \le \beta \sum_{j \neq i}  [S_\sigma(j) \cap P,P_j \cup P_{i}].
	\end{align*}
	Plugging these two bounds into~\eqref{eq:weight-first} and using~\eqref{eq:ground-state}, we get for $\sigma \in \Omega_B^i$
	\begin{align}
	w^\tau(\sigma)	&\le \exp\left[ \beta(d b -  [S_\sigma,V_B\setminus S_\sigma])+ \beta \sum_{j \neq i}  [S_\sigma(j) \cap P,P_j \cup P_{i}] +\sum_{e \in E(P,P_i)} \beta_{\BB} (e) + h|S_\sigma(\kappa)| \right] \notag \\
	&= w^{i} \cdot \exp\left[ -{\beta} [S_\sigma,V_B\setminus S_\sigma] + \beta \sum_{j \neq i} [S_\sigma(j) \cap P,P_j \cup P_{i}] + r(\sigma,i)\right]\notag\\
		&\le  w^{i} \cdot \exp\left[ -{\beta} \left([S_\sigma,V_B\setminus S_\sigma] -[(V_B\setminus S_\sigma) \cap P,\partial P]\right) + r(\sigma,i)\right].
	%&\le w^{i} \cdot \exp\left[ -{\beta} \left([S_\sigma,V_B\setminus S_\sigma] -|V_B\setminus S_\sigma \cap P|\right)\right].
	\label{eq:weight-bound-1}
	\end{align}
	
	When $3 \le d = O_b(1)$, $p = \floor{b^\alpha}$ with $\alpha \in (0,\frac{1}{4}]$, $\din = d-1$ and $\dout = 1$.
	Hence, $[V_B\setminus S_\sigma \cap P,\partial P] = |(V_B\setminus S_\sigma) \cap P|$.
	Theorems~5.2 and~5.3 from~\cite{colt} imply that  exists a constant $\gamma > 0$ such that with probability $1-o(1)$ over the choice of the random graph $B$ we have
	\begin{align*}
	\frac{[S_\sigma,V_B\setminus S_\sigma]}{|V_B \setminus S_\sigma|} &\ge \gamma d,~\textrm{and}\\
	\frac{[S_\sigma,V_B\setminus S_\sigma]}{|(V_B \setminus S_\sigma) \cap P|} &\ge 1 + \gamma.
	\end{align*}
	Combining these two inequalities we get for
	$\delta = \frac{\gamma^2}{1+\gamma}$ that
	\begin{equation*}
	\label{eq:exr-imeq}
	[S_\sigma,V_B\setminus S_\sigma] \ge |(V_B \setminus S_\sigma) \cap P| + \delta d |V_B \setminus S_\sigma|.
	\end{equation*}
	Plugging this bound into \eqref{eq:weight-bound-1},
	\begin{equation}
	\label{eq:weight-bound-2}
w^\tau(\sigma) \le w^{i}  \cdot {\exp}\left[ -\delta \beta d |V_B \setminus S_\sigma| + r(\sigma,i)\right],
	\end{equation}
	and we get~\eqref{eq:weight-bound:1}, since $\sigma \in \Omega_B^i$ and so $S_\sigma = S_\sigma(i)$.
	
	Under the assumptions in Theorem~\ref{thm:gadget:fact-hd},
	we can also establish \eqref{eq:weight-bound-2} as follows.
	When $b^{1-\rho} \ge d \ge \din = \floor{\theta d} \ge 3$,  Theorem 5.1 from~\cite{colt} implies that
	$$
	[S_\sigma,V_B\setminus S_\sigma]   \ge \frac{\rho \din}{300} |V_B\setminus S_\sigma| = \frac{\rho \floor{\theta d}}{300}   |V_B\setminus S_\sigma|.
	$$
	Moreover,
	$$
	[V_B \setminus S_\sigma , \partial P] \le \dout |V_B\setminus S_\sigma| = (d-\floor{\theta d}) |V_B\setminus S_\sigma|.
	$$
	Hence,  taking
	$
	\theta = \frac{300 + 0.75 \rho}{300+\rho}
	$\
	we get  that when $d \ge 4 + \frac{1200}{\rho}$:
	\begin{equation}
	\label{eq:theta}
	\frac{\rho \floor{\theta d}}{300}   - (d-\floor{\theta d})  \ge \frac{\rho d}{600}.
	\end{equation}
	Together with \eqref{eq:weight-bound-1} this implies
	$$
	w^\tau(\sigma) \le   w^i\cdot  {\exp}\left[ -\frac{\rho \beta d |V_B\setminus S_\sigma|}{600} + r(\sigma,i)\right],
	$$
	which gives \eqref{eq:weight-bound-2} for $\delta \leq \rho/600$, and thus we again obtain~\eqref{eq:weight-bound:1}.
	(Observe that our choice of $\theta$ guarantees $d-1 \ge \din = \floor{\theta d} \ge 3$ for all $d \ge 4$.)

%	For the case when $S_\sigma = L_\sigma^- \cup R_\sigma^+$ we deduce analogously that for a suitable $\delta > 0$
%	\begin{align}
%	w^\tau(\sigma) \le w^+ \cdot {\exp}\left[ -\delta |\beta| d |S_\sigma|\right].    \label{eq:weight-bound-5}
%	\end{align}
	
	We establish~\eqref{eq:weight-bound:2} next.
	Since
	$$
	\sum_{j=1}^q [S_\sigma(j),S_\sigma(j)] = bd - \frac{1}{2} \sum_{j=1}^q [S_\sigma(j),V_B\setminus S_\sigma(j)],
	$$
	and
	$$
	\sum_{j=1}^q \sum_{e \in E(P_j,S_\sigma(j) \cap P) }  \beta_{\BB} (e)  -  \sum_{e \in E(P,P_i)} \beta_{\BB} (e) \le \beta \dout |P|,
	$$
	we get from~\eqref{eq:ground-state} and~\eqref{eq:weight-first} that for $\sigma \in \Omega_B^i$
	\begin{align}
	w^\tau(\sigma)
	&\le w^{i} \cdot \exp\left[ -\frac{\beta}{2} \sum_{j=1}^q [S_\sigma(j),V_B\setminus S_\sigma(j)]  + \beta \dout |P|+ r(\sigma,i) \right] \notag\\
	&\le w^{i} \cdot \exp\left[ -{\beta} \left(\frac{1}{2}\sum_{j=1}^q [S_\sigma(j),V_B\setminus S_\sigma(j)] -\dout|P|\right) + r(\sigma,i)\right].
	\label{eq:weight-bound-3}
	\end{align}
	Since $|S_\sigma(j)| \le b$ for $j\in [q]$,
	our assumptions in Theorem~\ref{thm:gadget:fact-ld} combined with
	Theorem~5.2 from~\cite{colt} imply that there  exists a constant $\gamma > 0$ such that with probability $1-o(1)$ over the choice of the random graph $B$ we have for all $j \in [q]$
	\begin{align*}
	\frac{[S_\sigma(j),V_B \setminus S_\sigma(j)]}{|S_\sigma(j)|} &\ge \gamma d.
	\end{align*}
	Plugging this bound  into \eqref{eq:weight-bound-3}, and since $\dout=1$ by assumption, we get
	\begin{align*}
	%\label{eq:weight-bound-4}
	w^\tau(\sigma) &\le  w^{i} \cdot \exp\left[ -{\beta} \left(\frac{1}{2}\sum_{j=1}^q \gamma d |S_\sigma(j)|-|P|\right) + r(\sigma,i)\right] \\
	&= w^{i} \cdot \exp\left[ -{\beta} \left(\gamma d b-|P|\right) + r(\sigma,i) \right] \le  w^{i} \cdot \exp\left[ -{\beta \delta d b} + r(\sigma,i)\right],
	\end{align*}
	where the last inequality holds for a suitable constant $\delta > 0$ and $b$ sufficiently large since $|P| \le \lfloor b^{1/4} \rfloor$.

	Finally, the assumptions in Theorem~\ref{thm:gadget:fact-hd} and Theorem 5.1 from~\cite{colt} imply that
	$$
	[S_\sigma(j),V_B\setminus S_\sigma(j)]   \ge \frac{\rho \din}{300} |S_\sigma(j)| = \frac{\rho \floor{\theta d}}{300}   | S_\sigma(j)|.
	$$
	Hence, since $|P|=b$ %, for a suitable constant $\gamma > 0$
	\begin{align*}
	w^\tau(\sigma) &\le  w^{i} \cdot \exp\left[ -{\beta} \left(\frac{1}{2}\sum_{j=1}^q  \frac{\rho \floor{\theta d}}{300}   |S_\sigma(j)|-(d-\floor{\theta d})|P|\right)+ r(\sigma,i)\right] \\
	&\le w^{i} \cdot \exp\left[ -{\beta} b \left( \frac{\rho \floor{\theta d}}{300} -(d-\floor{\theta d})\right)+ r(\sigma,i)\right] \\
	&\le  w^{i} \cdot \exp\left[ -{\beta \delta d b}+ r(\sigma,i)\right],
	\end{align*}
	where the last inequality holds for a suitable constant $\delta > 0$ for $\theta$ satisfying \eqref{eq:theta}.
	This completes the proofs of the theorem.
\end{proof}

\subsection{Testing instance construction}
\label{subsection:testing-instance-potts}

Consider a Potts model on an $n$-vertex graph $G=(V_G,E_G)$,
with edge interactions $\beta_G: E_G \rightarrow \R$
and an $h$-vertex-monochromatic external field $h_G: V_G \times [q] \rightarrow \R$; see Definition~\ref{dfn:field}.
%We assume $h_G$ is such that
%$|h_G(v,i)| \le h$ for all $v \in V_G$, $i \in [q]$ and,
%if $h_G(v,i) \neq 0$ for $v \in V_G$, then $h_G(v,j) = 0$ for all $j \neq i$.
We show how to construct a Potts model on a larger graph of maximum degree at most $d$,
with edge interactions bounded by $\hat\beta$
and an $\hat h$-vertex-monochromatic external field
whose distribution captures that of the model $(G,\beta_G,h_G)$. We can think of $d$, $\hat \beta$ and $\hat h$ as the parameters for our construction.

%Let $\ell(e) = \lceil |\beta_G(e)|/ \hat \beta\rceil$ %and replace every edge in $e \in E_G$
%with $\ell(e)$ edges each of weight $\beta_G(e)/\ell(e)$, so that $|\beta_G(e)|/\ell(e) \le \hat \beta$.
%Let $\Gs = (V_{\Gs},E_{\Gs})$ be the resulting multigraph.

We use an instance of the random bipartite graph $\Gr(b,p,\din,\dout)$
from Section~\ref{subsection:potts-gadget}
as a gadget to define a simple graph $\Gs_\Gamma= (V_{\Gs_\Gamma},E_{\Gs_\Gamma})$, where
$\Gamma$ denotes the set parameters $\{b,p,\din,\dout\}$.
The graph $\Gs_\Gamma$ is constructed as follows:
\begin{enumerate}
	\item Generate an instance $B = (V_B , E_B)$ of the random graph model $\Gr(b,p,\din,\dout)$;
	\item Replace every vertex $v$ of $\Gs$ by a copy $B_v =  ( L_v \cup R_v, E_{B_v})$ of the generated instance $B$;
	
	\item For every edge $e = \{v,u\} \in E_{\Gs}$, let $\ell(e) = \lceil |\beta_G(e)|/ \hat \beta\rceil$ and 
	%add exactly $\lceil\ell(e)/2\rceil$ edges between ports in $L_v$ and ports in $R_u$ while keeping that every port is adjacent to at most $\dout$ ports from other gadgets. 
	choose $\dout \cdot \lceil \ell(e)/ \dout^2 \rceil$ unused ports in $L_v$, 
	$\dout \cdot\lceil \ell(e)/ \dout^2 \rceil$ unused ports in $R_u$ and 
	connect them with any simple bipartite graph of maximum degree at most $\dout$ and exactly $\ell(e)$ edges;
	
	\item Similarly, for every edge $e = \{v,u\} \in E_{\Gs}$, 
	%add exactly $\lfloor\ell(e)/2\rfloor$ edges between ports in $R_v$ and ports in $L_u$ while keeping that every port is adjacent to at most $\dout$ ports from other gadgets. 
	choose $\dout \cdot\lceil \ell(e)/ \dout^2 \rceil$ unused ports in $R_v$ 
	and $\dout \cdot \lceil \ell(e)/ \dout^2 \rceil$ unused ports in $L_u$ and connect them with any simple bipartite graph of maximum degree at most $\dout$ and exactly $\ell(e)$ edges; 
\end{enumerate}

Let $d_G$ be the maximum degree of the graph $\Gs$.
Our construction requires:
\begin{align}
&\din+\dout = d \le b,\label{eq:cons-cond-1}\\
&d_G \cdot \left( \dout \cdot \max_{e \in E_G} \left\lceil \frac{\ell(e)}{\dout^2} \right\rceil \right) \le p.  \label{eq:cons-cond-3}
%d_G \cdot \max_{e \in E_G} \ell(e)  &\le p \cdot \dout. \label{eq:cons-cond-3}
%\dout^2 &\leq \max_{e \in E_{\Gs}} \ell(e). \label{eq:cons-cond-3}
\end{align}
%To see that \eqref{eq:cons-cond-3} is necessary, note that the total out-degree of the ports in a gadget should be large enough to accommodate $d_G$ edges. 
Observe also that there is always a simple bipartite graph of maximum degree at most $\dout$ and exactly $\ell(e)$ edges for steps 3 and 4;
take, for example,  $\lfloor \ell(e)/ \dout^2 \rfloor$ disjoint copies of the complete bipartite graph with $\dout$ vertices on each side, and add one additional bipartite graph with $\dout$ vertices on each side 
for the remaining edges when $\ell(e)/ \dout^2$ is not an integer.
% (see~Gale-Ryser theorem~\cite{GR}).
%In steps 3 and 4,  each side of the bipartite graph will have $\lceil \ell(e)/ \dout \rceil - 1$ vertices of degree exactly $\dout$ and one vertex with degree $\ell(e) - (\lceil \ell(e)/ \dout \rceil - 1)\dout$.
%Observe also that when condition~\eqref{eq:cons-cond-3} holds, there is always a simple bipartite graph of maximum degree at most $\dout$ with $\lceil \ell(e)/ \dout \rceil$ vertices on each side for steps 3 and 4 (see~Gale-Ryser theorem~\cite{GR}).

We consider the Potts model on the graph $ (V_{\Gs_\Gamma},E_{\Gs_\Gamma})$
with edge weights $\beta_{\Gs_\Gamma}: E_{\Gs_\Gamma} \rightarrow \R$ 
and external field $h_{\Gs_\Gamma}: V_{\Gs_\Gamma} \times [q] \rightarrow \R$
defined as follows:
\begin{enumerate}
\item  each edge with both of its endpoints in the same gadget is assigned weight $\beta_B := \hat\beta$;
\item if the edge connects the gadgets corresponding to $u \neq v \in V_G$, then it is assigned weight
$\frac{\beta_G(\{u,v\})}{2\ell(\{u,v\})}$. 
%\zongchen{There should be a 2 in the denominator, since we connect $2\ell(u,v)$ edges between $B_u$ and $B_v$.}
\item for each vertex $v\in V_G$,
%the model $\Model(G)$ has external field $h_{\Model(G)}: V_{\Gs_\Gamma} \rightarrow \R$ where
every vertex $u$ in the gadget $B_v$ is assigned the field $h_{\Gs_\Gamma} (u,i) := h_G(v,i)/2b$ for $i\in[q]$.
\end{enumerate}
Note that if $h_G$ is $h$-vertex-monochromatic, then $h_{\Gs_\Gamma}$ is $(h/2b)$-vertex-monochromatic, and that in the gadget of every vertex only one spin may receive a non-zero weight; in particular, if $h_G$ is $h$-vertex-monochromatic, then the field in every gadget would satisfy the conditions Section~\ref{subsection:potts-gadget}.

%The number of vertices in $F(G,{\Gamma})$ is $m+t+N$
%and its maximum degree is $d = \din+\dout$.
%%; thus, $\ol{H}_w^{\Gamma} \in \mathcal{M}(2m(N+2),d)$.
%Consider the Potts model on the complete graph $K_N$ with edge interaction $\beta_N$; we can analogously define $F(K_N,{\Gamma})$.
%Let $\F$ and $\FN$ denote the Potts models
%$(F(G,{\Gamma}),\beta,\beta_{T,G},\beta_H)$ and $(F(K_N,{\Gamma}),\beta_N,\beta_{T,G},\beta_H)$, respectively.
%We show next that the models $\F$ and $\FN$
%are statistically close if and only if the dominant phase in $H$ is the disordered one.

%We introduce  some additional notation next.
For a configuration $\sigma$ on $\Gs_\Gamma$,
we say that the gadget $B_v \!=\!  (V_{B_v}, E_{B_v})$ is in the $i$-\emph{th phase} if all the vertices in $V_{B_v}$ are assigned  spin $i \in \{1,\dots,q\}$.
Let $\Good$ be the set of configurations of $\Gs_\Gamma$ where the gadget of every vertex is in a phase (not necessarily the same).
The set of all Potts configurations of $\Gs_\Gamma$ is denoted by $\Omega$.
We use $Z_{\Gs_\Gamma}$ for the partition function of the Potts model on $\Gs_\Gamma$ and $Z_{\Gs_\Gamma}(\Lambda)$ for its restriction to a subset of configurations $\Lambda \subseteq \Omega$.
That is,
$
Z_{\Gs_\Gamma} = \sum_{\sigma \in \Omega} w_{\Gs_\Gamma}(\sigma)
$
and
$
Z_{\Gs_\Gamma}(\Lambda) = \sum_{\sigma \in \Lambda} w_{\Gs_\Gamma}(\sigma)
$
where
$$
w_{\Gs_\Gamma}(\sigma) :=  \exp\left[ \sum_{\{u,v\} \in E_{\Gs_\Gamma}} \beta_{\Gs_\Gamma}(\{u,v\}) \cdot \1 (\sigma(u)=\sigma(v)) + \sum_{v \in  V_{\Gs_\Gamma}} h_{\Gs_\Gamma}(v,\sigma(v))\right]$$
is the {weight} of the configuration $\sigma$.

%When $\beta < 0$,
%$w_M(\sigma) =  {\e}^{-|\beta| A(\sigma)}$.

For a configuration $\sigma \in \Good$, let $\sigma_G$ be the corresponding configuration on $G$ where $\sigma_G(v)$ is set to the phase of gadget $B_v$ in $\sigma$.
Let $\mu_G$ and $\mu_{\Gs_\Gamma}$ denote the Gibbs distribution for the Potts models we just defined on $G$ and $\Gs_\Gamma$. 
	From our construction, we can deduce the following fact.

\begin{lemma}
		\label{lemma:model:eq}
	For any graph $G$, we have
	$
	\mu_{\Gs_\Gamma}(\sigma \mid \sigma \in \Good) = \mu_{G}(\sigma_G).
	$
\end{lemma}

\begin{proof}
	Let $Q \subseteq E_{\Gs_\Gamma}$ be the edges of $\Gs_\Gamma$ that connect vertices between different gadgets. Then, for $\sigma \in \Good$,
	\begin{align*}
	\sum\limits_{\{u,v\}\in Q} \beta_{\Gs_\Gamma}(\{u,v\}) \1 (\sigma(u)=\sigma(v)) &= \sum\limits_{\{u',v'\}\in E_G} \beta_{G}(\{u',v'\}) \1 (\sigma_G(u')=\sigma_G(v')),\\
	\sum\limits_{\{u,v\}\in E_{\Gs_\Gamma}\setminus Q} \beta_{\Gs_\Gamma}(\{u,v\}) \1 (\sigma(u)=\sigma(v)) &= \exp\left(\beta_B \din b n\right),~\text{and}\\
	\sum_{v \in  V_{\Gs_\Gamma}} h_{\Gs_\Gamma}(v,\sigma(v)) &= \sum_{v' \in  V_{G}} h_G(v',\sigma_G(v')).	
	\end{align*}
	Thus, $w_{\Gs_\Gamma}(\sigma) = w_{G}(\sigma_G)\exp\left(\beta_B \din bn\right)$,
	and
	\[
	\mu_{\Gs_\Gamma}(\sigma \mid \sigma \in \Good) = \frac{w_{\Gs_\Gamma} (\sigma)}{Z_{\Gs_\Gamma}(\Good)} =\frac{w_{G}(\sigma_G) \exp\left(\beta_B \din bn \right)}{Z_{G} \exp\left(\beta_B \din bn\right) }=\mu_{G}(\sigma_G).\qedhere
	\]
\end{proof}

\begin{lemma}\label{lemma:main-reduction-phase}
	Let $(G,\beta_G,h_G)$ and $(G^*,\beta_{G^*},h_{G^*})$ be two Potts on the $n$-vertex graphs $G$ and $G^*$, respectively.
	Let $\Gamma = (b,p,\din,\dout)$ be such that conditions \eqref{eq:cons-cond-1} and \eqref{eq:cons-cond-3} are satisfied.
	Suppose that
	$\mu_{\Gs_\Gamma}(\Omega_{\rm good}) \ge 1- \delta$ and $\mu_{\Gs_\Gamma^*}(\Omega_{\rm good}) \ge 1- \delta$  for some $\delta \in (0,1)$.
	Then,
	%\begin{enumerate}
	%	\item If $\TV{\mu_{G,\beta_G,h_G}}{\mu_{F,\beta_F,h_F}}  ~\le~ \eps$, then:
	%		$$
	%		\TV{\mu_{\Model(G)}}{\mu_{\Model(F)}} \le \varepsilon + 2\delta.
	%		$$
	%	\item If $\TV{\mu_{G,\beta_G,h_G}}{\mu_{F,\beta_F,h_F}}  ~\ge~ \gamma$, then:
	%	$$
	%	\TV{\mu_{\Model(G)}}{\mu_{\Model(F)}} > \gamma - \delta.
	%	$$
	\[
	\TV{\mu_G}{\mu_{G^*}} - 2\delta \le
	\TV{\mu_{\Gs_\Gamma}}{\mu_{\Gs_\Gamma^*}} \le \TV{\mu_G}{\mu_{G^*}} + 2\delta.
	\]	
	%\end{enumerate}
\end{lemma}

\begin{proof}
	%For ease of notation we set $\mu_{G} = \mu_{G,\beta_G,h_G}$ and $\mu_{F} = \mu_{F,\beta_F,h_F}$.
	From the assumptions that $\mu_{\Gs_\Gamma}(\Omega_{\rm good}) \ge 1- \delta$ and $\mu_{\Gs_\Gamma^*}(\Omega_{\rm good}) \ge 1- \delta$ we get
	\begin{align*}
	\TV{\mu_{\Gs_\Gamma}}{\mu_{\Gs_\Gamma}(\cdot\,|\,\Omega_{\rm good})} &= 1 - \mu_{\Gs_\Gamma}(\Omega_{\rm good}) \le \delta,~\text{and}\\
	\TV{\mu_{\Gs_\Gamma^*}}{\mu_{\Gs_\Gamma^*}(\cdot\,|\,\Omega_{\rm good})} &= 1 - \mu_{\Gs_\Gamma^*}(\Omega_{\rm good}) \le \delta.
	\end{align*}
	Also, from Lemma~\ref{lemma:model:eq} we have
	$
	\TV{\mu_{\Gs_\Gamma}(\cdot\,|\,\Omega_{\rm good})}{\mu_{\Gs_\Gamma^*}(\cdot\,|\,\Omega_{\rm good})} = \TV{\mu_G}{\mu_{G^*}}.
	$
	Therefore, it follows from the triangle inequality that
	\begin{align*}
	\TV{\mu_{\Gs_\Gamma}}{\mu_{\Gs_\Gamma^*}}
	\le{}&
	\TV{\mu_{\Gs_\Gamma}}{\mu_{\Gs_\Gamma}(\cdot\,|\,\Omega_{\rm good})} +
	\TV{\mu_G}{\mu_{G^*}} 
	+  \TV{\mu_{\Gs_\Gamma^*}}{\mu_{\Gs_\Gamma^*}(\cdot\,|\,\Omega_{\rm good})}\\
	\le{}& \TV{\mu_G}{\mu_{G^*}} + 2\delta.
	\end{align*}
	The lower bound is derived in similar fashion:
	\begin{align*}
		\TV{\mu_{\Gs_\Gamma}}{\mu_{\Gs_\Gamma^*}}
		\ge{}&
		\TV{\mu_G}{\mu_{G^*}} - \TV{\mu_{\Gs_\Gamma}}{\mu_{\Gs_\Gamma}(\cdot\,|\,\Omega_{\rm good})}
		 -  \TV{\mu_{\Gs_\Gamma^*}}{\mu_{\Gs_\Gamma^*}(\cdot\,|\,\Omega_{\rm good})}\\
		\ge{}& \TV{\mu_G}{\mu_{G^*}} - 2\delta,
	\end{align*}
	as claimed.
\end{proof}

We show next that if we have a sampling oracle for $\mu_{G}$, then we can generate approximate samples from $\mu_{\Gs_\Gamma}$ efficiently.

\begin{lemma}
	\label{lemma:main-reduction-sampling}
	%Let $N \ge 1$ be an integer and let $\beta > 0$.
	%Consider the Potts model on an $n$-vertex graph $G$.
	%Suppose that we are given a sampling oracle for the distribution $\mu_{G^*}$.
	%the Potts model
	%
	%with running time $\poly(N,\varepsilon^{-1})$
	%satisfying
	%$$
	%\TV{\mu_{M}}{\mu_{M}^{\textsc{alg}}} \le \varepsilon.
	%$$
	Consider the Potts model on an $n$-vertex graph $G$ and
	let $\Gamma = (b,p,\din,\dout)$ be such that  conditions \eqref{eq:cons-cond-1} and \eqref{eq:cons-cond-3} are satisfied.
	Suppose that $\mu_{\Gs_\Gamma}(\Omega_{\rm good}) \ge 1- \delta$ for some $\delta \in (0,1)$.
	Then, given a sampling oracle for the distribution $\mu_{G}$,
	there exists a sampling algorithm with running time $\poly(n,b)$
	such that
	the distribution $\mu_{\Gs_\Gamma}^{\normalfont\textsc{alg}}$ of its output satisfies:
	$$
	\TV{\mu_{\Gs_\Gamma}}{\mu_{\Gs_\Gamma}^{\normalfont\textsc{alg}}} \le \delta.
	$$
\end{lemma}

\begin{proof}
	The algorithm first draws a sample $\sigma_G$ from $\mu_{G}$ using the sampling oracle.
	It then constructs $\sigma \in \Omega_{\Gs_\Gamma}$ by assigning the spin $\sigma_{G}(v)$ to every vertex in the gadget corresponding to $v$ for each vertex $v$ of $G$. This can be done in $O(bn)$ time.
	From Lemma~\ref{lemma:model:eq} we see that the sampling algorithm in fact generates a sample from the distribution $\mu_{\Gs_\Gamma}(\cdot\,|\,\Omega_{\rm good})$, and so
	\[
	\TV{\mu_{\Gs_\Gamma}}{\mu_{\Gs_\Gamma}^{\normalfont\textsc{alg}}}
	= \TV{\mu_{\Gs_\Gamma}}{\mu_{\Gs_\Gamma}(\cdot\,|\,\Omega_{\rm good})} = 1 - \mu_{\Gs_\Gamma}(\Omega_{\rm good}) \le \delta. \qedhere
	\]
	%from the assumption that $\mu_{\Model(G)}(\Omega_{\rm good}) \ge 1- \delta$ and
	%the fact that $\mu_{M}^{\textsc{alg}}(\Omega \setminus \Omega_{\rm good})=0$
	%we get
	%\begin{align*}
	%\TV{\mu_{\Gs_\Gamma^*}}{\mu_{M}^{\textsc{alg}}} & \le \frac 12  \sum_{\sigma \in \Good} |\mu_{\Model(G)}(\sigma) - \mu_{\mu_{M}^{\textsc{alg}}}(\sigma)| + \frac \delta 2 \\
	%&= \frac 12  \sum_{\sigma \in \Good} |\mu_{M}(\sigma_G) \mu_{\Model(G)}(\Good) - \mu_{M}^{\textsc{alg}}(\sigma_G)| + \frac \delta 2.
	%\end{align*}
	%Since
	%$$
	%(1-\delta) \mu_{M}(\sigma_G) \le \mu_{M}(\sigma_G) \mu_{\Model(G)}(\Good)  \le \mu_{M}(\sigma_G),
	%$$
	%we deduce that
	%\begin{align*}
	%\TV{\mu_{\Model(G)}}{\mu_{M}^{\textsc{alg}}} &\le \frac 12  \sum_{\sigma \in \Good} |\mu_{M}(\sigma_G) - \mu_{M}^{\textsc{alg}}(\sigma_G)| + 2\delta\\
	%&= \TV{\mu_{M}}{\mu_{M}^{\textsc{alg}}}  + 2\delta \\
	%&\le \eps + 2\delta,
	%\end{align*}
\end{proof}

\subsection{Proof of Theorem~\ref{thm:deg-red}}
\label{subsec:proof-deg}

We are now ready to prove Theorem~\ref{thm:deg-red}.

\begin{proof}[Proof of Theorem~\ref{thm:deg-red}]
	We show that if there is an identity testing algorithm for $\MBiPottsMono(\hat{n},d,\hat\beta,\hat h)$
	with running time $T(\hat{n}) = \poly(\hat{n})$ and sample complexity $L(\hat{n}) = \poly(\hat{n})$,
	henceforth called the \textsc{Tester},
	then it can be used to solve the the identity testing problem for $\MPottsMono(n,n,\beta,h)$
	in $\poly(n)$ time; the parameters $n$, $\beta$ and $h$ depend on $\hat{n}$, $d$, $\hat\beta$ and $\hat h$ and will be specified next.

	Let us consider first the case when $3 \le d = O(1)$.
	In this case, we choose $n$ such that $\hat{n} = 2n^6$, $\beta = \hat\beta$ and $h = 2b\hat h$.
	Our identity testing algorithm for the family $\MPottsMono(n,n,\beta,h)$ constructs the graph $\Gs_\Gamma$ and the Potts model on $\Gs_\Gamma$ from Section~\ref{subsection:testing-instance-potts} using $\Gamma = (n^5,\lfloor n^{5/4}\rfloor,d-1,1)$ as the parameters for the random bipartite graph.
	%specifically, $h_{\Model(G)} (u,i) := h_G(v,i)/2b$ for each $u$ in the gadget $B_v$ and $i \in [q]$.
	This choice of parameters ensures that conditions \eqref{eq:cons-cond-1} and \eqref{eq:cons-cond-3} are satisfied. Note also $\Gs_\Gamma$ is bipartite by construction and that
	$|h_{\Gs_\Gamma} (u,i)| \le h/2b = O(\log n)$ for all $u\in V_{\Gs_\Gamma}$ and $i\in [q]$.
	
	Let $(G,\beta_G,h_G)$ be a Potts model from $\MPottsMono(n,n,\beta,h)$, and
	suppose that there is a hidden model $(G^*,\beta_{G^*},h_{G^*})$ from $\MPottsMono(n,n,\beta,h)$
	from which we are given samples.
	We want to use the \textsc{Tester} to distinguish with probability at least $3/4$ between the cases
	$\mu_G = \mu_{G^*}$ and
	$\TV{\mu_G}{\mu_{G^*}} > 1 - \varepsilon$.
	
	Suppose that $\sigma$ is sampled from $\mu_{\Gs_\Gamma}$.
	Since the field $h_G$ is $h$-vertex-monochromatic by assumption, it follows from our construction that for each gadget 
 	there exists $\kappa \in [q]$ such that for each vertex $v$ in the gadget $h_{\Gs}(v,j) = \hat h\cdot \1(j=\kappa)$.
	Hence, Theorem~\ref{thm:gadget:fact-ld} implies that with probability $1-o(1)$ over the choice of the random gadget $B$, if the configuration in the gadget $B_v$ for a vertex $v \in V_{\Gs}$
	is re-sampled, conditional on the configuration of $\sigma$ outside of $B_v$, then the new configuration in $B_v$ will be in a phase with probability
	at least
	$$
	\left(1 -    \frac{q^2{\e}^{2 \hat h}}{{\e}^{\delta' \beta_B d }} \right)^{2b} \ge 1 - \frac{2q^2b}{{\e}^{\delta \beta_B d}}
	$$
	for suitable constants $\delta,\delta' > 0$, since $\hat h = O(\log n)$ and $\beta_B d = \omega(\log n)$.
	A union bound then implies that after re-sampling the configuration in every gadget one by one, the resulting configuration $\sigma'$ is in the set $\Omega_{\rm good}$ with probability $1 - \frac{2q^2bn}{{\e}^{\delta \beta_B d}}$.
	Thus,
	\begin{align}
	\mu_{\Gs_\Gamma}(\Omega_{\rm good}) &\ge 1 - \frac{q^2 \hat{n}}{{\e}^{\delta \beta_B d}}. \label{eq:omega-good-G}
	\end{align}	
	We also consider the Potts model on $\Gs_\Gamma^*$, obtained from $G^*$ using the same random bipartite graph $B$.
	Note that we can not actually construct $\Gs_\Gamma^*$, since we only have sample access to $(G^*,\beta_{G^*},h_{G^*})$, but we can similarly deduce that
	\begin{align}
	\mu_{\Gs^*_\Gamma}(\Good) &\ge 1 - \frac{q^2\hat{n}}{{\e}^{\delta \beta_B d}}	\label{eq:omega-good-G*}.
	\end{align}
	
	Since we are given samples from $\mu_{G^*}$, \eqref{eq:omega-good-G*} and Lemma~\ref{lemma:main-reduction-sampling} imply that we can generate $L$ samples $\mathcal{S} = \{\sigma_1,\dots,\sigma_L\}$ from a distribution $\mu_{\Gs_\Gamma^*}^{\textsc{alg}}$ in $\poly(n)$ time such that
	\begin{equation}
	\label{eq:M*-alg}
	\TV{\mu_{\Gs_\Gamma^*}}{\mu_{\Gs_\Gamma^*}^{\textsc{alg}}} \leq \frac{q^2\hat{n}}{{\e}^{\delta \beta_B d}}.
	\end{equation}
	
	Our testing algorithm inputs the Potts model on $\Gs_\Gamma$ and the $L$ samples $\mathcal{S}$ to the \textsc{Tester} and outputs the  \textsc{Tester}'s output.	
	Recall that the \textsc{Tester} returns \textsc{Yes} if it regards the samples in $\mathcal{S}$ as samples from $\mu_{\Gs_\Gamma}$; it returns \textsc{No} if it regards them to be from some other distribution $\nu$ such that
	$\TV{\mu_{\Gs_\Gamma}}{\nu} > 1 - \varepsilon$.
	
	If $\mu_G = \mu_{G^*}$, then $\mu_{\Gs_\Gamma} = \mu_{\Gs_\Gamma^*}$.
	%\eqref{eq:omega-good-G} and Lemma~\ref{lemma:main-reduction-phase} implies that
	%$$
	%\TV{\mu_{\Model(G)}}{\mu_{\Model(G^*)}} \le  \frac{4q^2n}{{\e}^{\delta \beta_B d}}.
	%$$	
	Hence, \eqref{eq:M*-alg} implies that:
	$$
	\TV{\mu_{\Gs_\Gamma}}{\mu_{\Gs_\Gamma^*}^{\textsc{alg}}} = \TV{\mu_{\Gs^*_\Gamma}}{\mu_{\Gs_\Gamma^*}^{\textsc{alg}}} \leq   \frac{q^2\hat{n}}{{\e}^{\delta \beta_B d}}.
	$$
	Let $(\mu_{\Gs_\Gamma})^{\otimes L}$, $(\mu_{\Gs_\Gamma^*})^{\otimes L}$ and $(\mu_{\Gs_\Gamma^*}^{\textsc{alg}})^{\otimes L}$ be the product distributions corresponding to $L$ independent samples from $\mu_{\Gs_\Gamma}$, $\mu_{\Gs_\Gamma^*}$ and $\mu_{\Gs_\Gamma^*}^{\textsc{alg}}$ respectively.
	We have
	$$
	\TV{(\mu_{\Gs_\Gamma})^{\otimes L}}{(\mu_{\Gs_\Gamma^*}^{\textsc{alg}})^{\otimes L}} \leq L \TV{\mu_{\Gs_\Gamma}}{\mu_{\Gs_\Gamma^*}^{\textsc{alg}}} \leq   \frac{q^2\hat{n}L}{{\e}^{\delta \beta_B d}} = o_{\hat{n}}(1),
	$$
	since $L = \poly(\hat{n})$ and $\beta_B d = \hat \beta d = \omega(\log \hat{n})$.
%	Hence, if $\pi$ is the optimal coupling of the distributions $(\mu_{\Gs_\Gamma^*}^{\textsc{alg}})^{\otimes L}$ and $(\mu_{\Gs_\Gamma})^{\otimes L}$ (i.e., the marginal distribution is $\mathcal{S} \sim (\mu_{\Gs_\Gamma^*}^{\textsc{alg}})^{\otimes L}$ and $\mathcal{S}' \sim (\mu_{\Gs_\Gamma})^{\otimes L}$), and $(\mathcal{S},\mathcal{S}')$ is sampled from $\pi$, then $\pi(S\neq S' )  = o_{\hat{n}}(1)$. Therefore,
	Hence, using the optimal coupling of the distributions $(\mu_{\Gs_\Gamma^*}^{\textsc{alg}})^{\otimes L}$ and $(\mu_{\Gs_\Gamma})^{\otimes L}$ as in~\eqref{eq:coupling-main}, we obtain
	\begin{align}
	\label{eq:coupling}
	\Pr[ \textsc{Tester}~\text{outputs}~\textsc{No}~\text{given samples}~\mathcal{S}~\text{where}~\mathcal{S}\sim(\mu_{\Gs_\Gamma^*}^{\textsc{alg}})^{\otimes L}]
	\le
%	\notag\\
%	={}& \Pr[  \textsc{Tester}~\text{outputs}~\textsc{No}~\text{given samples}~\mathcal{S}~\text{where}~(\mathcal{S},\mathcal{S}')\sim\pi]\notag \\
%	\leq{}&	\Pr[ \textsc{Tester}~\text{outputs}~\textsc{No}~\text{given samples}~\mathcal{S}'~\text{where}~(\mathcal{S},\mathcal{S}')\sim\pi]+\pi(S\neq S' )\notag\\
%	={}& \Pr[\textsc{Tester}~\text{outputs}~\textsc{No}~\text{given samples}~\mathcal{S}'~\text{where}~\mathcal{S}'\sim(\mu_{\Gs_\Gamma})^{\otimes L}]+\pi(S\neq S' )\notag\\
%	\le{}&  
	\frac{1}{4} + o_{\hat{n}}(1) < \frac{1}{3}.\notag
	\end{align}	
	Hence, the \textsc{Tester} returns \textsc{Yes} with probability at least $2/3$ in this case.
	
	If $\TV{\mu_{G}}{\mu_{G^*}} \ge 1 - \varepsilon$,
	\eqref{eq:omega-good-G}, \eqref{eq:omega-good-G*} and
	%the second part of
	Lemma~\ref{lemma:main-reduction-phase} imply
	\begin{equation}
	\label{eq:error-lb}
	\TV{\mu_{\Gs_\Gamma}}{\mu_{\Gs_\Gamma^*}} \ge 1 - \varepsilon - \frac{2q^2\hat{n}}{{\e}^{\delta \beta_B d}} = 1 - \varepsilon - o_{\hat{n}}(1),
	\end{equation}
	because $\beta_B d = \hat \beta d = \omega(\log \hat{n})$.
	Moreover, from \eqref{eq:M*-alg} we get
	$$
	\TV{(\mu_{\Gs_\Gamma^*})^{\otimes L}}{(\mu_{\Gs_\Gamma^*}^{\textsc{alg}})^{\otimes L}} \leq L \TV{\mu_{\Gs_\Gamma^*}}{\mu_{\Gs_\Gamma^*}^{\textsc{alg}}} \leq   \frac{q^2\hat{n}L}{{\e}^{\delta \beta_B d}} = o_{\hat{n}}(1).
	$$
	Thus, analogously to~\eqref{eq:coupling-main} (i.e., using the optimal coupling between $(\mu_{\Gs_\Gamma^*}^{\textsc{alg}})^{\otimes L}$ and $(\mu_{\Gs_\Gamma^*})^{\otimes L}$), we get
	$$
	\Pr\left[  \textsc{Tester}~\text{outputs}~\textsc{Yes}~\text{given samples}~\mathcal{S}~\text{where}~\mathcal{S}\sim(\mu_{\Gs_\Gamma^*}^{\textsc{alg}})^{\otimes L}\right] \leq \frac{1}{3}.
	$$
	Hence, the \textsc{Tester} returns \textsc{No} with probability at least $2/3$.
	
	The case when $d$ is such that $d \le \hat{n}^{1-\rho}$ but $d = d(\hat{n}) \rightarrow \infty$ follows in similar fashion. In particular, we can take
	%\begin{align*}
	%N = \floor{n^{\rho/4}}-2 ,~~
	$b = \floor{{n^{4/\rho-1}}}$ and $\Gamma = \{b, b,\floor{\theta d},d-\floor{\theta d}\}$,
	%\end{align*}
	where $\theta = \theta(\rho)$ is a suitable constant.
	That is,
	$p = b$, $\din = \floor{\theta d}$, $\dout = d -  \floor{\theta d}$ and $\hat{n} = \Theta(n^{4/\rho})$.
	This choice parameters also satisfies conditions \eqref{eq:cons-cond-1} and \eqref{eq:cons-cond-3}.
	Hence,
	\eqref{eq:omega-good-G} and \eqref{eq:omega-good-G*} can be deduced similarly using Theorem~\ref{thm:gadget:fact-hd} instead.
	The rest of the proof remains unchanged for this case.
	%Note that for this choice of parameters, $|V| = N =  \lfloor n^{\rho/4} \rfloor-2 \ge \lfloor n^{\min\{\frac \rho 4,\frac{1}{14}\}} \rfloor - 2$.
\end{proof}

\section{Hardness of the decision version of approximate counting}
\label{app:app-cnt}

In this section we give a general reduction from the approximate counting problem to the decision version of the problem. In particular, we prove Theorem~\ref{thn:decision-cnt:potts} from Section~\ref{Potts-decision}.
We state our results for the models of interest in this paper, but they extend straightforwardly to other spin systems. 
%Examples of graphical models include RBMs, Ising models, Potts models, hardcore models, colorings, etc.

We restate first the definition of the decision version of $r$-approximate counting. 

\begin{center}
	\setlength{\fboxsep}{5pt}
	\noindent\fbox{
		\parbox{0.88\textwidth}{
			%\vspace{-5pt}
			%\begin{center}
			%	\textit{Decision version of $r$-approximate counting for the  ferromagnetic Potts Model}
			%\end{center}
			\defDecisionCnt*
		}
	}
\end{center}

Recall also that a \emph{fully polynomial-time randomized approximation scheme ($\fpras$)} for an optimization problem with solutions $\mathrm{OPT}$ is a randomized algorithm that for any $\rho > 0$ outputs a solution $\hat{Z}$ satisfying $e^{-\rho}\, \mathrm{OPT} \le \hat{Z} \le e^\rho\, \mathrm{OPT}$ with probability at least $3/4$ and has running time $\poly(n,1/\rho)$ where $n$ is the size of the input. 
To prove Theorem~\ref{thn:decision-cnt:potts}, we introduce an intermediate problem referred as $r$-approximate counting. 

\begin{center}
	\setlength{\fboxsep}{5pt}
	\noindent\fbox{
		\parbox{0.88\textwidth}{
			\begin{defn}[$r$-approximate counting]
				\label{def:cnt}
			%\emph{$r$-approximate counting:} 
			Given a Potts model ($G$,$\beta_G$,$h$) 
			and an approximation ratio $r>1$, output a real number $\hat{Z}$ satisfying the following with probability at least $3/4$:
			\[
			\frac{1}{r}\, Z_{G,\beta_G,h} < \hat{Z} < r\, Z_{G,\beta_G,h}.
			\]
			\end{defn}
		}
	}
\end{center}

%\begin{center}
%\setlength{\fboxsep}{5pt}
%\noindent\fbox{
%	\parbox{0.9\textwidth}{
%	\vspace{-5pt}
%	\begin{center}
%	\textit{Decision Version of $r$-Approximate Counting}
%	\end{center}
%	
%	We are given a spin system on a graph $G$ whose partition function is denoted by $Z_G$.
%	For $r>1$ and an input $\hat{Z}$, distinguish between the following two cases
%	\[
%	\text{(i)~} Z_G \le \frac{1}{r} \hat{Z}
%	\quad\quad \text{(ii)~}
%	Z_G \ge r \hat{Z}
%	\]
%	with success probability at least $5/8$.
%	\vspace{5pt}
%	}
%}
%\end{center}

%\begin{center}
%\setlength{\fboxsep}{5pt}
%\noindent\fbox{
%	\parbox{0.9\textwidth}{
%	\vspace{-5pt}
%	\begin{center}
%	\textit{$r$-Approximate Counting}
%	\end{center}
%	
%	We are given a spin system on a graph $G$ whose partition function is denoted by $Z_G$.
%	For $r>1$, output a real number $\hat{Z}$ such that it satisfies
%	\[
%	\frac{1}{r}\, Z_G < \hat{Z} < r\, Z_G
%	\]
%	with probability at least $3/4$.
%	\vspace{5pt}
%	}
%}
%\end{center}

%\begin{thm}%[\textcolor{blue}{Blue arrow}]
%\label{thm:decision-counting}
%Let $N,d \ge 1$ be integers and $\beta,h\ge0$ be reals.
%Suppose that there is no $\fpras$ for the approximate counting problem for $\MPotts(N,d,\beta,h)$. Then for any $c>0$ there is no polynomial-time algorithm for the decision version of $N^c$-approximate counting for the family $\MPotts(N,d,\beta,h)$.

%The statement remains true if we replace $\MPotts$ by $\MfPotts$ or $\MaPotts$ (in particular, true for $\MfIsing$ and $\MaIsing$).
%\end{thm}

Notice that an $\fpras$ for the counting problem is equivalent to an algorithm for the $e^\rho$-approximate counting problem with running time $\poly(n,1/\rho)$ for all $\rho>0$. 
We first show the equivalence of $r$-approximate counting and its decision version.
\begin{lemma}
\label{lem:binary-search}
%Consider a class of spin systems. 
%Suppose that there exist constants $c_1,c_2>0$ such that, for every $N$-vertex graph $G$ and parameters that specify the spin system, we have
%$$ \exp\left(-c_1 N^{c_2}\right)\le Z_G \le \exp\left(c_1 N^{c_2}\right). $$
Let $n,d \ge 1$ be integers and let $\beta,h\ge0$ be real numbers. 
Assume that $r=r(n)>1$ is the approximation ratio. 
Then, given a polynomial-time algorithm for the decision version of $r$-approximate counting for a family of Potts models $\mathcal M$, where 
$$
\mathcal M \in \{\hatMfPotts(n,d,\beta,h),\hatMaIsing(n,d,\beta,h),\hatMfIsing(n,d,{\beta},{h})\}, 
$$
there is also a polynomial-time algorithm for $2r$-approximate counting for $\mathcal{M}$. 
\end{lemma}

\begin{proof}
Consider a Potts model from $\mathcal{M}$ with the underlying graph $G$. 
We note first that using a standard argument we can boost the success probability of the algorithm for the decision version of $r$-approximate counting in polynomial time.
More precisely, for a given $\hat{Z} > 0$ we run the algorithm for 
$$k = 80\ceil{\log(8\log(4c_1n^2+4\log r))}+1$$
 times and output the majority answer.
Let $X_i$ be the indicator random variable of the event that the $i$-th answer is correct and let $X = \sum_{i=1}^k X_i$.
Then by our assumption we have $\Exp [X] \ge \frac{5}{8} k$.
The Chernoff bound then implies that the majority answer is incorrect with probability at most
\[
\Pr\left(X \le \frac{k}{2}\right) \le \Pr\left(X \le \frac{4}{5}\Exp [X]\right) \le \exp\left( -\frac{\Exp [X]}{50} \right) \le \exp\left( -\frac{k}{80} \right) \le \frac{1}{8\log(4c_1n^2+4\log r)}.
\]

Using the boosted version of the decision $r$-approximate counting algorithm, henceforth call \textsc{BoostedDecider}, we use binary search procedure to give an $r$-approximate counting algorithm.
First note that there exists a constant $c_1 := c_1(q,\beta,h)>0$ such that
$$ \exp\left(-c_1 n^2 \right)\le Z_G \le \exp\left(c_1 n^2 \right). $$

Then, let $\ell_0= \frac{1}{r}\exp(-c_1 n^2)$ and $u_0 = r\exp(c_1 n^2)$.
For $i\ge 1$, let $c_i = \sqrt{\ell_{i-1}u_{i-1}}$ and run the testing algorithm with $\hat{Z} = c_i$.
If \textsc{BoostedDecider} outputs $Z_G \le \frac{1}{r} \hat{Z}$ then we let $(\ell_i,u_i) = (\ell_{i-1},c_i)$, and if \textsc{BoostedDecider} outputs $Z_G \ge r \hat{Z}$ then we let $(\ell_i,u_i) = (c_i, u_{i-1})$.
We repeat this process until $u_i/\ell_i \le 2$, and then output $\hat{Z} = \ell_i$.
Observe that $\log u_i - \log\ell_i$ decreases by a factor $2$ in each iteration.
Thus, the number of times that outputs is called is at most
\[
\log_2\left( \frac{\log u_0 - \log\ell_0}{\log 2} \right) = \log_2 \left( \frac{2c_1 n^2 + 2\log r}{\log 2} \right) \le 2\log(4c_1n^2+4\log r).
\]
Assume that \textsc{BoostedDecider} never makes a mistake in all these calls; this happens with probability at least $3/4$ by a union bound.
Then, for each $j\ge 0$, the algorithm outputs $Z_G \le \frac{1}{r} \hat{Z}$ for $\hat{Z} = u_j$ and $Z_G \ge r \hat{Z}$ for $\hat{Z} = \ell_j$. This implies that
$$ \frac{1}{r}\ell_j< Z_G < ru_j $$
for all $j\ge 0$.
Hence, the final output satisfies 
$$ \frac{1}{r}\ell_i < Z_G < ru_i \le 2r \ell_i $$
with probability at least $3/4$. 
The running time of the algorithm is polynomial in $n$, assuming that $r\le \exp(c_1 n^2)$. If we have $r> \exp(c_1 n^2)$ instead, then the algorithm can just output $1$, which is already a $r$-approximation of $Z_G$. 
\end{proof}

%\noindent
We show
next that a polynomial-time $n^c$-approximate counting algorithm for a family of Potts models on $n$-vertex graphs can be turned into an $\fpras$. 

\begin{lemma}
\label{lem:poly-approx}
Let $n,d \ge 1$ be integers and let $\beta,h\ge0$ be real numbers. 
For any $c>0$, given a polynomial-time $n^c$-approximate counting algorithm for a family of Potts models $\mathcal M$, where 
$$
\mathcal M \in \{\hatMfPotts(n,d,\beta,h),\hatMaIsing(n,d,\beta,h),\hatMfIsing(n,d,{\beta},{h})\}, 
$$
there is an $\fpras$ for the counting problem for $\mathcal{M}$. 
\end{lemma}

\begin{proof}
Suppose that there is a polynomial-time $n^c$-approximate counting algorithm for $\mathcal{M}$ where $c>0$ is a constant. 
Consider a Potts model from $\mathcal{M}$ defined on a graph $G$ of $n$ vertices. 
We will give an $\fpras$ for its partition function. 
For an arbitrary $\rho>0$, let $k$ be the smallest integer such that $k\ge (c\log (kn))/\rho$. Notice that $k \le \poly(\log n, 1/\rho)$. 
Define a Potts model that is a disjoint union of $k$ copies of the Potts model on $G$.
That is, the underlying graph $G'$ consists of $k$ copies of $G$, and the weights for each copy are the same as the original model. 
It follows immediately that $Z_{G'} = (Z_G)^k$.
We run the $(kn)^c$-approximate counting algorithm for the Potts model on $G'$ and assume the output is $\hat{Z}$. Then with probability at least $3/4$ we have
\[
(kn)^{-c} Z_{G'} < \hat{Z} < (kn)^c Z_{G'}.
\]
Assuming this holds, then we get
\[
e^{-\rho} Z_G \le (kn)^{-c/k} Z_G <  \hat{Z}^{1/k} < (kn)^{c/k} Z_G \le e^\rho Z_G
\]
Thus, $\hat{Z}^{1/k}$ is a $e^\rho$-approximation of $Z_G$ with probability at least $3/4$ and can be computed in $\poly(kn) = \poly(n,1/\rho)$ time.
\end{proof}

\begin{proof}[Proof of Theorem~\ref{thn:decision-cnt:potts}]
Follows immediately from Lemmas~\ref{lem:binary-search} and \ref{lem:poly-approx}. 
\end{proof}

\bibliographystyle{plain}
\bibliography{testing}

\appendix

\section{The ferromagnetic mean-field Potts model: proofs}
\label{app:mean-field-potts}

In this appendix we prove our detailed results concerning the phase transitions of the ferromagnetic mean-field Potts models (i.e., Lemmas~\ref{lem:binary-betaH-meta} and~\ref{lem:binary-betaH-ratio}). As mentioned, several variants of these results have appeared before, e.g.,~\cite{BGJ,LL,GJ,CDLLPS,GLP,GSVmf,BSmf}, but we need slightly more precise results.

\begin{proof}[Proof of Lemma~\ref{lem:binary-betaH-meta}]
	
	Let us introduce some convenient notation first. 
	For an integer $m \ge 1$, let
	\begin{align*}
	\hat A &= \left\{ (\alpha_1, \dots, \alpha_q) \in \R^q: \alpha_i \ge 0,~ \sum_{i=1}^q \alpha_i = 1,~ \alpha_i m \in \N \right\},\\
	\hat{{D}} &= \mathrm{Ball}_\infty (u, m^{-1/4}) = \left\{ \alpha\in \hat A: \norm{\alpha - u}_\infty \le m^{-1/4} \right\},\\
	\hat{{M}} &= \bigcup_{i=1}^q \mathrm{Ball}_\infty (\alpha^{*,i}, m^{-1/4})
	= \bigcup_{i=1}^q \left\{ \alpha\in \hat A: \norm{\alpha - \alpha^{*,i}}_\infty \le m^{-1/4} \right\},
	\end{align*}
	and $\hat{{S}} = \hat A \backslash (\hat{{D}} \cup \hat{{M}})$. Setting $\hat{\beta}_H = \beta_H \cdot m$, we have
	
	\begin{align}
	\label{eq:mf:z:dis}
	Z_H^{\DD}({\beta}_H) =
	\sum_{\alpha \in \hat{{D}}} \binom{m}{\alpha_1 m ~\cdots~ \alpha_q m} \exp\left( \frac{\hat{\beta}_H}{m} \sum_{i=1}^q \binom{\alpha_i m}{2} \right),
	\end{align}	
	and similarly for $Z_H^{\MM}({\beta}_H)$ and $Z_H^{\SS}({\beta}_H)$ with the summation over $\hat M$ and $\hat S$ respectively.
	%	Then we have
	%	\[
	%	Z_H^\mathrm{D}(\beta_H) = 
	%	Z_H^{\hat{\mathrm{D}}}(\hat{\beta}_H) = \sum_{\alpha \in \hat{\mathrm{D}}} \binom{m}{\alpha_1 m ~\cdots~ \alpha_q m} \exp\left( \frac{\hat{\beta}_H}{m} \sum_{i=1}^q \binom{\alpha_i m}{2} \right),
	%	\]
	%	and similarly $Z_H^\mathrm{M}(\beta_H) = Z_H^{\hat{\mathrm{M}}}(\hat{\beta}_H)$, $Z_H^\mathrm{S}(\beta_H) = Z_H^{\hat{\mathrm{S}}}(\hat{\beta}_H)$.
	
	%	Then
	%	\[
	%	|\hat A| = \binom{m+q-1}{q-1} \le (m+q-1)^{q-1}.
	%	\]
	Using standard bounds for the multinomial coefficient (see, e.g., \cite[Lemma 2.2]{CS04}), we have 
	for every $\alpha\in \hat A$ 
	\begin{equation}
	\label{eq:mf:mcoef}
	\frac{1}{|\hat A|} e^{H(\alpha) m} \le \binom{m}{\alpha_1 m ~\cdots~ \alpha_q m} \le e^{H(\alpha) m},
	\end{equation}
	where $	H(\alpha) = \sum_{i=1}^q - \alpha_i \ln \alpha_i$. Hence, for $\beta \in \R$ and $\alpha \in \R^q$, we introduce:
	\[
	\Phi_{{\beta}}(\alpha) = H(\alpha) + \frac{{\beta}}{2} \norm{\alpha}_2^2.
	\]
	The function $\Phi_{{\beta}}$ have the following properties, which we prove later and will be useful throughout the proof. 
	\begin{fact}\label{fact:phi}
		\begin{enumerate}[(i)]
			\item For $\alpha \in \hat A$ and $\beta_1,\beta_2>0$,  we have
			$
			|\Phi_{\beta_1}(\alpha) - \Phi_{\beta_2}(\alpha)| \le \frac{1}{2} |\beta_1 - \beta_2|.
			$
			\item When $\hat{\beta}_H = \Bo$, the function $\Phi_{\Bo}$ has exactly $q+1$ global maxima in $\hat A$ consisting of
			one disordered phase
			$
			u = (1/q,\dots,1/q)
			$
			and $q$ majority phases 
			$\alpha^{*,i}$ with $ i \in [q]$,
			where the $i$-th coordinate of $\alpha^{*,i}$ is strictly larger than $1/q$.
			%$\frac{q-2}{q-1}$
			%and every other coordinate is $\frac{1}{(q-1)^2}$.
			\item 	There exist constants $\eps,c>0$ such that $\Phi_{\Bo}(\alpha)$ is $c$-strongly concave in the balls $\mathrm{Ball}_\infty (u, \eps)$ and $\mathrm{Ball}_\infty (\alpha^{*,i}, \eps)$ for $i\in [q]$. That is, for all $\alpha\in \hat A$ such that $\norm{\alpha - u}_\infty \le \eps$ or $\norm{\alpha - \alpha^{*,i}}_\infty \le \eps$ for some $i \in [q]$, we have
			$\hessian \Phi_{\Bo}(\alpha) \preceq -c  \cdot I$, where $I$ is the $q \times q$ identity matrix.
		\end{enumerate}
	\end{fact}
	Hence, \eqref{eq:mf:mcoef} and part (ii) of this fact imply
	\begin{align}
	\label{eq:mf:dis}
	Z_H^{\DD}({\beta}_H) &\ge \sum_{\alpha \in \hat{{D}}} \frac{e^{-\hat{\beta}_H/2}}{|\hat A|} \exp\left[ \Phi_{\hat{\beta}_H}(\alpha) m \right] \notag\\
	&\ge \frac{e^{-\hat{\beta}_H/2}}{|\hat A|} \exp\left(-\frac{1}{2}|\hat{\beta}_H-\Bo|m\right) \sum_{\alpha \in \hat{{D}}} \exp\left[ \Phi_{\Bo}(\alpha) m \right] \notag\\
	&\ge \frac{e^{-\hat{\beta}_H/2}}{|\hat A|} \exp\left(-\frac{c'}{2}\sqrt{m}\right) \exp\left[ \Phi_{\Bo}(u) m \right].
	\end{align}
	Similarly, we deduce that
	\begin{align}
	\label{eq:mf:maj}
	Z_H^{\MM}({\beta}_H) &\ge \frac{e^{-\hat{\beta}_H/2}}{|\hat A|} \exp\left(-\frac{c'}{2}\sqrt{m}\right) \sum_{i=1}^q \exp\left[ \Phi_{\Bo}(\alpha^{*,i}) m \right] \notag\\
	&= \frac{q e^{-\hat{\beta}_H/2}}{|\hat A|} \exp\left(-\frac{c'}{2}\sqrt{m}\right) \exp\left[ \Phi_{\Bo}(u) m \right],
	\end{align}
	and 
	\begin{align}
	\label{eq:mf:comp}
	Z_H^{\SS}({\beta}_H) 
	%&=
	%\sum_{\alpha \in \hat{{S}}} \binom{m}{\alpha_1 m ~\cdots~ \alpha_q m} \exp\left( \frac{\hat{\beta}_H}{m} \sum_{i=1}^q \binom{\alpha_i m}{2} \right)\\
	%&\le \sum_{\alpha \in \hat{{S}}} e^{-\hat{\beta}_H/2} \exp\left[ \Phi_{\hat{\beta}_H}(\alpha) m \right]\\
	&\le e^{-\hat{\beta}_H/2} \exp\left(\frac{1}{2}|\hat{\beta}_H-\Bo|m\right) \sum_{\alpha \in \hat{{S}}} \exp\left[ \Phi_{\Bo}(\alpha) m \right] \notag\\
	&\le e^{-\hat{\beta}_H/2} |\hat A| \exp\left(\frac{c'}{2}\sqrt{m}\right) \exp\left[ \left(\max_{\alpha \in \hat{{S}}} \Phi_{\Bo}(\alpha)\right) m \right].
	\end{align}
	Let $\hat{{S}} = \hat{{S}}_1 \cup \hat{{S}}_2$ where
	\[
	\hat{{S}}_1
	%= S \backslash \left( \mathrm{Ball}_\infty (u, \eps) \cup \bigcup_{i=1}^q \mathrm{Ball}_\infty (\alpha^{*,i}, \eps) \right)
	= \hat A \backslash \left( \mathrm{Ball}_\infty (u, \eps) \cup \bigcup_{i=1}^q \mathrm{Ball}_\infty (\alpha^{*,i}, \eps) \right)
	\]
	and
	\[
	\hat{{S}}_2 = \left( \mathrm{Ball}_\infty (u, \eps) \cup \bigcup_{i=1}^q \mathrm{Ball}_\infty (\alpha^{*,i}, \eps) \right) \backslash (\hat{{D}} \cup \hat{{M}}).
	\]
	Since the function $\Phi_{\Bo}$ is continuous, and $u, \alpha^{*,1},\dots, \alpha^{*,q}$ are its only global maxima, for constant $\eps > 0$ there exists constant $\delta =\delta(\varepsilon) > 0$ such that for all $\alpha \in \hat{{S}}_1$ we have
	\[
	\Phi_{\Bo}(\alpha) \le \Phi_{\Bo}(u) - \delta.
	\]
	By part (iii) of Fact~\ref{fact:phi}, $\Phi_{\Bo}(\alpha)$ is $c$-strongly concave in $\hat{{S}}_2$; thus, for all $\alpha \in \mathrm{Ball}_\infty (u, \eps) \backslash \hat{{D}}$ we have
	\begin{align*}
	\Phi_{\Bo}(\alpha) &\le \Phi_{\Bo}(u) + \grad \Phi_{\Bo}(u) (\alpha-u) - c\norm{\alpha - u}^2\\
	&= \Phi_{\Bo}(u) - c\norm{\alpha - u}^2\\
	&\le \Phi_{\Bo}(u) - c m^{-1/2},
	\end{align*}
	and similarly for all $\alpha \in \mathrm{Ball}_\infty (\alpha^{*,i}, \eps) \backslash \hat{{M}}$ we have
	\begin{align*}
	\Phi_{\Bo}(\alpha) &\le \Phi_{\Bo}(\alpha^{*,i}) + \grad \Phi_{\Bo}(\alpha^{*,i}) (\alpha-u) - c\norm{\alpha - \alpha^{*,i}}^2\\
	&= \Phi_{\Bo}(u) - c\norm{\alpha - \alpha^{*,i}}^2\\
	&\le \Phi_{\Bo}(u) - c m^{-1/2}.
	\end{align*}
	Therefore,
	\[
	\max_{\alpha \in \hat{{S}}} \Phi_{\Bo}(\alpha) \le \Phi_{\Bo}(u) - c m^{-1/2}.
	\]
	Plugging this bound into \eqref{eq:mf:comp} and combining it with \eqref{eq:mf:dis}, we get
	\begin{align*}
	Z_H^{\SS}({\beta}_H) &\le e^{-\hat{\beta}_H/2} |\hat A| \exp\left(\frac{c'}{2}\sqrt{m}\right) \exp\left( -c\sqrt{m} \right) \exp\left[ \Phi_{\Bo}(u) m \right]\\
	&\le |\hat A|^2 \exp\left( -(c-c')\sqrt{m} \right) Z_H^{\DD}({\beta}_H).
	\end{align*}
	Combining with \eqref{eq:mf:maj} instead we obtain
	\begin{align*}
	Z_H^{\SS}({\beta}_H) &\le \frac{|\hat A|^2}{q} \exp\left( -(c-c')\sqrt{m} \right) Z_H^{\MM}({\beta}_H).
	\end{align*}
	The results then follows by picking $c' = c/2$.
\end{proof}

We wrap up the proof of Lemma~\ref{lem:binary-betaH-meta} by establishing the facts used of the function in $\Phi_{\Bo}$.

\begin{proof}[Proof of Fact~\ref{fact:phi}]
	Part (i) follows from the definition of the function $\Phi_{\beta}$, since when $\alpha \in \hat A$, $\|\alpha\|_1=1$, and so $\|\alpha\|_2\leq 1$.
	
	For part (ii),
	suppose $\alpha = (\alpha_1,\dots,\alpha_q)$ is a local maxima for $\Phi_{\Bo}(\alpha)$.
	Using the method of Lagrange multipliers, we obtain that $\alpha$ must satisfy:
	$$
	\Bo \alpha_i - \log(\alpha_i) = 1-\lambda,\ i\in [q].
	$$
	The function $\Bo x - \log x$ is decreasing for $x<1/\Bo$ and increasing for $x>1/\Bo$.
	This implies that for any $\lambda$ there are at most 2 solutions to $\Bo x - \log x = 1 - \lambda$
	and hence there are at most two different values of $\alpha_i$. If there is only one value of $\alpha_i$ then
	$\alpha_i=1/q$ for $i\in [q]$. If there are two values of $\alpha_i$ then one of them is in $(0,1/\Bo)$ and
	one of them is in $(1/\Bo,1)$. 
	
	Now the Hessian of $\Phi_{\Bo}$ is
	\begin{equation}\label{eee}
	\hessian \Phi_{\Bo}(\alpha) = - \diag(\alpha_1^{-1},\dots, \alpha_q^{-1}) + \Bo I,
	\end{equation}
	and since $\alpha$ is a maxima for $\Phi_{\Bo}$, then $\hessian \Phi_{\Bo}(\alpha)$ must be negative definite in the subspace of vectors perpendicular to $1$ (since
	the sum of $\alpha_i$ is constrained to be $1$ the perturbations must maintain this constraint).
	If there were at least two indexes (w.l.o.g., make the indexes $1$ and $2$) such that
	$\alpha_1=\alpha_2 > 1/\Bo$ then the Hessian is not negative definite in the subspace of
	vectors perpendicular to $1$ (e.g., take the vector $x=(1,-1,0,\dots,0)$; then $x^T \hessian \Phi_{\Bo}(\alpha) x = 2(\Bo - 1/\alpha_1) > 0$).
	Thus a (constrained) maxima $\alpha$ of $\Phi_{\Bo}$ will either have all $\alpha_i$ equal to $1/q$, or exactly $q-1$ of the $\alpha_i$'s will be the same.
	
	Hence, the maxima of $\Phi_{\Bo}$ 
	will coincide with those of a one-dimensional version of it , denoted by $\Psi_1$, previously studied in~\cite{GSVmf}.
	The function $\Psi_1: [0,1] \rightarrow \R$ is define as  $\Psi_1(x) = \Phi_{\Bo}(x,y,\dots,y)$, where $y = \frac{1-x}{q-1}$.
	The function $\Psi_1$ has $2$ global maxima
	(see Lemma 2 in~\cite{GSVmf}) and hence $\Phi_{\Bo}$ has exactly
	$q+1$ global maxima (one of the maxima of $\Psi_1$ corresponds to $q$ maxima of $\Phi_{\Bo}$).
	Finally, observe that $\Bo < q$, and so the coordinate of the maxima of $\Phi_{\Bo}$ in $(1/\Bo,1)$ is greater than $1/q$.
	
	For part (iii), note that the Hessian in equation~\eqref{eee} is continuous around $(\alpha_1,\dots,\alpha_q)$
	and hence it is negative definite in a sufficiently small ball around $u$ and $\alpha^{*,i}$.
\end{proof}

We will provide next the proof of Lemma~\ref{lem:binary-betaH-ratio}, in which we will use the following bound on the ratio $\frac{Z_H^{\MM}(\Bo)}{Z_H^{\DD}(\Bo)}$, which is derived 
similarly to Lemma~\ref{lem:binary-betaH-meta}.

\begin{fact}\label{fact:mean} $
	\frac{1}{q |\hat A|} \leq \frac{Z_H^{\MM}(\Bo/m)}{Z_H^{\DD}(\Bo/m)}\leq q |\hat A|^2.
	$
\end{fact}

\begin{proof}
	From \eqref{eq:mf:z:dis} and \eqref{eq:mf:mcoef}, we obtain
	\begin{align*}
	Z_H^{\DD}(\Bo/m) &=
	\sum_{\alpha \in \hat{\mathrm{D}}} \binom{m}{\alpha_1 m ~\cdots~ \alpha_q m} \exp\left( \frac{\Bo}{m} \sum_{i=1}^q \binom{\alpha_i m}{2} \right)\\
	&\le \sum_{\alpha \in \hat{\mathrm{D}}} e^{-\Bo/2} \exp\left[ \Phi_{\Bo}(\alpha) m \right]\\
	&\le |\hat A| e^{-\Bo/2} \exp\left[ \Phi_{\Bo}(u) m \right].
	\end{align*}
	Similarly we have
	\begin{align*}
	Z_H^{\MM}(\Bo/m) &\le |\hat A| e^{-\Bo/2} \sum_{i=1}^q \exp\left[ \Phi_{\Bo}(\alpha^{*,i}) m \right]\\
	&\le q |\hat A| e^{-\Bo/2} \exp\left[ \Phi_{\Bo}(u) m \right] .
	\end{align*}
	Combining our upper and lower bounds on $Z_H^{\DD}(\Bo/m)$ and $Z_H^{\DD}(\Bo/m)$ we obtain the result.
\end{proof}

We are now ready to proof Lemma~~\ref{lem:binary-betaH-ratio}.

\begin{proof}[Proof of Lemma~\ref{lem:binary-betaH-ratio}]
	For ease of notation let $f(\beta) = \frac{Z_H^{\MM}(\beta)}{Z_H^\DD(\beta)}$.
	We show that for suitable constants $c,c' > 0$, for $\beta_L = \mathfrak{B}_o/m - c' m^{-3/2}$ we have
	\begin{equation}\label{lb}
	f(\beta_L) \leq \exp(-c\sqrt{m}),
	\end{equation}
	and for $\beta_U = \mathfrak{B}_o/m + c' m^{-3/2}$ we have
	\begin{equation}\label{ub}
	f(\beta_U)\geq \exp( c\sqrt{m}).
	\end{equation}
	Since $|\hat M|=O(m^q)$ and $|\hat D|= O(m^q)$, we can compute 
	$Z_H^{\MM}(\beta)$ and $Z_H^{\DD}(\beta)$ for any $\beta \in [\beta_L,\beta_U]$
	in $\poly(m)$ time by enumerating over elements of $\hat M$ and $\hat D$, respectively.
	(Note that this involves computing multinomial coefficients, which can be done for example by expressing them as product of $q$ binomial coefficients; see~\eqref{eq:mf:z:dis}.) 
	Then, given~\eqref{lb}
	and~\eqref{ub}, for any
	$R\in [\exp(-c\sqrt{m}),\exp(c\sqrt{m})]$  and small enough $\xi >0$,
	we can use the bisection method
	with $[\beta_L,\beta_U]$ as the starting interval to find a $\beta \in [\beta_L,\beta_U]$ such that
	$$
	f(\beta) \le R \le f(\beta + \xi) \le f(x) + \xi \cdot \max_{\beta_0 \in [\beta_L,\beta_U]} f'(\beta_0)
	$$
	in time polynomial in $m$ and $\log \xi^{-1}$. 
	Since $f'(\beta_0) = \exp(O(m))$ for $\beta_0 \in [\beta_L,\beta_U]$, we can choose 
	$\xi = \exp(-\Theta(m))$ so that $f(\beta) \le R \le f(\beta) + \delta R$ as desired. 
	
	To establish~\eqref{lb}
	and~\eqref{ub} we consider the function
	$$
	g(\beta) = \log Z_H^\MM(\beta) - \log Z_H^\DD(\beta).
	$$
	Note that
	\begin{equation}\label{rho}
	\frac{\partial}{\partial \beta } g(\beta) = \frac{\frac{\partial}{\partial \beta }Z_H^\mathrm{M}(\beta)}{Z_H^\mathrm{M}(\beta)} - \frac{\frac{\partial}{\partial \beta }Z_H^\mathrm{D}(\beta)}{Z_H^\mathrm{D}(\beta)}.
	\end{equation}
	By a direct (and standard) calculation, we can check that
	the first term in the right-hand-side expression in~\eqref{rho} corresponds to the expected number of monochromatic edges in a random configuration $\sigma$ of
	the model conditioned on $\sigma$ being in the set $M$. Therefore,
	\begin{equation}\label{hjk1}
	\frac{\frac{\partial}{\partial \beta }Z_H^\MM(\beta)}{Z_H^\MM(\beta)} \geq \binom{\hat \alpha m - m^{3/4}}{2} + (q-1)\binom{\frac{(1-\hat\alpha)m}{q-1} - m^{3/4}}{2},
	\end{equation}
	where $\hat \alpha$ is the constant in the definition of the set $M$.
	Similarly, the second term in the right-hand-side of~\eqref{rho} is the expected number of  monochromatic edges in a random configuration $\sigma$ of
	the model conditioned on $\sigma$ being in the set $D$ and so
	\begin{equation}\label{hjk2}
	\frac{\frac{\partial}{\partial \beta }Z_H^\mathrm{D}(\beta)}{Z_H^\mathrm{D}(\beta)} \leq q \binom{m/q + m^{3/4}}{2}.
	\end{equation}
	Combining~\eqref{hjk1} and~\eqref{hjk2} 
	and using the fact that $\hat \alpha>1/q$,
	we obtain for a suitable constant $\rho >0$ and sufficiently large $m$ that for any $\beta \in [\beta_L,\beta_U]$
	\begin{equation}\label{hjk3}
	\frac{\partial}{\partial \beta } g(\beta) \geq \rho m^2.
	\end{equation}
%	Similarly, we can derive a $O(m^2)$ upper bound for $\frac{\partial}{\partial \beta } g(\beta)$.
	
	Since $|\hat A| = \Theta(m^q)$, Fact~\ref{fact:mean} implies that $|g(\Bo/m)| = \Theta(\log m)$. Hence,
	by the mean value theorem 
	$$
	g(\beta_L) \le g(\Bo/m) - \rho m^2 |\Bo/m - \beta_L| \le -c\sqrt{m}
	$$
	and similarly $g(\beta_U) \ge c\sqrt{m} $ for a suitable constant $c > 0$. 	Since 
	$g = \log f$, \eqref{lb}
	and~\eqref{ub} follow and the proof is complete.
	%	Note that the value of $g$ changes from
	%	a value close to zero at $\beta=\mathfrak{B}_o/m$ by $-c'''\sqrt{m}$ (for some $c'''>0)$ when $\beta$ is decreased by $c'/m^{3/2}$
	%	and similarly it changes by $c'''\sqrt{m}$ when $\beta$ is increased by $c'/m^{3/2}$. Thus at $\beta=\beta_L$ the value is
	%	less than $-c''\sqrt{m}$ and at $\beta=\beta_U$ the value is more than $c''\sqrt{m}$ (for some value of $c''>0$).
\end{proof}

\end{document}